%% file: main.tex
\documentclass{article}
\usepackage{geometry}
\input{macros.tex}
\usepackage{authblk}
\usepackage{subcaption}
\begin{document}
\title{Formulations and Constructions of Remote State Preparation with Verifiability, with Applications}
	\author[1]{Jiayu Zhang\footnote{zhangjy@zgclab.edu.cn}}
	\affil[1]{Zhongguancun Laboratory}
	\maketitle\thispagestyle{empty}
	\begin{abstract}
		Remote state preparation with verifiability (RSPV) is an important quantum cryptographic primitive \cite{GVRSP,cvqcinlt}. In this primitive, a client would like to prepare a quantum state (sampled or chosen from a state family) on the server side, such that ideally the client knows its full description, while the server holds and only holds the state itself. In this work we make several contributions on its formulations, constructions and applications. In more detail:
		\begin{itemize}\item We first work on the definitions and abstract properties of the RSPV problem. We select and compare different variants of definitions \cite{firstrsp,GVRSP,cvqcinlt,GMP}, and study their basic properties (like composability and amplification). 
			\item We also study a closely related question of how to certify the server's operations (instead of solely the states). We introduce a new notion named \emph{remote operator application with verifiability} (ROAV). 
			We compare this notion with related existing definitions \cite{SW87,MY04,MIPstar,MV21,NZ23}, study its abstract properties and leave its concrete constructions for further works.
			\item Building on the abstract properties and existing results \cite{BGKPV23}, we construct a series of new RSPV protocols. Our constructions not only simplify existing results \cite{GVRSP} but also cover new state families, for example, states in the form of $\frac{1}{\sqrt{2}}(\ket{0}\ket{x_0}+\ket{1}\ket{x_1})$. All these constructions rely only on the existence of weak NTCF \cite{BKVV,AMR22}, without additional requirements like the adaptive hardcore bit property \cite{BCMVV,AMR22}.
			\item As a further application, we show that the classical verification of quantum computations (CVQC) problem \cite{ABEM,MahadevVerification} could be constructed from assumptions on group actions \cite{ADMP20}. This is achieved by combining our results on RSPV with group-action-based instantiation of weak NTCF \cite{AMR22}, and then with the quantum-gadget-assisted quantum verification protocol \cite{FKD}.
		\end{itemize}
	\end{abstract}\clearpage
	\pagenumbering{arabic}
	\newpage
	\tableofcontents
	\section{Introduction}
	
    \subsection{Background}
	Development of quantum computers \cite{supremacy,zuchongzhi,neutralatom} leads to demands of various quantum cryptographic protocols. In a typical setting, there are a client and a remote quantum server. The client would like to achieve some quantum tasks, but it does not trust the server; thus the client would like to make use of cryptography to achieve its goal. Famous examples include quantum computation verification \cite{MahadevVerification,cvqcinlt}, multiparty quantum computations \cite{B21}, etc. In this work, we are interested in a basic and very important primitive called \emph{remote state preparation} (RSP) \cite{firstrsp}, which we introduce below. \par
    \subsubsection{Remote state preparation: an overview}\label{sec:1.1.1}
    In the RSP problem, ideally, the client would like to prepare a quantum state (sampled or chosen from a state family) on the server side; thus in the end the client knows the description of the state, while the server simply holds the state. The trivial solution is to simply send the quantum state through a quantum channel. RSP asks: how could we achieve this task using cheaper resources (for example, only classical communication), possibly under computational assumptions?\par
    Studies of RSP have a long history \cite{Pati99,firstrsp}. One setting of RSP is the fully honest setting \cite{firstrsp}. In this work, we are interested in the cryptographic setting where the server could be malicious. Then a formulation of RSP should at least have a correctness requirement and a security requirement.\par
	The natural correctness requirement for RSP says that when the server is honest, the server gets the state while the client gets the state description. For security, there are different security notions, including blindness (or privacy, secrecy) and verifiability (or soundness) \cite{qfactory,GVRSP,jiayu20}. In this paper we focus on RSP with verifiability (RSPV). In RSPV, intuitively, the client is able to verify that in the case of passing (or called acceptance, non-aborting) the server  (approximately) really gets the state, as if it is sent through a quantum channel. A malicious server who attempts to get other information by deviating from the protocol would be caught cheating by the client.\par
    As a natural quantum task, the RSPV problem is interesting on its own. What's more, it has become an important building block in many other quantum cryptographic protocols. For example, \cite{GVRSP} first constructs a classical channel cryptography-based RSPV and uses it to achieve classical verification of quantum computations; \cite{GMP} explores more applications of RSPV; \cite{cvqcinlt} takes the RSPV approach to achieve classical verification of quantum computations with linear total time complexity. Many quantum cryptographic protocols rely on the quantum channel and quantum communication, and an RSPV protocol could serve as a compiler: it allows us to replace these quantum communication steps by other cheaper resources, like the classical communication.\par
    On the one hand, there have been many important and impressive results in this direction; on the other hand, there are also various limitations or subtleties in existing works, including formulations and constructions. Below we discuss existing works in more detail and motivate our results.
	\subsection{Existing Works and Motivating Questions}
    \subsubsection{Formulations and abstract properties of RSPV}\label{sec:1.1.3}
    We first note that there are many variants of definitions for RSPV. 
	For example, there are two subtlely different types of security notions, the \emph{rigidity-based} (or isometry-based) soundness \cite{qfactory,GMP} and \emph{simulation-based} soundness \cite{firstrsp,GVRSP,cvqcinlt}. 
    Existing works do not seem to care about the differences; we note that these differences could have impact on the abstract properties of the definitions and could affect their well-behaveness. 
	For example, we would like RSPV to have sequential composability between independent instances: if the client and the server execute an RSPV protocol for a state family $\cF_1$, and then execute an RSPV protocol for a state family $\cF_2$, we would like the overall protocol to be automatically an RSPV for $\cF_1\otimes \cF_2$. If such a sequential composability property holds, protocols for tensor products of states could be reduced to protocols for each simple state family.\par
	In this background, we argue that:
	\begin{center}
		\emph{It's helpful to compare variants of definitions and formalize basic abstract properties.}
	\end{center}
	We would like to have a more well-behaved framework for RSPV, which will lay a solid foundation for concrete constructions.\par
	\subsubsection{Definitions of the primitive for certifying the server's operations}\label{sec:1.2.2}
	RSPV talks about the certification of server-side \emph{states}. A closely related question is: how could the client certify the server-side \emph{operations}?\par
	To address this problem, existing works raise the notion of self-testing \cite{SW87,PR92,MY04}. One famous scenario of self-testing is in non-local games \cite{MIPstar,RUV}. In this scenario, the verifier sends questions to two spatially-separated but entangled quantum provers (or called servers). The verifier's questions and passing conditions are specially designed so that the provers have to perform specific operations to pass. 
	This provides a way to constrain the provers' operations through only classical interactions and spatial separation, which has become a fundamental technique in the study of non-local games.\par
    Recently a series of works \cite{MV21,FWZ,BGKPV23,NZ23} study the single-server cryptographic analog of the non-local game self-testing. \cite{MV21} studies the cryptographic analog of the CHSH game, where the server needs to prepare the Bell states and perform measurements on two of its registers. \cite{MTHAT22,FWZ} further extend it to three-qubit and $N$-qubit; \cite{KLVY22,NZ23} make use of QFHE \cite{Mahadev2017ClassicalHE} to formulate and address the problem, where the FHE-encrypted part takes the role of one prover and the unencrypted part takes the role of the other prover.\par 
	What's common in these existing works is that they are defining the single-server cryptographic analog of the non-local game self-testing to be a protocol where the single server is playing the roles of \emph{both} of the two provers. In this work we are interested in another viewpoint, which is:
	\begin{center}
		\emph{How could we define a single-server cryptographic analog of the non-local game self-testing, where the single server is playing the role of \emph{one} of the two provers?}
	\end{center}
	\subsubsection{Existing RSPV constructions}
	Maybe the most natural question in the direction of RSPV is:
	\begin{center}\emph{For what state families could we achieve RSPV?}
	\end{center} 
	Let's quickly review what state families have been achieved in existing works. \cite{GVRSP} achieves RSPV for $\{\ket{+_{\theta}}:=\frac{1}{\sqrt{2}}(\ket{0}+e^{\mi\theta\pi/4}\ket{1}),\theta\in \{0,1,2\cdots 7\}\}$; independently \cite{qfactory} gives a candidate RSPV construction for this state family and conjectures its security. \cite{GVRSP} achieves RSPV for tensor products of BB84 states: $\{\ket{0},\ket{1},\ket{+},\ket{-}\}^{\otimes n}$. \cite{cvqcinlt} achieves RSPV for tensor products of $\ket{+_\theta}$ states. \cite{BGKPV23} achieves RSPV for BB84 states.\par
	 For a broader viewpoint let's also review some famous self-testing protocols. \cite{MV21} achieves cryptographic self-testing for Bell states and the corresponding X/Z measurements. \cite{MTHAT22} achieves self-testing for a 3-qubit magic state and the corresponding measurements. \cite{FWZ} achieves self-testing for multiple Bell pairs and measurements. \cite{NZ23} achieves self-testing for ``all-X'' and ``all-Z'' operations.\par
	As a summary, one significant limitations of existing works is that they could only handle simple tensor product states and operators. We consider this situation very undesirable since ideally we want a protocol that is sufficiently powerful to cover all the computationally-efficient state families. The lack of concrete protocols also restricts the applications of RSPV as a protocol compiler: suppose that we have a quantum cryptographic protocol that starts with sending states in (for example) $\{\frac{1}{\sqrt{2}}(\ket{0}\ket{x_0}+\ket{1}\ket{x_1}), x_0,x_1\in \{0,1\}^n\}$, it's not clear how to use existing RSPV protocols to compile this quantum communication step to the classical communication. (Note that although there are mature techniques for creating this type of states with some other security notions \cite{Mahadev2017ClassicalHE,BCMVV,MahadevVerification}, it's not known how to construct RSPV for it.)\par
	 Besides the existency problem, it's also desirable if we could weaken the assumptions needed or simplify existing results. 
	\subsubsection{Classical verification of quantum computations and its assumption}\label{sec:1.2.4}
	We will apply our results on RSPV to the classical verification of quantum computations (CVQC) problem. In this problem a classical client would like to delegate the evaluation of a BQP circuit $C$ to a quantum server, but it does not trust the server; thus it wants to make use of a protocol to verify the results. 
	The first and perhaps the most famous construction is given by Mahadev \cite{MahadevVerification}. Building on Mahadev's work and related results \cite{Mahadev2017ClassicalHE,BCMVV,MahadevVerification}, a series of new CVQC protocols are constructed \cite{ACGH,GVRSP,CCT,cvqcinlt,NZ23,AMR22}.\par 
	However, one undesirable situation is that all the existing constructions are either based on LWE, or based on the random oracle heuristic, or based on some new conjecture on the reduction between security notions. In more detail:
	\begin{itemize}
		\item Most existing works \cite{MahadevVerification,CCT,GVRSP} make use of NTCF with the ``adaptive hardcore bit property'' (or some stronger variants). The only known instantiation of this primitive from standard assumption is based on LWE.\par
		\item \cite{AMR22} gives an instantiation of NTCF with a weaker variant of the adaptive hardcore bit property using cryptographic group actions. However, to achieve the adaptive hardcore bit property it relies on an unproven conjecture on the reduction between these two notions. Alternatively one could make use of the random oracle to construct CVQC without the adaptive hardcore bit property \cite{BKVV,cvqcinlt} but the instantiation of the random oracle is heuristic.
		 \item If we only make use of NTCF without the adaptive hardcore bit property, we could do RSPV for BB84 states \cite{BGKPV23}, but BB84 states are not known to be sufficient for constructing CVQC.
		\item \cite{NZ23} makes use of a quantum FHE-based approach to construct CVQC. However instantiations of quantum FHE rely on (variants of) NTCF and the classical FHE \cite{Mahadev2017ClassicalHE,GupteVinod}, which still rely on the LWE assumption.
	\end{itemize}
	As mentioned above, a type of cryptographic assumption that is different from LWE is the cryptographic group actions, for example, supersingular isogeny \cite{ADMP20,csidh}. In this background, we ask:
	\begin{center}
		\emph{Could we construct CVQC from cryptographic group actions?}
	\end{center}
    \subsection{Our Contributions}
	In this work we address or make progress to the questions above. 
	\subsubsection{Definitions and abstract properties of RSPV}\label{sec:1.3.1}
	We first work on the definitions and abstract properties of RSPV. This part is in Section \ref{sec:2} and \ref{sec:3}.
	\begin{enumerate}
		\item We first formulate and review the hierarchy of notions for formalizing RSPV. This includes the notions of registers, cq-states, protocols, paradigms of security definitions, etc.
		\item We then formalize the notion of RSPV. We formulate the soundness as a simulation-based definition and compare this definition with several other variants (like the rigidity-based soundness, which is also popular).
		\item We then study the abstract properties of RSPV including the sequential composability and amplification. This allows us to reduce the constructions of RSPV to primitives that are relatively easier to construct.
	\end{enumerate}
	In summary, we clarify subtleties and build a well-behaved framework for RSPV, which enables us to build larger protocols from smaller or easier-to-construct components. In later part of this work we will build concrete RSPV constructions under this framework.
	\subsubsection{Our cryptographic analog of self-testing: remote operator application with verifiability (ROAV)}\label{sec:1.3.2}
	We then introduce a new notion called remote operator application with verifiability (ROAV), as our answer to the question in Section \ref{sec:1.2.2}. This part is in Section \ref{sec:4}.\par
	Let's first review how the self-testing-based protocols typically work in the non-local game setting. Suppose the verifier would like to make use of the prover 1 and prover 2 to achieve some tasks ---- for example, testing the ground state energy of a local Hamiltonian. One typical technique is to design two subprotocols $\pi_{\ttest}$ and $\pi_{\tcomp}$. Subprotocol $\pi_{\ttest}$ is a non-local game with a self-testing property, while subprotocol $\pi_{\tcomp}$ is to test the Hamiltonian. Furthermore, the games are designed specially so that the prover 2, without communicating with the prover 1, could not decide which subprotocol the verifier is currently performing. 
	Then the setting, the overall protocol and the security proof roughly go as follows:
	\begin{enumerate}
		\item[(Setting)]The prover 1 and prover 2 initially hold EPR states. They receive questions from the verifier, make the corresponding measurements and send back the results.
		\item[(Overall protocol)] The verifier randomly chooses to execute either $\pi_{\ttest}$ or $\pi_{\tcomp}$ (without telling the provers the choices).
		\item[(Security proof)]
		\item To pass the overall protocol the provers have to pass $\pi_{\ttest}$ with high probability. By the property of $\pi_{\ttest}$ the operations of the prover 2 has to be close to the honest behavior.
		\item By the design of $\pi_{\ttest}$ and $\pi_{\tcomp}$, and the fact that the prover 2 is close to the honest behavior in $\pi_{\ttest}$, we could argue that the prover 2 is also close to the honest behavior in $\pi_{\tcomp}$.
		\item Since we already know the prover 2 is close to be honest in $\pi_{\tcomp}$, it's typically easy to analyze the execution of $\pi_{\tcomp}$ directly and show that it achieves the task.
	\end{enumerate}
Recall that we would like to define a single-server cryptographic analog of the non-local game self-testing where the single server plays the role of one of the two provers. So what does the ``non-local game self-testing'' mean here? Our idea is to consider both the step 1 and step 2 in the security proof template above as the ``non-local game self-testing''. Then let's focus on the prover 2 (as ``one of the two provers'' in our question) and assume the prover 1 is honest. Then we could do the following simplifications in the non-local game self-testing:
\begin{itemize}\item In $\pi_{\ttest}$ we could assume the prover 1 first measures its states following the verifier's question; as a result, the joint state of the prover 1 and prover 2 becomes a cq-state where the verifer knows its description. Then the verifier also gets the measurement results from the prover 1's answer; so in the end we only need to consider the joint cq-state between the verifier and the prover 2.
	\item For $\pi_{\tcomp}$, since we do not consider the step 3 above as a part of the self-testing notion, the question to the prover 1 in $\pi_{\tcomp}$ could be left undetermined. Then $\pi_{\tcomp}$ is not designed for any specific task any more, which in turn makes the self-testing a general notion.
\end{itemize}
	Now the setting and the first two steps in the security proof could be updated as follows:
\begin{enumerate}
	\item[(Setting)]In $\pi_{\ttest}$ the client and the prover 2 initially hold a cq-state, where the verifier holds the classical part and the prover 2 holds the quantum part. In $\pi_{\tcomp}$ the prover 1 and prover 2 hold EPR states and the verifier does not have access to the prover 1's information.
	\item[(Security proof)] To make the first two steps in the security proof template above go through, we should at least require that:\par ``The prover could pass in $\pi_{\ttest}$'' should imply that the prover's operation in $\pi_{\tcomp}$ is close to the desired one.
\end{enumerate}
	This gives us the basic intuition for formalizing our single-server cryptographic analog of the non-local game self-testing.
	Below we introduce the notion, which we call \emph{reomte operator application with verifiability} (ROAV).\par 
	An ROAV for a POVM $(\cE_{1},\cE_{2},\cdots, \cE_{D})$ is defined as a tuple $(\rho_{\ttest},\pi_{\ttest},\pi_{\tcomp})$ where:
	\begin{itemize}
		\item The setting contains the following registers. $\bD$ is a client-side classical register, $\bQ$ is a server-side quantum register, $\bP$ is a quantum register in the environment with the same dimension with $\bQ$.
		\item $\rho_{\ttest}$ is a cq-state on registers $\bD$ and $\bQ$.
		\item $\pi_{\ttest}$ is the protocol for the test mode. In the test mode $\rho_{\ttest}$ is used as the input state and $\bP$ is empty.
	\item $\pi_{\tcomp}$ is the protocol for the computation mode (where the operations are applied). Denote the state of maximal entanglement (multiple EPR pairs) between $\bP$ and $\bQ$ as $\Phi$. In the comp mode $\Phi$ is used as the input state and $\bD$ is empty. Note that the execution of $\pi_{\tcomp}$ does not touch $\bP$.	
	\end{itemize} 
	\begin{figure}
		\begin{subfigure}{0.22\textwidth}
			\begin{tikzpicture}
				\node[draw,circle] (P1) at (0,0) {P1};
				\node[draw,circle] (P2) at (2,0) {P2};
				\node[draw,circle] (V) at (1,-1.5) {V};
				\draw[<->] (P1) -- (V);
				\draw[<->] (P2) -- (V);
				\draw[dashed] (P1) -- (P2);
			\end{tikzpicture}
		\end{subfigure}
			\begin{subfigure}{0.3\textwidth}
				\begin{tikzpicture}
					\node[draw,circle] (P1) at (0,0) {P1};
				\node[draw,circle] (P2) at (2,0) {P2};
				\node[draw,rectangle,minimum height=20,minimum width=20] (V) at (1,-1.5) {C};
				\draw (-1.1,0.6) rectangle (2.6,-0.6);
				\coordinate (lowercenterofrect) at (1,-0.6);
				\node (S) at (-0.9,0.3) {S};
				\draw[dashed] (P1) -- (P2);
				\draw[<->] (V) -- (lowercenterofrect);
				\end{tikzpicture}
			\end{subfigure}
				\begin{subfigure}{0.18\textwidth}
					\begin{tikzpicture}
				\node[draw,circle] (P2) at (2,0) {Q};
				\node[draw,rectangle,minimum width=50,minimum height=20] (V) at (2,-1.5) {C};
				\draw (1.4,0.6) rectangle (2.8,-0.6);
				\coordinate (arrowcqpointto) at (2,-0.6);
				\node (S) at (2.6,0.3) {S};
				\node[draw,rectangle,minimum width=15,minimum height=15] (K) at (1.5,-1.5) {D};
				\draw[<->] (arrowcqpointto) -- (V);
					\end{tikzpicture}
				\end{subfigure}
				\begin{subfigure}{0.2\textwidth}
					\begin{tikzpicture}
						\node[draw,circle] (P1) at (0,0) {P};
				\node[draw,circle] (P2) at (2,0) {Q};
				\node[draw,rectangle,minimum width=50,minimum height=20] (V) at (2,-1.5) {C};
				\draw (1.4,0.6) rectangle (2.8,-0.6);
				\coordinate (arrowcqpointto) at (2,-0.6);
				\node (S) at (2.6,0.3) {S};
				
				\draw[<->] (arrowcqpointto) -- (V);
				\draw[dashed] (P1) -- (P2);
					\end{tikzpicture}
				\end{subfigure}
		\caption{From the left to the right: the non-local game self-testing, previous definitions of its single party analog, and our definition ($\pi_{\ttest}$ and $\pi_{\tcomp}$ for the last two diagrams). Here $\leftrightarrow$ stands for interactions, and {-}{-}{-} stands for quantum entanglements; C stands for the client and S stands for the server.}
	\end{figure} 
	
	The soundness of ROAV is defined roughly as follows: for any adversary $\fAdv$, at least one of the following is true:
	\begin{itemize}\item In the test mode ($\pi_{test}$ is executed with input state $\rho_{\ttest}$ against adversary $\fAdv$), the adversary gets caught cheating with significant probability; \item In the comp mode ($\pi_{\tcomp}$ is executed with input state $\Phi$ against adversary $\fAdv$), the final state gives the outcome of the following operations:\par
		 The measurement described by $(\cE_{1},\cE_{2},\cdots, \cE_{D})$ is applied on $\bQ$, and the client gets the measurement result $i\in [D]$; the corresponding output state (on register $\bP$ and $\bQ$) is $(\bbI\otimes \cE_i)(\Phi)$.\end{itemize}
		 We finally note that the definition of ROAV also use the simulation-based security definition paradigm, like RSPV.
	\paragraph{ROAV as the operator analog of RSPV} In the previous discussion we focus on the analog between our ROAV primitive and the non-local game self-testing. There is also another intuition that is analogous to the intuition of RSPV (see Section \ref{sec:1.1.1}): In the ROAV problem the server is provided an undetermined input state, and the client would like to apply an operation (from an operator family) on it; in the end the client knows the description of the operation, while the server simply holds the output of the operation.\par
	We also note that such a state-operator analog also appears in other directions like \cite{RY21,JLS}. 
\paragraph{Abstract constructions using ROAV}
	After giving the abstract definitions, we show several potential applications of our notions. We first show that ROAV is potentially a useful tool for constructing RSPV protocols for more general state families. 
	Then we construct a Hamiltonian ground energy testing protocol based on specific RSPV and ROAV. Our construction shares similarities to Grilo's Hamiltonian verification protocol in the 2-party setting \cite{Grilo17}. We note that these constructions are abstract constructions and concrete constructions for nontrivial ROAV remain open.
	\subsubsection{New RSPV constructions}\label{sec:1.3.3}
	Now we introduce a series of new RSPV constructions. Our results not only give arguably simplified constructions for existing results but also cover new state families. This part is in Section \ref{sec:5}.
	\paragraph{Overall approach} Instead of constructing each protocol directly from cryptographic assumptions, we study how different protocols could be \emph{reduced} to existing ones in a black-box manner. Following this approach, we build a series of RSPV protocols step by step from RSPV for BB84 states $\{\ket{0},\ket{1},\ket{+},\ket{-}\}$. Furthermore, these steps could be classified into two classes:
	\begin{itemize}
		\item In one class of steps, the reduction either has a simple intuition or is an application of the abstract properties that we have proved in Section \ref{sec:3} (in more detail, sequential composability and amplification).
		\item In the other class of steps, we could work on an ``information-theoretic core'' (IT-core) where the analysis is purely quantum information theoretic, and there is no appearance of computational notions like computational indistinguishability.\par
	\end{itemize}
	Our reductions are illustrated in Figure \ref{fig:2}. We elaborate technical details in Section \ref{sec:5.1}. As an example, let's describe the first step of our reductions. We assume an RSPV for BB84 states; then the client would like to prepare a sequence of such states in which there is only one $\ket{+}$ state and no $\ket{-}$ state. Equivalently, this could be written as $\frac{1}{\sqrt{2}}(\ket{x_0}+\ket{x_1})$ where the hamming weight of $x_0\oplus x_1$ is exactly 1. What the client needs to do is to repeat the RSPV-for-BB84 protocol for many rounds and tell the server which states it wants to keep. This could be easily understood as a repeat-and-pick process.\par
	We argue that our approach is cleaner and easier to understand compared to the original construction in \cite{GVRSP} at least in the following sense: In \cite{GVRSP} the computational indistinguishability arguments and quantum information theoretic arguments are mixed together, which could lead to complicated details \cite{discussionwithVidick}. By separating two types of steps in the reductions explicitly, each step has a relatively simple intuition and there is much less room for complicated details.
	\paragraph{Main results} Among these reductions, we consider the following two results particularly interesting.
	\begin{figure}
		\begin{tikzpicture}
			\node (n1) at (0,0) {RSPV for BB84};
			\node (n2) at (0,-1.3) {RSPV for $\{\frac{1}{\sqrt{2}}(\ket{x_0}+\ket{x_1}):x_0,x_1\in \{0,1\}^m,\text{HW}(x_0\oplus x_1)=1,\tParity(x_0)=0\}$};
			\node (n3) at (0,-2.65) {RSPV for $\{\frac{1}{\sqrt{2}}(\ket{x_0}+\ket{x_1}):x_0,x_1\in \{0,1\}^m,\text{HW}(x_0\oplus x_1)=1,\tParity(x_0)=0\}^{\otimes n}$};
			\node(n3b) at (0,-3.3) {{\color{purple}\large\textbf{+}}};
			\node (n4p) at (0,-4.3) {\begin{tabular}{c}IT-core for preparing $\{\frac{1}{\sqrt{2}}(\ket{x_0^{(1)}||x_0^{(2)}||\cdots ||x_0^{(n)}}+\ket{x_1^{(1)}||x_1^{(2)}||\cdots ||x_1^{(n)}})$\\$:\forall i\in [n],x_0^{(i)},x_1^{(i)}\in \{0,1\}^m,\text{HW}(x_0^{(i)}\oplus x_1^{(i)})=1;\tParity(x_0^{(1)})=0\}$\end{tabular}};
			\node (n4) at (0,-6.3) {\begin{tabular}{c}RSPV for $\{\frac{1}{\sqrt{2}}(\ket{x_0^{(1)}||x_0^{(2)}||\cdots ||x_0^{(n)}}+\ket{x_1^{(1)}||x_1^{(2)}||\cdots ||x_1^{(n)}})$\\$:\forall i\in [n],x_0^{(i)},x_1^{(i)}\in \{0,1\}^m,\text{HW}(x_0^{(i)}\oplus x_1^{(i)})=1;\tParity(x_0^{(1)})=0\}$\end{tabular}};
			\node (n5) at (0,-8) {RSPV for $\{\frac{1}{\sqrt{2}}(\ket{0}\ket{x_0}+\ket{1}\ket{x_1}):x_0,x_1\in \{0,1\}^n\}$};
			\node(n3b) at (0,-8.7) {{\color{purple}\large\textbf{+}}};
			\node (n6p) at (0,-9.4) {IT-core for prparing $\ket{+_\theta},\theta\in \{0,1\cdots 7\}$};
			\node (n6) at (0,-10.8) {RSPV for $\ket{+_\theta},\theta\in \{0,1\cdots 7\}$};
			\draw[->,blue] (n1) -- (n2);
			\draw[->,blue] (n2) -- (n3);
			\draw[->,blue] (n4p) -- (n4);
			\draw[->,blue] (n4) -- (n5);
			\draw[->,blue] (n6p) -- (n6);
		\end{tikzpicture}
	\caption{The reduction diagram. Here {\color{blue}$\rightarrow$} means this step either has a simple intuition or is from the framework that we formalize in Section \ref{sec:3}; {\color{purple}\textbf{+}} means the analysis of the IT-core below it is purely quantum information theoretic, and the RSPV above it is used to compile the IT-core to a full cryptographic protocol. HW is the hamming weight, Parity is the total parity.}\label{fig:2}
	\end{figure}
	\begin{thm}\label{thm:1.2p}
		Assuming the existence of RSPV for BB84 states, there exists an RSPV for state family $\{\frac{1}{\sqrt{2}}(\ket{0}\ket{x_0}+\ket{1}\ket{x_1}),x_0,x_1\in \{0,1\}^n\}$.
	\end{thm}
	This is an RSPV protocol for a new state family. Note that although there are plenty of works \cite{BCMVV,BKVV} that allow the client to prepare these states with some other types of security (like claw-freeness, etc), as far as we know, this is the first time that an RSPV for it is constructed.\par
	We also recover the results of RSPV for 8-basis states \cite{GVRSP}. 
	\begin{thm}\label{thm:1.2}
		Assuming the existence of RSPV for BB84 states, there exists an RSPV for state family $\{\ket{+_\theta}:=\frac{1}{\sqrt{2}}(\ket{0}+e^{\mi\pi\theta/4}\ket{1}),\theta\in \{0,1,2\cdots 7\}\}$.
	\end{thm}
	\paragraph{Instantiation of RSPV for BB84}We note that we still need to instantiate the RSPV for BB84 part to get a concrete protocol. Luckily the RSPV for BB84 has been studied relatively thoroughly: there are multiple constructions \cite{GMP,BGKPV23} and we have a better understanding on the assumptions needed \cite{BGKPV23}. After instantiating the BB84 part by \cite{BGKPV23}, we get RSPV constructions for these state families from weak NTCF, without requiring the adaptive hardcore bit property (elaborated below).\par
	\subsubsection{Application: CVQC from cryptographic group actions}\label{sec:1.3.4}
	Now we apply our results to the classical verification of quantum computations (CVQC) problem. This part is in Section \ref{sec:6}.\par
	 As the preparation, we give a more detailed review on the variants of NTCF, and their relations to RSPV. We refer to Section \ref{sec:1.2.4} for a CVQC-centric background. 
	\paragraph{More backgrounds} Noisy trapdoor claw-free functions (NTCF) is a popular and powerful primitive in quantum cryptographic tasks like CVQC and RSPV. Informally, this primitive is defined to be a function family that satisfies the following requirements. Below we use $f$ to denote a function sampled from this function family.
	\begin{itemize}\item Trapdoor: this means that all the parties could evaluate $f$, but only the client, who holds the ``trapdoor'' information, could invert $f$.
		\item Noisy 2-to-1: the un-noisy 2-to-1 means that for each $y$ in the range there exist exactly two preimages $x_0,x_1$ such that $f(x_0)=f(x_1)=y$. ``Noisy'' means that the evaluation of the function could be randomized, which makes the instantiation of the primitive easier.\par
		Below when we describe other properties we use the un-noisy version to simplify the description.
		\item  Claw-free: the adversary could not efficiently find $x_0,x_1$ such that $f(x_0)=f(x_1)$.
	\end{itemize}
	In practice, we often use some variants of NTCF instead of the standard requirements above, to either make it more powerful or make the instantiation easier. Popular variants or additional requirements of NTCF include:
	\begin{itemize}
		\item Adaptive hardcore bit: the adversary could not find $(d,y)$ with probability better than $\frac{1}{2}$ (the probability of random guessing) such that $d\cdot x_0\equiv d\cdot x_1\mod 2\land d\neq 0$, where $x_0, x_1$ are two preimages of $y$.\par
		The RSPV for BB84 states is known to exist without this property \cite{BGKPV23}, but Mahadev's constructions for CVQC \cite{MahadevVerification} and RSPV for  $\ket{+_\theta}$ \cite{GVRSP} require this property.
		\item Extended function family: there is another function family $g$ that is injective (instead of 2-to-1) and indistinguishable to $f$. Existing constructions for CVQC and RSPV for $\ket{+_\theta}$ \cite{MahadevVerification,GVRSP} also require this property.
		\item Inverse-polynomial correctness error (or weak correctness): this means that the 2-to-1 property is allowed to hold up to an inverse-polynomial error. This is used in \cite{AMR22} to define a primitive called weak TCF.
	\end{itemize}
	In our work we build our protocols on weak NTCF, that is, we allow inverse-polynomial correctness error and do not require additional properties like the adaptive hardcore bit property. In other words, we only use a relatively weak assumption from different variants of NTCF/TCF. This assumption could be instantiated from either LWE \cite{BCMVV} or assumptions on group actions \cite{AMR22}.
	\paragraph{Applications of our results on CVQC}
	We first note that, \cite{BGKPV23} could be adapted easily to weak NTCF:
	\begin{fact}\label{fact:1}
		By \cite{BGKPV23}, there exists an RSPV for BB84 assuming weak NTCF.
	\end{fact}
	As discussed in Section \ref{sec:1.3.3}, this gives us a series of RSPV protocols from weak NTCF.\par
	Then by \cite{FKD}, if the client samples and sends a series of $\ket{+_\theta}$ states to the server, it could do quantum computation verification using these states. Combining Fact \ref{fact:1}, Theorem \ref{thm:1.2} and \cite{FKD}, we get:
	\begin{thm}
		Assuming the existence of weak NTCF, there exists a classical verification of quantum computation protocol.
	\end{thm}
	Finally we recall the results in \cite{AMR22}:
	\begin{thm}[\cite{AMR22}]
		Under certain assumptions on cryptographic group actions, there exists a family of weak TCF.
	\end{thm}
	Thus we have:
	\begin{thm}
		Under certain assumptions on  cryptographic group actions, there exists a CVQC protocol.
	\end{thm}
	As an additional note, the work in \cite{AMR22} is largely on how to deal with the adaptive hardcore bit property; the fact that we do not need the adaptive hardcore bit may help us simplify the analysis or even construction in \cite{AMR22}.
	\subsection{More Related Works}\label{sec:1.4}
	\paragraph{Previous versions} 
	The major versions of this works could be roughly described as follows.
	\begin{itemize}
		\item In the initial versions (comming out in around 2023), we develop a simple framework for working on RSPV problem, and raised the notion of ROAV.
		\item The next major versions are prepared for ITCS25 and made public in Nov 2024 (below we call it the ITCS version). Compare to the previous versions, the abstract framework parts (described in Section \ref{sec:1.3.1}, \ref{sec:1.3.2}) of this work are sigfinicantly updated, and the concrete constructions part (described in Section \ref{sec:1.3.3}, \ref{sec:1.3.4}) are completely new.
		\item The current versions have the following changes compared to the ITCS version: in Section \ref{sec:r4} we give explicitly the translation from the results in \cite{BGKPV23} to an RSPV for BB84 states (that is, the proof of Fact \ref{fact:1}). The additional technical works come from the fact that \cite{BGKPV23} describes their protocol as a test of a qubit, which is not the same as an RSPV for BB84 states; so some technical works are needed for translation.\footnote{We thank Kaniuar Bacho, James Bartusek, Yasuaki Okinaka and anonymous reviewers for pointing this out.} We also discuss another approach for achieving CVQC from RSPV for BB84 state in Appendix \ref{app:p2}, based on the results in \cite{morimae20}.\footnote{We thank anonymous reviewer for pointing out this approach.}
	\end{itemize}
	\paragraph{Concurrent works} After the ITCS version is made public, we get aware of two other concurrent works \cite{Bacho2024CompiledNG,Bartusek2024kit} that have overlap with our work. In more detail:
	\begin{itemize}
		\item \cite{Bacho2024CompiledNG}:  this work propose a compiler for transforming non-local games to single-server cryptographic protocols. Interestingly, their approach also leads to a protocol for CVQC from any plain (weak) NTCF, which implies a CVQC protocol from cryptographic group actions. Besides this overlap, the focus of their works and our works are different. (Note that although they also make use of a type of remote state preparation in their work, they do not study RSP with verifiablity; for comparison, our work primarily investigates RSPV.)
		\item \cite{Bartusek2024kit}: this work studies a notion that they call \emph{oblivious state preparation}. Interestingly, their approach also leads to a protocol for CVQC from any plain (weak) NTCF, which implies a CVQC protocol from cryptographic group actions. Besides this overlap, the focus of their works and our works are different. (Note that oblivious state preparation could be seen as a type of remote state preparation, but its soundness is quite different from RSPV. Our work focuses on RSPV, which is different from their work.)
	\end{itemize}
	\paragraph{RSP with other types of security} There are many works about remote state preparation but with other types of security (instead of RSPV). These works may or may not use the name ``remote state preparation''. As examples, \cite{GupteVinod} gives an RSP for the gadgets in \cite{DSS16}; \cite{shmueli22} gives an RSP for the quantum money states. When we only consider the honest behavior, RSP is a very general notion.
	\paragraph{Other state-operation analog} There are also several notions in quantum cryptography and complexity theory that have a state-operation analog. As examples, \cite{RY21} studies the complexity of interactive synthesis of states and unitaries. \cite{JLS} studies pseudorandom states and unitaries. In a sense, the relation of states and unitaries in these works is a ``state-operation analog'', as the RSPV and ROAV in our work.
	\subsection{Open Questions and Summary}
	Our work gives rise to a series of open questions; we consider the following two particularly interesting.
	\begin{itemize}
	\item One obvious open question coming out of this work is to give a construction for ROAV. Our work focuses on its definitions and applications in an abstract sense; an explicit construction of ROAV would allow us to instantiate these applications.
	\item For RSPV, although our results give new constructions for new state families, this is still far away from a general solution. Ideally we would like to have an RSPV for each computationally efficient state family. Whether this is possible and how to achieve it remain open.
	\end{itemize}
	In summary, two major part of this work is the framework and the constructions of RSPV protocols. By formulating and choosing notions and studying their abstract properties, we build an abstract framework for RSPV that is sufficiently well-behaved, which enables us to build more advanced, complicated protocols from more elementary, easy-to-construct components. Building on this framework, we construct a series of new RSPV protocols that not only (in a sense) simplify existing results but also cover new state families. Then we combine our results with existing works to show that the CVQC problem could be constructed from assumptions on cryptographic group actions. We also raised a new notion for certifying the server's operations. We consider our results as a solid progress in the understanding of RSPV; hopefully this will lay the foundation for further works.
	\section*{Acknowledgements}
 	This work is supported by different fundings at different time during its preparation:
 	\begin{itemize}\item Partially supported by the IQIM, an NSF Physics Frontiers Center 
 (NSF Grant PHY-1125565) with support of the Gordon and Betty Moore 
 Foundation (GBMF-12500028).\item This work is partially done when the author was visiting Simons Institute for Theory of Computing. 
 \item This work is partially supported by Zhongguancun Laboratory.
 	\end{itemize}
 The author would like to thank Kaniuar Bacho, James Bartusek, Yasuaki Okinaka, Thomas Vidick, Zhengfeng Ji, Anne Broadbent, Qipeng Liu and anonymous reviewers for discussions.
 \section{Preparation}\label{sec:2}
	\subsection{Basics of Quantum Information and Mathematics}
	We refer to \cite{NielsenChuangs} for basics of quantum computing, and refer to \cite{KLtextbook} for basics of cryptography. In this section we clarify basic notations and review basic notions.
	\subsubsection{Basic mathematics}
	\begin{nota}
		We use $[m]$ to denote $\{1,2,\cdots m\}$ and use $[0,m]$ to denote $\{0,1,2,\cdots m\}$. We use $1^n$ to denote the string $\underbrace{111\cdots 1}_{\text{for $n$ times}}$. We use $S\backslash T$ to denote the set difference operation. We use $|S|$ or $\tsize(S)$ to denote the size of a set $S$. We use $\leftarrow_r$, or simply $\leftarrow$, to denote ``sample from''. $|\cdot |$ refers to the Euclidean norm by default when it's applied on a vector. We use $\bN$ to denote the set of positive integers, use $\bZ$ to denote the set of integers, and use $\bZ_q$ to denote $\bZ/q\bZ$. 
	\end{nota}
	Below we review famous bounds in probability \cite{Azuma}.
	\begin{fact}[Markov inequality]
		Suppose $X$ is a non-negative random variable. Then
		$$\Pr[X\geq a]\leq \bE[X]/a.$$
	\end{fact}
	\begin{fact}[Chernoff bounds]
		Suppose for all $i\in [K]$, $s_i$ is a random variable independently sampled from $\{0,1\}$ with probability $1-p,p$ correspondingly. Then
		$$\forall \delta>0,\Pr[\sum_{i\in [K]}s_i\geq (1+\delta)pK]\leq e^{-\delta^2 pK/(2+\delta)}$$
		$$\forall \delta\in (0,1),\Pr[\sum_{i\in [K]}s_i\leq (1-\delta)pK]\leq e^{-\delta^2 pK/2}$$
	\end{fact}
	\begin{thm}[Azuma inequality]
		Suppose $(Z_t)_{t\in [0,K]}$ is a martingale. If for each $t\in [K]$, $|Z_t-Z_{t-1}|\leq c_t$, then
		$$\forall \delta>0, \Pr[|Z_K-Z_0|\geq \delta]\leq 2e^{-\frac{\delta^2}{\sum_{t\in [K]}c_t^2}}$$
	\end{thm}
	\begin{nota}\label{nota:2.2}
		$\fpoly (n)$ means a function of $n$ asymptotically the same as a polynomial in $n$. $\fneg (n)$ means a function of $n$ asymptotically tends to $0$ faster than any polynomial in $n$.
	\end{nota}
	\subsubsection{Basics of quantum information}
	\begin{nota}\label{nota:2.3}
		We use $\tPos(\cH)$ to denote the set of positive semidefinite operators over some Hilbert space $\cH$ and we use $\tD(\cH)$ to denote the set of density operators over $\cH$.\par
		Recall that density operators are positive semidefinite operators with trace equal to $1$. We call a positive semidefinite operator a subnormalized density operator if its trace is less than or equal to $1$.
	\end{nota}
	\begin{nota}
		For a pure state $\ket{\Phi}$, $\Phi$ is an abbreviation of $\ket{\Phi}\bra{\Phi}$.
	\end{nota}
	\begin{fact}\label{fact:3}
		Any $\rho\in \tD(\cH)$ could be purified to a state $\ket{\varphi}\in \cH\otimes\cH_{\bR}$ such that the partial state of $\ket{\varphi}$ restricted on $\cH$ is $\rho$.
	\end{fact}
	\begin{nota}\label{nota:1}
		We use $\cE(\rho)$ to denote the operation of an operator (either unitary or superoperator) on (normalized or unnormalized) density operator $\rho$.\par
		We use CPTP as an abbreviation of ``completely-positive trace-preserving''. We use POVM as an abbreviation of ``positive operator value measurement''. A POVM is described by a tuple of superoperators $(\cE_1,\cE_2,\cdots \cE_D)$.
	\end{nota}
	\begin{nota}
		$\Pr[\cE(\rho)\rightarrow o]$ is defined to be $\tr(\Pi_o(\cE(\rho)))$, where $\Pi_o$ is the projection on the first qubit onto value $o$.
	\end{nota}
	\begin{nota}
		For positive semidefinite operators $\rho,\sigma$ we use $\rho\approx_{\epsilon}\sigma$ to denote $\tTD(\rho,\sigma)\leq \epsilon$, where $\tTD(\cdot, \cdot)$ is the trace distance. We also say ``$\rho$ is $\epsilon$-close to $\sigma$''. For two pure states $\ket{\varphi},\ket{\psi}$, we use $\ket{\varphi}\approx_{\epsilon}\ket{\psi}$ to denote $|\ket{\varphi}-\ket{\psi}|\leq \epsilon$ where $|\cdot |$ is the Euclidean norm.\par
		Note that there is a difference between  trace distance and the Euclidean distance even for pure states.
	\end{nota}
	\begin{defn}[Bell basis]\label{defn:bellbasis}
		In a two qubit system, define the following four states as the Bell basis:
		$$\frac{1}{\sqrt{2}}(\ket{00}+\ket{11}),\frac{1}{\sqrt{2}}(\ket{00}-\ket{11}),$$
		$$\frac{1}{\sqrt{2}}(\ket{01}+\ket{10}),\frac{1}{\sqrt{2}}(\ket{01}-\ket{10}).$$
		Define $\ket{\Phi}=\frac{1}{\sqrt{2}}(\ket{00}+\ket{11})$, then these states could be denoted as $\fX^a\fZ^b\ket{\Phi}$, where $\fX^a$ means to apply $\fX$ if $a=1$ and apply identity if $a=0$. $\fZ^b$ is defined similarly.\par
		Now define the Bell-basis measurement as follows: the projection onto the Bell basis $\fX^a\fZ^b\ket{\Phi}$ has output value $(a,b)$.
	\end{defn}
		
	Finally we review the local Hamiltonian problem.
	\begin{defn}[\cite{Grilo17}]\label{defn:2.2}
		The following problem is called the XZ k-local Hamiltonian problem:\par
		Given input $(H,a,b)$ where $H$ is a Hamiltonian on $n$-qubit registers, $a,b$ are real value function of $n$, and they satisfy:
		\begin{equation}\label{eq:1}H=\sum_{j\in [m]}\gamma_jH_j,\quad \forall j,|\gamma_j|\leq 1\end{equation}
		\begin{equation}\label{eq:2}\forall j,H_j\in \{\sigma_X,\sigma_Z,I\}^{\otimes n}\text{ with at most $k$ appearances of non-identity terms}\end{equation}
		Decide which is the case:
		\begin{itemize}
			\item Yes-instance: The ground energy of $H$ is $\leq a$
			\item No-instance: The ground energy of $H$ is $\geq b$.
		\end{itemize}
	\end{defn}
	\begin{thm}[\cite{Ji16}]
		There exist $a(n),b(n)\in [0,1],b-a\geq 1/\fpoly(n)$ such that the XZ 5-local Hamiltonian problem is QMA-complete.
	\end{thm}
	\subsection{Registers, States, and Protocols}
	\subsubsection{Registers}In this work we work on \emph{registers}. Intuitively, a register is like a modeling of a specific space in the memory of a computer. It allows us to use its name to refer to the corresponding state space. There are two types of registers: classical registers and quantum registers. For a classical register, the underlying state space is a finite set; for a quantum register, it corresponds to a Hilbert space. A classical register could be a tuple of other classical registers and a quantum register could be a tuple of other quantum registers. We refer to \cite{watroustqi} for more information.\par
	\begin{nota}In this work we use the bold font (for example, $\bS$) to denote registers. The corresponding Hilbert space of register $\bS$ is denoted as $\cH_\bS$.\end{nota}
	Then when we design quantum protocols, we could write protocols that take registers as part of inputs. So the inputs of a protocol could be either values or registers. This is also called ``call-by-value'' versus ``call-by-reference'' in programming languages like Java or Python. For classical values, both call-by-value and call-by-reference (registers) work; but for quantum states the call-by-value does not work in general due to the quantum no-cloning principle.\par
	\begin{nota}
	We use $\Pi^{\breg}_{v}$ to denote the projection onto the subspace that the value of register $\breg$ is $v$. We use $\Pi^{\breg_1}_{=\breg_2}$ to denote the projection onto the subspace that the value of register $\breg_1$ is equal to the value of register $\breg_2$.
	\end{nota}
	\subsubsection{States}For a classical register $\bbC$, we could consider a probability distribution over states in its state space. For a quantum register $\bS$ and the associated Hilbert space $\cH_{\bS}$, we could consider the (normalized) density operators $\tD(\cH_{\bS})$ as the quantum analog of probability distributions.\par
	The joint state over a classical register $\bbC$ and a quantum register $\bS$ is modeled as a \emph{cq-state}. The density operator could be written as:
	$$\rho=\sum_{c\in \Domain(\bbC)}\underbrace{\ket{c}\bra{c}}_{\bbC}\otimes \underbrace{\rho_c}_{\bS}$$
	where $\forall c, \rho_c\in \tPos(\cH_{\bS})$.
	\subsubsection{Execution model of protocols}\label{sec:2.2.3}
In our work, we consider a setting where a client and a server interacting with each other. A more general modeling of cryptographic protocols might consider multiple parties and the scheduling of operations and messages might be concurrent; in this work we only consider the two-party setting between a client and a server and we could without loss of generality assume the operations are non-concurrent. That is, the overall protocol repeats the following cycle step by step:
\begin{enumerate}
	\item The client does some operations on its own registers.
	\item The client sends a message to the server.
	\item The server does some operations on its own registers.
	\item The server sends a messages to the client.
\end{enumerate}
We call one cycle above as one round.\par
Each party should have a series of operations as the honest execution. We also consider the setting where some party is corrupted (or called malicious, adversarial). In this work we care about the setting where the server could be malicious. The adversary's (that is, malicious server's) operations could also be described by a tuple of operations for each round: $\fAdv=(\fAdv_{\text{round 1}},\fAdv_{\text{round 2}}\cdots \fAdv_{\text{round n}})$, used in the step 3 in the cycle above. In later proofs we may simply use $\fAdv_1$, $\fAdv_2$, etc to refer to the adversarial operations in each round.
\begin{nota}\label{nota:2.8}Consider a client-server setting described above. Suppose the protocol takes input registers $\bregs$ and input values $\tvals$ as its arguments. The overall operation of the whole protocol run against an adversary $\fAdv$ is denoted as:
	\begin{equation}\label{eq:3}\text{ProtocolName}^{\fAdv}(\bregs,\tvals)\end{equation}
	Then the final state of the protocol execution is determined by the operation above and the initial state. Suppose the initial state (on the registers considered) is $\rho_{\text{in}}$, the final state could be denoted as
	\begin{equation}\text{ProtocolName}^{\fAdv}(\bregs,\tvals)(\rho_\text{in})\end{equation}
\end{nota}
\paragraph{About the messages and transcripts} In this work we model the message sending operation in step 2 above as an operation that writes messages on a specific server-side empty register. Since this register is on the server side the adversary in the later steps might erase its value. One alternative way for modeling classical messages is to introduce a stand-alone transcript registers explicitly, whose values could not be changed by later operations. These two modelings are equivalent in the problems that we care about.
	\subsection{Basic Notions in Cryptographic Protocols}
	We would like to design protocols to achieve some cryptographic tasks. In this section we review the basic components for defining and constructing cryptographic protocols.
	\subsubsection{Completeness, soundness, efficiency}\label{sec:2.3.1}
		Consider the client-server setting where the client would like to achieve some task with the server but does not trust the server. To define the cryptographic primitive for this problem, typically there are at least two requirements: a completeness (or called correctness) requirement and a soundness (or called security) requirement. These two requirements could roughly go as follows:
		\begin{itemize}
			\item (Completeness) If the server is honest, the task is achieved and the client accepts.
			\item (Soundness) For any malicious server, if the task is not achieved correctly the client catches the server cheating.
		\end{itemize}
		Allowing for the completeness error and soundness error, the definition roughly goes as follows.
		\begin{itemize}
			\item (Completeness) If the server is honest, the task is achieved and the client accepts with probability $c$.
			\item (Soundness) For any malicious server, the probability of ``the task is not achieved correctly and the client does not catch the server cheating'' is at most $s$.
		\end{itemize}
		There should be $0<s<c<1$ for the definition to make sense. $1-c$ is called the completeness error\footnote{We note that for tasks for which no protocol could achieve perfect completeness or perfect soundness, the definitions of completeness error and soundness error might change.}. $s$ is called soundness or soundness error.\par
The discussion above does not consider the efficiency. For many primitives we want the overall protocol to be efficient, and we could only achieve soundness against efficient adversaries. Thus there are three requirements for the primitive:
\begin{itemize}
	\item (Completeness) Same as above.
	\item (Soundness) For any efficient malicious server, the probability of ``the task is not achieved correctly and the client does not catch the server cheating'' is at most $s$.
	\item (Efficiency) The honest opeartion of the whole protocol runs in polynomial time.
\end{itemize}
For many problems, it's not sufficient to simply use the security definitions informally described above. In Section \ref{sec:2.4} we will discuss different paradigms for defining the security.
	\subsubsection{Inputs, the security parameter, outputs, initial states}
	In this section we discuss the register set-ups for a cryptographic protocol and its initial states.
	\paragraph{Inputs} A protocol might take several parameters, input values, and input registers as its arguments.\par 
	 A common special parameter for cryptographic protocols is the \emph{security parameter}, which determines the security level of the protocol.\footnote{In addition to the security, we could also relate the completeness error to the security parameter so that when the security parameter increases both the completeness error and soundness error tend to zero.}
	\begin{nota}\label{nota:2.9} In this paper we denote the security parameter as $\kappa$. It is typically given as inputs in the form of $1^\kappa$.\end{nota}
\paragraph{Outputs} As discussed in the previous section, there is a flag in the output which shows whether the client accepts or rejects (in other words, pass/fail, non-abort/abort). Thus for the cryptographic protocols that we care about in this work, the output registers consist of two parts: (1) the registers that hold the outputs of the tasks (for example, preparing a specific state); (2)the flag register.
\begin{nota} 
	The the flag register is denoted as $\bflag$. We denote the projection onto the subspace $\bflag=\fpass$ as $\Pi_{\fpass}$.\par
	In later constructions where there are multiple rounds of calling to some subprotocol, there might be multiple temporary flag registers. We use $\Pi_{\fpass}^{\bflag^{(\leq i)}}$ to denote the projection onto the space that $\bflag^{(1)},\bflag^{(2)}\cdots \bflag^{(i)}$ are all in value $\fpass$.
   \end{nota}
	\paragraph{Initial states} In quantum cryptography, we need to consider initial states that are possibly entangled with the running environment of the protocol. The environment is not touched by the protocol itself (but will be given to the distinguisher, see Section \ref{sec:2.4.2}). Denote the client-side registers as $\bbC$, denote the server-side registers as $\bS$, and denote the register corresponding to the environment as $\bbE$, the initial state could be described as $\rho_{\text{in}}\in \tD(\cH_{\bbC}\otimes\cH_{\bS}\otimes\cH_{\bbE})$. For the client-server setting (as described in Section \ref{sec:2.2.3}), when there is no client-side input states, we do not need to consider the client-side part of the initial state; then the initial state could be described as $\rho_{\text{in}}\in \tD(\cH_{\bS}\otimes\cH_{\bbE})$. 
	\paragraph{Paper organization using set-ups} In this paper we organize the information of parameters, registers etc in a latex environment called Set-up. This helps us organize protocol descriptions more clearly.
	\subsubsection{Cryptographic assumptions}
	For many problems, it's not possible to construct a protocol and prove its security unconditionally; we need to rely on \emph{cryptographic assumptions}. One commonly-used and powerful cryptographic assumption is the Learning-with-Errors (LWE) assumption \cite{regevLWE}.
	\begin{defn}[LWE]\label{defn:lwe}
		Suppose $n=n(\kappa),m=m(\kappa),q=q(\kappa)$ are integers, $\chi=\chi(\kappa)$ is a probability distribution. The $\text{LWE}_{n,m,q,\chi}$ assumption posits that, when
		$$\mathbf{A}\leftarrow \bZ_q^{n\times m},\mathbf{s}\leftarrow \bZ_q^n,\mathbf{e}\leftarrow \chi^m,\mathbf{u}\leftarrow \bZ_q^m$$
		for any BQP distinguisher\footnote{It also makes sense to consider non-uniform distinguishers like BQP/qpoly.}, it's hard to distinguish
		$$(\mathbf{A},\mathbf{As}+\mathbf{e})$$
		from
		$$(\mathbf{A},\mathbf{u}).$$
		The $\text{LWE}_{n,q,\chi}$ assumption posits that the $\text{LWE}_{n,m,q,\chi}$ assumption holds for any $m=\fpoly(\kappa)$.
	\end{defn}
	A typical choice of $\chi$ is a suitable discretization of the Gaussian distribution.\par
	Another important cryptographic assumption is the quantum random oracle model (QROM) \cite{QRO}. In this model we work in a setting where there is a global random oracle. This assumption is also used in the studies of RSP problems \cite{jiayu20,cvqcinlt}; in this work we do not work in this model.
	\subsection{Security Definition Paradigms}\label{sec:2.4}
	In this section we discuss different paradigms for defining securities. Especially, we review the simulation-based paradigm, which will be used in this work.
	\subsubsection{State family and the indistinguishability}
	\begin{conv}We first note that when we work on the security, we need to consider a family of states parameterized by the security parameter. More explicitly, consider the family of states $(\rho_{\kappa})_{\kappa\in \bN}$. Then we could use $\rho$ to denote the operator in this state family corresponding to $\kappa\in \bN$. For notation simplicity, we do not explicitly write down the state family and simply work on $\rho$.\footnote{Note that this notation convention is similar to how we express Definition \ref{defn:lwe}: in that definition we use $m$ to denote the value of function $m(\kappa)$ at $\kappa$.}\end{conv}
	Below we review the notion of the indistinguishability, which is the foundation of many security definitions.
	\begin{nota}\label{nota:2.12}
		We write $\rho\approx_{\epsilon}^{ind:\cF}\sigma$ if $\forall D\in \cF,$ $$|\Pr[D(\rho)\rightarrow 1]-\Pr[D(\sigma)\rightarrow 1]|\leq \epsilon+\fneg(\kappa).$$ We write $\rho\approx^{ind}_{\epsilon}\sigma$ when $\cF$ is taken to be all the BQP algorithms. We omit $\epsilon$ when $\epsilon=\fneg(\kappa)$.
	 \end{nota}
	 The $D$ in the definition above is called the distinguisher.\par
	\subsubsection{The simulation-based paradigm}\label{sec:2.4.2}
	In this section we introduce the simulation-based paradigm for security definitions.\par
	In this paradigm, we compare the ``real world'' (or called the real protocol) with the ``ideal world'' (or called the ideal functionality) in the security definition. \begin{itemize}\item The real world is the real execution of the protocol itself. \item The ideal functionality achieves the tasks that we aim at as a trusted third party. So the security in the ideal world itself is simply by definition.\end{itemize}
	 For the simulation-based paradigm, intuitively we want to say, each attempt of attacking the real protocol corresponds to an attempt on the ideal functionality; this ``attempt on the ideal functionality'' is described by a simulator. Since the ideal functionality is ideally secure by definition, we could be satisfied once we achieve the indistinguishability between the real protocol and the ideal functionality.\par
	Let's make the security parameter implicit. Then the simulation-based security roughly goes as follows.\par
	For any polynomial time quantum adversary $\fAdv$, there exists a polynomial time quantum simulator $\fSim$ such that:
	\begin{equation}\label{eq:5}\tReal^{\fAdv}\approx^{ind} \underbrace{\fSim}_{\text{take the role of $\fAdv$}}\circ\quad\tIdeal\end{equation}
	Note that $\fSim$ might contain multiple phases, for preparing inputs and processing outputs of $\tIdeal$.\par 
	Then note that \eqref{eq:5} is about the indistinguishability of operations (instead of solely the indistinguishability of states as defined in Notation \ref{nota:2.12}): it means that the initial state is adversarily chosen for distinguishing the output states. Furthermore, we need to consider the case that the initial state is entangled with some environment: although this part is not affected by the real protocol or the ideal functionality, the distinguisher has access to it.\footnote{Omitting this part in the modeling could lead to missing of important attacks in some problems \cite{QBCreview}.}\par
	Recall that in this work we work in a client-server setting where the server could be malicious; let's expand \eqref{eq:5} in this setting. Below we suppose for the initial states the server-side registers are denoted as $\bS$, and the environment is denoted as $\bbE$. For the output registers, denote the server-side output register as $\bQ$. Then the security definition goes as follows.\par
	For any polynomial time quantum adversary $\fAdv$, there exists a polynomial time quantum simulator $\fSim$ such that for any initial state $\rho_{0}\in \tD(\cH_{\bS}\otimes\cH_{\bbE})$:
	\begin{equation}\label{eq:6}\tReal^{\fAdv}(\rho_{0})\approx^{ind} (\underbrace{\fSim}_{\text{preparing inputs for Ideal and generating states on }\bS,\bQ}\circ\tIdeal)(\rho_{0})\end{equation}
	\paragraph{Variants}We discuss some variants of the definition.
	\begin{itemize}
		\item \eqref{eq:6} is stated with negligible distinguishing advantage. In later sections we will need to work on non-negligible distinguishing advantage.
	\item We could consider different uniformity on the preparation of initial states. This is related to a notion called ``advice'' and we refer to \cite{AroraBarak} for its background.
	\item A variant that is less desirable but still useful is the inefficient simulator version, where the simulator is not required to be efficient. Typically a protocol with security under inefficient simulation could be sequentially composed with information-theoretically secure protocols, but in general the inefficient simulator might break other parts of the protocol in the security analysis.
	\item If stand-alone transcript registers are used explicitly (see Section \ref{sec:2.2.3}), the simulator should also simulate the transcripts.
	\end{itemize}
	\paragraph{An example of the definition of the ideal functionality}We note that in the client-server setting that we consider, a frequently used template of the ideal functionality is as follows.
	\begin{exmp}\label{exmp:2.1} The Ideal takes a bit $b\in \{0,1\}$ on the server side as part of the inputs and writes $\fpass/\ffail$ on a client-side register $\bflag$ as part of the outputs. \begin{itemize}\item If $b=0$, Ideal sets $\bflag$ to be $\fpass$ and prepares the outputs as in the honest behavior. \item If $b=1$, Ideal sets $\bflag$ to be $\ffail$ and leave all the other output registers empty.\end{itemize}
	Intuitively this means that either both parties get the expected outputs in the honest setting, or the server is caught cheating.
	\end{exmp}
	\paragraph{Different modes of protocols} In this work and several other works \cite{MIPstar,Grilo17}, we need to deal with a pair of protocols under the same set-up, which are called two modes; typically it includes a test mode and a comp mode. The client does not tell the server which mode it will execute, so the adversary for these two modes could be assumed to be the same operation. The security is typically as follows: if the adversary passes the test mode protocol with high probability, the comp mode will achieve the desired security tasks (which may be defined using a simulation-based notion).
	\subsubsection{Other security definition paradigms}\label{sec:2.4.3}
	There are also other types of security definition paradigms. For example:
	\begin{itemize}
		\item Game-based security: in this paradigm there are a challenger and a distinguisher; the challenger samples a challenge bit and gives the distinguisher different states depending on the challenge bit. The distinguisher tries to predict the challenge bit. The security is roughly defined as follows: the distinguisher could not predict the challenge bit with probability non-negligibly higher than $\frac{1}{2}$.
		\item Framework for universal composability: a protocol with the simulation-based security does not necessarily remain secure when it's composed arbitrarily (for example, concurrently) with other protocols. The universal composability framework \cite{UC} and the AC framework \cite{AC} give a stronger security definition which allows universal composability. 
		\item It's also possible to define securities by the unpredictability of certain strings or some information-theoretic properties.
	\end{itemize}
	Which paradigms or definitions to use could depend on the problems and our goals; stronger security properties are desirable on their own but might be hard to achieve. Depending on what we need, there could be many choices of security definitions even for the same problem.
	\section{Remote State Preparation with Verifiability: Definitions, Variants and Properties}\label{sec:3}
	In this section we formalize the definitions of RSPV, compare it with other variants and other RSP security notions, and study its basic properties.
	\begin{itemize}
		\item In Section \ref{sec:3.1}, we formalize the notion of RSPV. We focus on the random RSPV, which means, the states are sampled (instead of chosen by the client) from a finite set of states.
		\item In Section \ref{sec:3.2} we discuss variants of security definitions.
		\item In Section \ref{sec:3.3} we discuss variants of the problem modeling. For example, we discuss a variant of RSPV where the states are chosen by the client (instead of sampled during the protocol).
		\item In Section \ref{sec:3.4} we study the sequential composability of RSPV.
		\item In Section \ref{sec:3.5} we study the security amplification of RSPV. We define a notion called preRSPV and show that this primitive could be amplified to an RSPV.
	\end{itemize}
	\subsection{Definitions of RSPV}\label{sec:3.1}	
	 Below we consider a random RSPV that samples uniformly from a tuple of states.\par
	 Let's first formalize the register set-up of an RSPV protocol.
	\begin{setup}[Set-up for an RSPV protocol]\label{setup:1}
	Consider a tuple of states $(\ket{\varphi_1},\ket{\varphi_2}\cdots \ket{\varphi_D})$.\par
	The protocol takes two parameters: security parameter $1^\kappa$ and approximation error parameter $1^{1/\epsilon}$.\par
	The output registers are as follows: \begin{itemize}\item the server-side quantum register $\bQ$ whose dimension is suitable for holding the states; \item the client-side classical register $\bD$ with value in $[D]$;\item the client-side classical register $\bflag$ with value in $\{\fpass,\ffail\}$.\end{itemize}\par
	For modeling the initial states in the malicious setting, assume the server-side registers are denoted by register $\bS$, and assume the environment is denoted by register $\bbE$.
	\end{setup}
 As discussed in Section \ref{sec:2.3.1}, there are three requirements for an RSPV protocol: completeness, soundness, efficiency. Below we formalize them one by one.
\begin{defn}[Completeness of RSPV]\label{defn:3.1} We say a protocol under Set-up \ref{setup:1} is $\mu$-complete if when the server is honest, the output state of the protocol is $\mu$-close to the following state:
	\begin{equation}\label{eq:pretar}\underbrace{\ket{\fpass}\bra{\fpass}}_{\bflag}\otimes\underbrace{\rho_{tar}}_{\bD,\bQ}\end{equation}
	where 
	\begin{equation}\label{eq:tar}\rho_{tar}=\sum_{i\in [D]}\frac{1}{D}\underbrace{\ket{i}\bra{i}}_{\bD}\otimes \underbrace{\ket{\varphi_i}\bra{\varphi_i}}_{\bQ}.	
	\end{equation}
	And we simply say the protocol is complete if it has completeness error $\fneg(\kappa)$.	
\end{defn}
The efficiency of the protocol is defined in the normal way. Below we formalize the soundness based on the paradigm in Section \ref{sec:2.4.2}.
	
		 \begin{defn}[Soundness of RSPV]\label{defn:3.2} We say a protocol $\pi$ under Set-up \ref{setup:1} is $\epsilon$-sound if:\par
		  For any efficient quantum adversary $\fAdv$, there exists an efficient quantum operation $\fSim$ such that for any state $\rho_{0}\in \tD(\cH_{\bS}\otimes \cH_{\bbE})$:
		\begin{equation}\label{eq:sim}\Pi_{\fpass}(\pi^{\fAdv}(\rho_{0}))\approx^{ind}_{\epsilon} \underbrace{\Pi_{\fpass}}_{\text{on }\bflag}(\underbrace{\fSim}_{\text{on }\bS,\bQ,\bflag}(\underbrace{\rho_{tar}}_{\bD,\bQ}\otimes \underbrace{\rho_{0}}_{\bS,\bbE}))\end{equation}
	\end{defn} 
	Here we only aim at approximate soundness (where the approximation error might be inverse polynomial). This is because in RSPV problems achieving negligible approximation error remains open and seems very difficult; and existing works all aim at inverse-polynomial errors \cite{GVRSP,qfactory,cvqcinlt}.\par
	We note that Definition \ref{defn:3.2} is subtlely different from the discussion of simulation-based paradigm in Section \ref{sec:2.4.2} in terms of how to treat the failing space. We choose Definition \ref{defn:3.2} because we feel that it's easier to work on; below we discuss this difference in more detail.
	\subsubsection{Handling the failing space within Ideal}\label{sec:3.2.1}
	A difference of Definition \ref{defn:3.2} from equation \eqref{eq:6} is how we treat the failing (or called aborting, rejection) space. In Definition \ref{defn:3.2} the simulator is responsible for writing the passing/failing choices to the $\bflag$ register, which is slightly different from the usual form of the simulation-based paradigm where the passing/failing flag is handled by the the ideal functionality (see Section \ref{sec:2.4.2}). Then in Section \ref{sec:2.4.2} the simulator simulates both the passing space and the failing space, while in Definition \ref{defn:3.2} only the indistinguishability on the passing space is considered. Below we formalize the variant of the RSPV security that follows the usual form of the simulation-based paradigm.
	\begin{defn}[Variant of the soundness of RSPV]\label{defn:3.3}
	Under Set-up \ref{setup:1}, define RSPVIdeal as follows:
	\begin{enumerate}
\item RSPVIdeal takes a classical bit $b\in \{0,1\}$ from the server side. \begin{itemize}\item If $b=0$: it sets $\bflag$ to be $\fpass$  and prepares $\rho_{tar}$ on $\bD,\bQ$.\item If $b=1$: it sets $\bflag$ to be $\ffail$.\end{itemize}
	\end{enumerate}
Then the soundness of RSPV could be defined as follows. We say a protocol $\pi$ is $\epsilon$-sound if:\par
For any efficient quantum adversary $\fAdv$, there exist efficient quantum operations $\fSim=(\fSim_0,\fSim_1)$ such that for any state $\rho_{0}\in \tD(\cH_{\bS}\otimes \cH_{\bbE})$:
\begin{equation}\label{eq:9}
	\pi^\fAdv(\rho_{0})\approx^{ind}_{\epsilon}\underbrace{\fSim_1}_{\text{on }\bS,\bQ}(\text{RSPVIdeal}(\underbrace{\fSim_0}_{\text{on }\bS,\text{ generate $b\in \{0,1\}$ as the input of RSPVIdeal}}(\rho_{0})))
\end{equation}
\end{defn}
We compare Definition \ref{defn:3.3} with Definition \ref{defn:3.2} as follows. We first compare the strength of the definitions, and compare the convenience of usage.\par
\paragraph{Comparison of strength} First, we note that Definition \ref{defn:3.3} implies Definition \ref{defn:3.2}. Starting from \eqref{eq:9}, we construct $\fSim$ that satisfies \eqref{eq:sim} as follows: $\fSim$ executes $\fSim_0$, gets $b$, and executes $\fSim_1$; then it sets $\bflag$ to be $\fpass$ if $b=0$ and sets $\bflag$ to be $\ffail$ if $b=1$.\par
But as far as we know, Definition \ref{defn:3.2} does not seem to imply Definition \ref{defn:3.3} directly (at least in an easy way). Intuitively, \eqref{eq:sim} only says the passing space is simulated by $\fSim$, and there is no guarantee that the failing space is also simulated by the same simulator. It's possible to construct another simulator that simulates the failing space but it's not clear how to combine them into a single simulator directly. (If we are working in the classical world, since the initial states could always be cloned, it's much easier to combine these two simulators into one simulator that works for both cases; but in the quantum world since $\rho_{0}$ could not be cloned the simulation might destroy the initial states and the rewinding is not always possible, it's not clear how to go from equation \eqref{eq:sim} to equation \eqref{eq:9}.)\par
However, we will show that, for RSPV problem under Set-up \ref{setup:1}, once we have an RSPV protocol with soundness under Definition \ref{defn:3.2}, we could amplify it to a new protocol with soundness under Definition \ref{defn:3.3}. The amplification is by a simple sequential repetition and cut-and-choose procedure: the same protocol is repeated for many rounds and one round is randomly chosen as the outputs. However, we note that this reduction will not work when there is additional client-side inputs (for example, some variants of RSPV discussed in Section \ref{sec:3.3.1}); the amplification relies on the fact that the random RSPV does not have any client-side or server-side inputs other than the parameters.\par
The amplification protocol is given below.
\begin{mdframed}[backgroundcolor=black!10]
	Below we assume $\pi$ is an RSPV under Set-up \ref{setup:1} that is complete and $\epsilon_0$-sound under Definition \ref{defn:3.2}.
\begin{prtl}[Soundness amplification from Definition \ref{defn:3.2} to Definition \ref{defn:3.3}]\label{prtl:0}An RSPV protocol under Set-up \ref{setup:1} is constructed as follows.\par
	Parameters: approximation error parameter $1^{1/\epsilon}$, security parameter $1^\kappa$. It is required that $\epsilon>\epsilon_0$.\par
	Output registers: client-side classical registers $\bflag^{(\tout)},\bD^{(\tout)}$, server-side quantum register $\bQ^{(\tout)}$.
	\begin{enumerate}
		\item Define $L=\frac{216}{(\epsilon-\epsilon_0)^3}$.\par
		For each $i\in [L]$:\begin{enumerate}
		\item The client executes $\pi$ with the server. Store the outputs in registers $\bflag^{(i)},\bD^{(i)},\bQ^{(i)}$.
		\end{enumerate}
		\item The client randomly chooses $i\in [L]$.\par
		The client sets $\bflag^{(\tout)}$ to be $\ffail$ if any one of $\bflag^{(1)},\bflag^{(2)}\cdots\bflag^{(i)}$ is $\ffail$. Otherwise it sets $\bflag^{(\tout)}$ to be $\fpass$
		 and sends $i$ to the server. Both parties use the states in $\bD^{(i)},\bQ^{(i)}$ as the outputs.
	\end{enumerate}
\end{prtl}
\end{mdframed}
The completeness and efficiency are from the protocol; below we state and prove the soundness.
\begin{thm}\label{thm:3.1.1thm}
Protocol \ref{prtl:0} is $\epsilon$-sound under Definition \ref{defn:3.3}.
\end{thm}
\begin{proof}
	Consider an adversary $\fAdv$. Denote the initial state as $\rho_0$, and denote the output state by the end of the $i$-th round of step 1 as $\rho_i$. Then by the soundness of $\pi$ (Definition \ref{defn:3.2}) we get that, for any $i\in [L]$, there exists an efficient simulator $\fSim_i$ such that:
	\begin{equation}\label{eq:18r00}
		\Pi_{\fpass}^{\bflag^{(\leq i)}}(\rho_{i})\approx_{\epsilon_0+\fneg(\kappa)}^{ind}\Pi_{\fpass}^{\bflag^{(i)}}(\underbrace{\fSim_i}_{\text{on }\bS,\bQ^{(i)},\bflag^{(i)}}(\underbrace{\rho_{tar}}_{\bD^{(i)},\bQ^{(i)}}\otimes \Pi_{\fpass}^{\bflag^{(\leq i-1)}}(\rho_{i-1})))
	\end{equation}
	Define $S_{\text{low pass}}$ as the set of $i$ such that $\tr(\Pi_{\fpass}^{\bflag^{(\leq i)}}(\rho_{i}))\leq (1-\frac{1}{6}(\epsilon-\epsilon_0))\tr(\Pi_{\fpass}^{\bflag^{(\leq i-1)}}(\rho_{i-1}))$. Then we note that:
	\begin{itemize}\item For each $i\not\in S_{\text{low pass}}$, \begin{equation}\label{eq:neweq1}\tr(\Pi_{\ffail}^{\bflag^{(i)}}(\pi^{\fAdv_i}(\Pi_{\fpass}^{\bflag^{(\leq i-1)}}(\rho_{i-1}))))=\tr(\Pi_{\fpass}^{\bflag^{(\leq i-1)}}(\rho_{i-1}))-\tr(\Pi_{\fpass}^{\bflag^{(\leq i)}}(\rho_{i}))\leq \frac{1}{6}(\epsilon-\epsilon_0).\end{equation}
		\item For each $i\in S_{\text{low pass}}$ such that\footnote{Recall that $|\cdot|$ denotes the size of a set.} $|S_{\text{low pass}}\cap [i]|\geq 36/(\epsilon-\epsilon_0)^2$, \begin{equation}\label{eq:neweq2}\tr(\Pi_{\fpass}^{\bflag^{(\leq i-1)}}(\rho_{i-1}))\leq \frac{1}{6}(\epsilon-\epsilon_0).\end{equation}
	\end{itemize}
Then the simulator $(\fSim_0,\fSim_1)$ where $\fSim_0$ operates on the server-side of $\rho_0$, and $\fSim_1$ operates on $\text{RSPVIdeal}(\fSim_0(\rho_0))$, is defined as follows (note that we abuse the notation and this $(\fSim_0,\fSim_1)$ is different from the simulator in \eqref{eq:18r00}).\par
$\fSim_0$:
\begin{enumerate}
	\item Sample a random coin $i\leftarrow [L]$. This is for simulating the client's random choice in the second step of Protocol \ref{prtl:1}.
	\item Run $\tilde\pi_{<1.i}$ on $\rho_0$ and get $\tilde\rho_{i-1}$. Here $\tilde\pi_{<1.i}$ denotes the simulated protocol execution of Protocol \ref{prtl:1} until the beginning of the $i$-th round of the first step. Here ``simulated protocol execution'' means that, instead of interacting with the client, the simulator simulates a client on its own. So the difference of $\rho_{i-1}$ and $\tilde\rho_{i-1}$ is only the locations of these ``client-side registers''.
	\item If all the $\bflag$ registers so far have value $\fpass$, set\footnote{Recall Definition \ref{defn:3.3} on the bit $b$ of RSPVIdeal inputs.} $b=0$. If any $\bflag$ registers so far has value $\ffail$, set $b=1$.
\end{enumerate}
$\fSim_1$:
\begin{enumerate}
	\item If $b=0$, run $\fSim_i$.\\ If $b=1$, run $\tilde\pi_{1.i}$ where $\tilde\pi_{1.i}$ is defined as above for simulating the $i$-th rounds of the first step of Protocol \ref{prtl:0}.
	\item Run $\tilde\pi_{>1.i}$. Here $\tilde\pi_{>1.i}$ is defined similarly as above for simulating the $i+1\sim L$ rounds of the first step of Protocol \ref{prtl:0}.
	\item Disgard all the auxiliary registers used to simulate the client-side information.
\end{enumerate}\par
We prove this simulator achieves what we want. We could compare the simulated output states with the output states from the real execution and see where the distinguishing advantage could come from (where the distinguishing advantage refers to the approximation error in the soundness definition; see equation \eqref{eq:9}):
	\begin{itemize}
		\item For $i\not\in S_{\text{low pass}}$:
		\begin{itemize}
		\item Equation \eqref{eq:18r00} itself contributes an error of $\epsilon_0$.
		\item $\tr(\Pi_{\ffail}^{\bflag^{(i)}}(\pi^{\fAdv_i}(\Pi_{\fpass}^{\bflag^{(\leq i-1)}}(\rho_{i-1}))))$ contributes an error of $\frac{\epsilon-\epsilon_0}{6}$ by \eqref{eq:neweq1} on both sides of equation \eqref{eq:9}.
		\item $(\bbI-\Pi_{\fpass}^{\bflag^{(\leq i-1)}})(\rho_{i-1})$ part is simulated perfectly in later simulation.
		\end{itemize}
		\item For $i\in S_{\text{low pass}}$, although we do not know the size of $S_{\text{low pass}}$, we could divide this set into $S_{\text{head}}$ and $S_{\text{tail}}$ where $S_{\text{head}}$ is the set of the first $\frac{36}{(\epsilon-\epsilon_0)^2}$ elements and $S_{\text{tail}}$ is the set of the remaining elements. Below we bound the effect of these two sets on the distinguishing advantage.
		\begin{itemize}
			\item The size of $S_{\text{head}}$ is at most $\frac{36}{(\epsilon-\epsilon_0)^2}$ thus it contributes an error of $\frac{\epsilon-\epsilon_0}{6}$ on both size of equation \eqref{eq:9}.
			\item For $i\in S_{\text{tail}}$, this part contributes at most an error of $\frac{\epsilon-\epsilon_0}{6}$ by \eqref{eq:neweq2} on both size of equation \eqref{eq:9}.
			\item $(\bbI-\Pi_{\fpass}^{\bflag^{(\leq i-1)}})(\rho_{i-1})$ part is simulated perfectly in later simulation.
		\end{itemize}
	\end{itemize}
	Summing them up we get the total approximation error to be bounded by $\epsilon_0+2(\frac{\epsilon-\epsilon_0}{6}+\frac{\epsilon-\epsilon_0}{6}+\frac{\epsilon-\epsilon_0}{6})\leq \epsilon$.
\end{proof}
\paragraph{Comparison on the convenience of usage}Definition \ref{defn:3.3} is more standard: it follows the usual form of the simulation-based security (as in Section \ref{sec:2.4.2}). The fact that it seems stronger than Definition \ref{defn:3.2} could also be considered as an advantage. But Definition \ref{defn:3.2} is simpler to work on when we build new RSPV protocols on smaller protocols: we only need to work on a single simulator modeled as a superoperator instead of two.  
\paragraph{Convention} Starting from Section \ref{sec:3.3}, we will mainly use Definition \ref{defn:3.2}. 
	\subsection{Variants of the Security Definition}\label{sec:3.2}
	
	\subsubsection{The rigidity-based soundness}\label{sec:3.2.2}
	A popular variant of the RSPV security that is subtlely different from the simulation-based soundness is the \emph{rigidity-based soundness}. Roughly speaking, the rigidity-based soundness says the output state, after going through an isometry on the server side, is close to the target state. 
	\begin{defn}[Rigidity-based soundness of RSPV] 
		We say a protocol $\pi$ under Set-up \ref{setup:1} is $\epsilon$-sound under rigidity-based soundness if:\par
		  For any efficient quantum adversary $\fAdv$, there exist a efficient isometry $V$ and\footnote{We note that this definition is slightly different from the (rigidity-based) definitions in existing works \cite{GVRSP,GMP}. In \cite{GVRSP,GMP} the left hand side of \eqref{eq:10} is statistically close to a state in the form of $\sum_i\ket{\varphi_i}\bra{\varphi_i}\otimes \sigma_i$, and then a computational indistinguishability requirement is put on $\sigma_i$ for different $i$. We argue that our global indistinguishability captures the same intuition and is more general; what's more, a suitable formulation of variants of definitions in \cite{GVRSP,GMP} should imply this definition.\par
		  One obstacle of showing a direct implication from the definitions in \cite{GVRSP,GMP} is on the efficient simulatable property on these $\sigma_i$: the definitions of rigidity-based soundness used in \cite{GVRSP,GMP} does not seem to imply it could be written in the form of $\fSim(\rho_{0})$. A more careful analysis of the relations between this definition and existing rigidity definitions remains to be done and is out of the scope of this work.} an efficient quantum operation $\fSim$ such that for any state $\rho_{0}\in \tD(\cH_{\bS}\otimes \cH_{\bbE})$:
		\begin{equation}\label{eq:10}\Pi_{\fpass}(\underbrace{V}_{\text{on }\bS,\bQ}(\pi^{\fAdv}(\rho_{0})))\approx^{ind}_{\epsilon} \Pi_{\fpass}(\underbrace{\rho_{tar}}_{\bD,\bQ}\otimes \underbrace{\fSim}_{\text{on }\bS,\bflag}(\underbrace{\rho_{0}}_{\bS,\bbE}))\end{equation}
	\end{defn}
	\paragraph{Comparison with the simulation-based soundness} We compare the intuition of simulation-based soundness and rigidity-based soundness as follows.
	\begin{itemize}
		\item An interpretation of the rigidity-based soundness is ``if the protocol passes, the server really gets the state''.
		\item An interpretation of the simulation-based soundness is ``if the protocol passes, what the server gets is no more than the state''.
	\end{itemize}
	And we could prove the simulation-based soundness as defined in Definition \ref{defn:3.2} is no stronger than the rigidity-based soundness defined above:
	\begin{thm}
		Suppose a protocol $\pi$ under Set-up \ref{setup:1} is $\epsilon$-sound under rigidity-based soundness, then it's also $\epsilon$-sound under simulation-based soundness.
	\end{thm}
	\begin{proof}
		By the rigidity-based soundness we get $V$, $\fSim$ that satisfies \eqref{eq:10}. Then taking $$\fSim^\prime(\underbrace{\cdot}_{\rho_{tar}}\otimes \underbrace{\cdot}_{\rho_{0}})=V^\dagger(\underbrace{\cdot}_{\rho_{tar}}\otimes \fSim(\underbrace{\cdot}_{\rho_{0}}))$$ as the simulator in \eqref{eq:sim} completes the proof.
	\end{proof}
	But the inverse does not seem to be true. Actually, the rigidity-based soundness of RSPV is not even resilient to an additional empty timestep (that is, no party does anything) at the end of the protocol: the adversary could destroy everything in the end to violate the rigidity requirement. For comparison, the simulation-based notion has such resilience: the state destroying operation could be absorbed into the simulator in \eqref{eq:sim}. Arguably this also means the simulation-based notion has more well-behaved properties.\par
	Although the rigidity-based notion seems stronger, intuitively the simulation-based version is as useful as the rigidity-based version in common applications of RSPV. When we construct cryptographic protocols, what we are doing is usually to enforce that the malicious parties could not do undesirable attacks. In the simulation-based soundness it is certified that what the adversary gets is no more than the target state, which intuitively means that it is at least as secure as really getting the target state.
	\subsubsection{Other subtleties}
	 We refer to Section \ref{sec:2.4.2} for discussions on simulator efficiency and the modeling of transcripts.\par
	 There are other choices of uniformity when we model the inputs and initial states. In Definition \ref{defn:3.2} we consider all initial states $\rho_{0}\in \tD(\cH_{\bS}\otimes \cH_{\bbE})$, which implicitly taking the quantum advice into consideration. Alternatively, we could require that all the inputs and initial states are uniformly efficiently preparable, so there is no advice string or state; or we could require the initial quantum states to be efficiently preparable but allow the classical advice. In this work we feel that allowing the quantum advice is more convenient and intuitive since we need to work on quantum states that are not necessarily efficiently preparable (for example, ground states of Hamiltonians).
	\subsubsection{Other security notions}
	\paragraph{Basis blindness}
	Besides the RSP with verifiability, there is also a security notion for RSP called \emph{basis blindness}. This is defined in \cite{qfactory} as a security notion for RSP for $\ket{+_\theta}$ states. Intuitively it means that, no malicious server could predict the two less-significant bits of $\theta$; alternatively it could be formulated as a game-based security notion, where the distinguisher needs to distinguish the two less-significant bits from uniformly random bits. An RSP with this type of security does not work as a protocol compiler (for replacing the transmissions of quantum states with cheaper resources) in general; but it could indeed be used to compile certain protocols securely, especially the universal blind quantum computation protocol \cite{UBQC,jiayu20,RSPVImp}.
	\paragraph{Universal composability}A stronger security requirement is the universal composability, as mentioned in Section \ref{sec:2.4.3}. Unfortunately, as shown in \cite{RSPVImp}, it's not possible to construct RSP protocol with this type of security. But this does not rule out the possibility of using weaker security notions, like Definition \ref{defn:3.2} or any other notion mentioned above. We note that although Definition \ref{defn:3.2} does not have universal composability, it still satisfies sequential composability (see Section \ref{sec:3.4}). 
	\paragraph{Unpredictability of keys}For certain state families, we could also formalize the security as unpredictability of certain strings. For example, consider the RSP for states in the form of $\frac{1}{\sqrt{2}}(\ket{x_0}+\ket{x_1})$, $x_0,x_1\in \{0,1\}^n$. A useful security notion is to require that $x_0||x_1$ is unpredictable (where $||$ is the concatenation notation). Then we could use a primitive called NTCF to construct such an RSP (see \cite{BCMVV} for details). As another example, \cite{jiayu20} makes heavy use of this type of security notions (possibly in some advanced forms) as tools.
	\paragraph{Other RSP notions}There are also RSP works with other types of security notions. For example, \cite{shmueli22} studies the remote preparation of quantum money states; the final state satisfies unclonability and verifiability for quantum money states. Note that the verifiability for quantum money states is different from the verifiability in RSPV.
	\subsubsection{Purified joint states}\label{sec:3.2.4}
	As a security proof technique, in the studies of RSP we often need to work on the purified joint states. In all the previous discussions we work on density operators. In quantum information density operators could be purified into pure states by introducing a reference register $\bR$ (see Fact \ref{fact:3}); then superoperators could be equivalently replaced by the corresponding unitaries.\par 
	In more detail, we could first purify the initial state $\rho_{0}$ in Definition \ref{defn:3.2} as a pure state on $\cH_{\bS}\otimes \cH_{\bbE}\otimes\cH_{\bR}$. Then notice that we could enlarge $\bbE$ to contain $\bR$, which gives the distinguisher access to more registers; and this does not weaken the overall statement because indistinguishability on more registers implies indistinguishability on parts of them. Thus to prove Definition \ref{defn:3.2} it's sufficient to only consider pure initial states in $\cH_{\bS}\otimes \cH_{\bbE}$.\par
	Then we could do purification for each classical register appeared in the protocol. There are two ways to do this purification:
	\begin{itemize}
		\item Assume the classical registers are purified by corresponding reference registers.
		\item Simply replace the classical registers and classical states by quantum registers and quantum states that look the same with the original states on computational basis; and then put a requirement that operators operated on them could only read the values (that is, CNOT the contents to other registers) but should not revise the values of these registers. When this requirement is satisfied, the purified state is indistinguishable to the original cq-state.
	\end{itemize}
	Working on purified joint states enables some proof techniques that are hard to use or even unavailable when we work on density operators. Different security proofs could use different ways of purification.\par
	As examples, \cite{B15} purifies the client-side classical information to reduce the original protocol to a entanglement-based protocol; \cite{jiayu20} uses heavily ``state decomposition lemmas'' and the triangle inequality, which works on purified joint states. In this work we will use this type of technique to prove Lemma \ref{prop:5.3} (see Section \ref{sec:5.3.2}), which requires decomposing and combining components of the purified joint state.
	\subsection{Variants and Generalizations of the Problem Modeling}\label{sec:3.3}
	In the previous sections we focus on random RSPV for a fixed size state family; in this section we discuss its variants and generalizations. In Section \ref{sec:3.3.1} we formalize the RSPV notions where the state families are parameterized by problem size parameters (instead of a fixed size) or where the states are chosen by the client (instead of being sampled randomly), and discuss the notion where the states are prepared from some initial resource states (instead of being generated from scratch). In Section \ref{sec:3.3.2} we discuss a generalization that allows us to prepare arbitrary efficiently preparable states (instead of some specific state families).
	\subsubsection{Problem size parameters, sampling versus choosing and resource states}\label{sec:3.3.1}
	\paragraph{Problem size parameters}
	In the previous sections we work on a tuple of states of a fixed size (see Set-up \ref{setup:1}). In the later sections we will need to work on state families parameterized by a size parameter. As an example, we need to work on states in $\{\frac{1}{\sqrt{2}}(\ket{0}\ket{x_0}+\ket{1}\ket{x_1}):x_0,x_1\in \{0,1\}^n\}$; here the problem size parameter is $n$. Below we formalize the set-up and the security definition, which are adapted from Set-up \ref{setup:1} by adding the problem size parameter.
	\begin{setup}\label{setup:2}
		Consider a state family parameterized by $1^n$: denoted it as $S=(S_n)_{n\in \bN}$, where each $S_n$ is a set of states; suppose the set of their descriptions is $D_n$, and suppose $S_n=\{\ket{\varphi_k}\}_{k\in D_n}$.\par
		The protocol takes the following parameters: \begin{itemize}\item problem size parameter $1^n$;\item security parameter $1^\kappa$;\item approximation parameter $1^{1/\epsilon}$.\end{itemize}\par
		Corresponding to problem size $1^n$, the output registers are as follows:\begin{itemize}\item the server-side quantum register $\bQ$ for holding states in $S_n$;\item the client-side classical register $\bD$ for holding descriptions in $D_n$;\item the client-side classical register $\bflag$ with value in $\{\fpass,\ffail\}$.\end{itemize}\par
		For modeling the initial states in the malicious setting, assume the server-side registers are denoted by register $\bS$, and assume the environment is denoted by register $\bbE$.
	\end{setup}
	\begin{defn}Corresponding to problem size parameter $1^n$, define the target state $\rho_{tar}\in \tD(\cH_{\bD}\otimes\cH_{\bQ})$ as
	\begin{equation}\label{eq:tar2}\rho_{tar}=\sum_{k\in D_n}\frac{1}{\tsize(D_n)}\underbrace{\ket{k}\bra{k}}_{\bD}\otimes \underbrace{\ket{\varphi_k}\bra{\varphi_k}}_{\bQ}.	
	\end{equation}
	And the completeness is defined similarly to Definition \ref{defn:3.1}.
\end{defn}
	\begin{defn}
		We say a protocol $\pi$ under Set-up \ref{setup:2} is $\epsilon$-sound if for any $n=\fpoly(\kappa)$, for any efficient quantum adversary, there exists an efficnent quantum operation $\fSim$ such that for any state $\rho_{0}\in \tD(\cH_{\bS}\otimes \cH_{\bbE})$:
		\begin{equation}\label{eq:sim2}\Pi_{\fpass}(\pi^{\fAdv}(1^n)(\rho_{0}))\approx^{ind}_{\epsilon} \underbrace{\Pi_{\fpass}}_{\text{on }\bflag}(\underbrace{\fSim(1^n)}_{\text{on }\bS,\bQ,\bflag}(\underbrace{\rho_{tar}}_{\bD,\bQ\text{ as given in equation \eqref{eq:tar2}}}\otimes \underbrace{\rho_{0}}_{\bS,\bbE}))\end{equation}
	\end{defn}
	Note that in \eqref{eq:sim2} we make $1^n$ explicit and the other two parameters are implicit as before.\par
	In later sections we also need to consider multiple problem size parameters (for example, $1^m,1^n$), and the definitions could be adapted correspondingly. We could also consider problem sizes that satisfy some conditions (for example, $m$ is bigger than some functions of $n$), and the security definition could also be adapted correspondingly.
	\paragraph{RSPV where the client chooses the states}
	So far we focus on the \emph{random} RSPV, where the states are sampled randomly by the protocol. We could also study a form of RSPV notion where the states are chosen by the client.\par
	Below we formalize this notion for a tuple of states $(\ket{\varphi_1},\ket{\varphi_2}\cdots \ket{\varphi_D})$. 
	\begin{setup}[Set-up for an RSPV where the client chooses the state]\label{setup:3}
		As in Set-up \ref{setup:1}, consider a tuple of states $(\ket{\varphi_1},\ket{\varphi_2}\cdots \ket{\varphi_D})$.\par
		The protocol takes the security parameter $1^\kappa$ and approximation error parameter $1^{1/\epsilon}$ as in Set-up \ref{setup:1}.\par
		 The protocol takes a classical input $i\in [D]$ from the client.\par
		Compared to Set-up \ref{setup:1}, the output registers are solely $\bflag$ and $\bQ$.\par
		$\bS,\bbE$ are defined in the same way as Set-up \ref{setup:1}.
	\end{setup}
	\begin{defn}
		When the input is $i\in [D]$, the target state is simply $\ket{\varphi_i}\bra{\varphi_i}\in \tD(\cH_{\bQ})$.\par
		Then the completeness is defined similarly to Definition \ref{defn:3.1}:\par
		We say a protocol under Set-up \ref{setup:3} is complete if for any input value $i\in [D]$, when the server is honest, the output state of the protocol is negligibly close to the following state:
		$$\underbrace{\ket{\fpass}\bra{\fpass}}_{\bflag}\otimes\underbrace{\ket{\varphi_i}\bra{\varphi_i}}_{\bQ}$$
	\end{defn}
	To formalize the soundness, we first define the ideal functionality, and formalize the soundness using the simulation-based paradigm (see Section \ref{sec:2.4.2}).\footnote{Note that, as discussed in Section \ref{sec:3.2.1}, the simulate-only-passing-space simplification used in Definition \ref{defn:3.2} does not apply to the case where there are client-side inputs, which is the case here.}
	\begin{defn}\label{defn:choseninputs}
		Define $\text{RSPVIdeal}$ as follows:\par
		Ideal takes $i\in [D]$ from the client and $b\in \{0,1\}$ from the server. Then depending on the value of $b$:
		\begin{itemize}
			\item If $b=0$: it sets $\bflag$ to be $\fpass$ and prepares $\ket{\varphi_i}\bra{\varphi_i}$ on $\bQ$.
			\item If $b=1$: it sets $\bflag$ to be $\ffail$.
		\end{itemize}
		We say a protocol under Set-up \ref{setup:3} is $\epsilon$-sound if:\par 
		For any efficient quantum adversary $\fAdv$, there exist efficient quantum operations $\fSim_0,\fSim_1$ such that for any input value $i\in [D]$,
	 any state $\rho_{0}\in \tD(\cH_{\bS}\otimes \cH_{\bbE})$:
	  \begin{equation}\label{eq:sim3}\pi^{\fAdv}(\underbrace{i}_{\text{client-side input}})(\rho_{0})\approx^{ind}_{\epsilon} \underbrace{\fSim_1}_{\text{on }\bS,\bQ}(\text{RSPVIdeal}(\underbrace{\fSim_0}_{\text{on }\bS,\text{ generate $b\in \{0,1\}$ as the input of RSPVIdeal}}(\underbrace{\rho_{0}}_{\bS,\bbE})))\end{equation}
	\end{defn}
	\paragraph{Adding resource states} In later sections, we will need to consider the following variants: the protocol does not prepare states from scratch; instead the server holds some initial resource states while the client knows their descriptions, and the overall protocol could be seen as a transformation from the resource states to the target states. To formalize such a notion the set-up, completeness and soundness need to be adapted to include these resources states. We will encounter a notion in this form in Section \ref{sec:3.5.3}.\par
	Finally we note that we could consider combinations of these variants; the set-ups, completeness and soundness are defined in the natural way.
	\subsubsection{RSPV for arbitrary efficiently preparable states}\label{sec:3.3.2}
	In the previous sections we focus on RSPV for some specific state families. We would like to get RSPV protocols for more and larger state families, and ideally, we could like to have a universal RSPV protocol for arbitrary efficiently preparable states. 
	There are two notions for this problem that seem natural. These two notions roughly go as follows:
	\begin{itemize}
		\item We could consider an arbitrary efficient quantum circuit $C$ that prepares a cq-state on $\tD(\cH_{\bD}\otimes\cH_{\bQ})$, and use this state as the target state. $C$ is given as a public input string.
		\item We could consider an arbitrary efficient quantum circuit $C$ that prepares a quantum state on $\tD(\cH_{\bQ})$, where $C$ is chosen by the client (in other words, $C$ is a client-side input string). 
	\end{itemize}
	We leave the formalizations and concrete constructions for further studies.

	\subsection{Sequential Composition Property of RSPV}\label{sec:3.4}
	\subsubsection{Composition between RSPV protocols}\label{sec:3.4.1}
	First we could prove RSPV (under the simulation-based soundness as defined in Definition \ref{defn:3.2}) has a natural composition property when different RSPV protocols are sequentially composed with each other. 
	\begin{mdframed}[backgroundcolor=black!10]
	Below we work under Set-up \ref{setup:1}. Suppose $\pi_1$ is an RSPV for target state $\rho_{tar}$ and is $\epsilon_1$-sound, $\pi_2$ is an RSPV for target state $\sigma_{tar}$ and is $\epsilon_2$-sound.
	\begin{prtl}[Sequential composition between different RSPV protocols]\label{prtl:sc}
		Output registers: client-side classical register $\bflag$, client-side classical register  $\bD=(\bD^{(1)},\bD^{(2)})$, server-side quantum register $\bQ=(\bQ^{(1)},\bQ^{(2)})$.
	 \begin{enumerate}
		\item Execute $\pi_1$. Store the outputs in $(\bflag^{(1)},\bD^{(1)},\bQ^{(1)})$. In the honest setting $\bflag^{(1)}$ has value $\fpass$ and $(\bD^{(1)},\bQ^{(1)})$ holds state $\rho_{tar}$.
		\item Execute $\pi_2$. Store the outputs in $(\bflag^{(2)},\bD^{(2)},\bQ^{(2)})$. In the honest setting $\bflag^{(2)}$ has value $\fpass$ and $(\bD^{(2)},\bQ^{(2)})$ holds state $\rho_{tar}$.
		\item The client sets $\bflag$ to be $\ffail$ if any one of $\bflag^{(1)}$, $\bflag^{(2)}$ is $\ffail$. Otherwise it sets $\bflag$ to be $\fpass$.
	 \end{enumerate}
	\end{prtl}
\end{mdframed}
\begin{thm}
	Protocol \ref{prtl:sc} is an RSPV protocol for target state $\rho_{tar}\otimes \sigma_{tar}$ and is $(\epsilon_1+\epsilon_2)$-sound.
\end{thm}
	The completeness and efficiency are from the protocol description. Below we prove the soundness. 
	\begin{proof}
For adversary $\fAdv$, suppose the initial joint state is $\rho_{0}\in \tD(\cH_{\bS}\otimes \cH_{\bbE})$, and the output state of the first step of the protocol (after the execution of $\pi_1$) is $\rho_1$. By the soundness of $\pi_1$ there exists an efficient simulator $\fSim_1$ such that:
\begin{equation}\label{eq:comp1}\Pi_{\fpass}^{\bflag^{(1)}}(\rho_1)\approx^{ind}_{\epsilon_1}\Pi_{\fpass}^{\bflag^{(1)}}(\underbrace{\fSim_1}_{\text{on }\bS,\bQ^{(1)},\bflag^{(1)}}(\rho_{tar}\otimes \rho_0))\end{equation}
Then assume the output state after the second step of the protocol (after the execution of $\pi_2$) is $\rho_2$. Then by the soundness\footnote{We note that to get \eqref{eq:comp2} we apply Definition \ref{defn:3.2} to an unnormalized initial state $\Pi_{\fpass}^{\bflag^{(1)}}(\rho_1)$. This is fine because we could do a case analysis on the trace of $\Pi_{\fpass}^{\bflag^{(1)}}(\rho_1)$ before applying Definition \ref{defn:3.2}: when $\tr(\Pi_{\fpass}^{\bflag^{(1)}}(\rho_1))<\epsilon_2/2$ \eqref{eq:comp2} is automatically true, while when $\tr(\Pi_{\fpass}^{\bflag^{(1)}}(\rho_1))\geq \epsilon_2/2$ we could first rescale the initial state before applying Definition \ref{defn:3.2}.} of $\pi_2$ there exists an efficient simulator $\fSim_2$ such that:
\begin{equation}\label{eq:comp2}\Pi_{\fpass}^{\bflag^{(2)}}\Pi_{\fpass}^{\bflag^{(1)}}(\rho_2)\approx^{ind}_{\epsilon_2}\Pi_{\fpass}^{\bflag^{(2)}}(\underbrace{\fSim_2}_{\text{on }\bS,\bQ^{(1)},\bQ^{(2)},\bflag^{(2)}}(\sigma_{tar}\otimes \Pi_{\fpass}^{\bflag^{(1)}}(\rho_1)))\end{equation}
Let's construct the overall simulator $\fSim$. $\fSim$ is defined as follows:
\begin{enumerate}
	\item Execute $\fSim_1$ on registers $\bS$, $\bQ^{(1)}$ and $\tilde\bflag^{(1)}$. (Note that this simulation does not have access to $\bflag^{(1)}$, but we could use a temporary register $\tilde\bflag^{(1)}$ as a replacement.)
	\item Execute $\fSim_2$ on registers $\bS$, $\bQ^{(1)}$, $\bQ^{(2)}$ and $\bflag^{(2)}$. (Note that this simulation does not have access to $\bflag^{(2)}$, but we could use a temporary register $\tilde\bflag^{(2)}$ as a replacement.)
	\item Set $\bflag$ to be $\ffail$ if any one of $\tilde\bflag^{(1)}$, $\tilde\bflag^{(2)}$ is $\ffail$. Otherwise set $\bflag$ to be $\fpass$.
\end{enumerate}
Combining \eqref{eq:comp1}\eqref{eq:comp2} we get that
$$\Pi_{\fpass}^{\bflag}(\rho_2)\approx^{ind}_{\epsilon_1+\epsilon_2}\Pi_{\fpass}^{\bflag}(\underbrace{\fSim}_{\text{on }\bS,\bQ,\bflag}(\sigma_{tar}\otimes\rho_{tar}\otimes\rho_0))$$
which completes the proof.
	\end{proof}
	\subsubsection{RSPV as a protocol compiler}\label{sec:3.4.2}
	One important application of RSPV is to serve as a protocol compiler: given a protocol starting with ``the client prepares and sends some quantum states'', we could replace this step with an RSPV protocol for these states, and the security property of the overall protocol will be preserved (approximately).\par
	In more detail, assume a protocol $\pi$ has the following set-up: the client holds a classical register $\bD$, the server holds a quantum register $\bQ$, and in the honest setting initially the client and the server should hold the state $\rho_{tar}\in \tD(\cH_{\bD}\otimes \cH_{\bQ})$.\par
	Then suppose $\pi_{\text{RSPV}}$ is an RSPV protocol for target state $\rho_{tar}$ and is $\epsilon$-sound. Then consider the following protocol:
	\begin{enumerate}
		\item Execute $\pi_{\text{RSPV}}$. Break out if the protocol fails.
		\item Execute $\pi$. 
	\end{enumerate}
	Then by the RSPV soundness there is
	$$\Pi_{\fpass}(\pi^{\fAdv_2}(\pi_{\text{RSPV}}^{\fAdv_1}(\rho_{0})))\approx_{\epsilon}^{ind}\Pi_{\fpass}(\pi^{\fAdv_2}(\fSim(\rho_{tar}\otimes \rho_{0})))$$
	which means that we could replace the honest set-up of initial states by a call to the RSPV protocol and the properties of the final states will be approximately preserved. 
	\subsection{PreRSPV and Security Amplification}\label{sec:3.5}
	In this section we discuss the soundness amplification of RSPV. We will define a notion called preRSPV and shows that it could be amplified to an RSPV. The name and the notion of preRSPV have already been used in \cite{cvqcinlt} and here we re-formalize this notion. In more detail:
	\begin{itemize}\item In Section \ref{sec:preRSPVdef} we first introduce the notion of preRSPV. The preRSPV is defined as a pair of protocols $(\pi_{\ttest},\pi_{\tcomp})$ that share the same set-up. The soundness is defined intuitively as follows: if $\pi_{\ttest}$ running against an adversary $\fAdv$ could pass with high probability, then $\pi_{\tcomp}$ running against the same adversary $\fAdv$ will satisfy the RSPV soundness (that is, simulated by the target state).\par
		Then a preRSPV defined above could be amplified to an RSPV by a \emph{cut-and-choose} procedure: for each round, the client randomly chooses to execute either $\pi_{\ttest}$ or $\pi_{\tcomp}$ with some probabilities (without telling the server the choices). If the server could keep passing in many rounds, then for most of the $\pi_{\tcomp}$ rounds the target state is prepared as expected.
		\item In Section \ref{sec:3.5.2} we formalize a variant of preRSPV where the client counts \emph{scores} for the server besides the pass/fail flag. The set-up is as follows: in the end of the protocol the client will record a win/lose decision in a $\bscore$ register, and the honest server could achieve the optimal winning probability. We use $\OPT$ to denote the optimal winning probability. (Note that when $\OPT=1$ the notion downgrades to the notion in Section \ref{sec:preRSPVdef}.) The soundness is roughly defined as follows: if $\pi_{\ttest}$ running against an adversary $\fAdv$ could win with probability close to $\OPT$, then $\pi_{\tcomp}$ running against the same adversary $\fAdv$ will satisfy the RSPV soundness.\par
		In the amplification procedure, the client needs to count the number of winning rounds and see whether the total score is close to the expected value of an honest execution.
		\item In Section \ref{sec:3.5.3} we formalize a temporary variant of preRSPV where there are initial resource states (instead of preparing states from scratch), which is used in later sections.
	\end{itemize}
	In the constructions in later sections, we often need to first construct a preRSPV, and then amplify it to an RSPV. (See Section \ref{sec:5.3}, \ref{sec:5.5}.)
	\subsubsection{PreRSPV}\label{sec:preRSPVdef}
Let's formalize the preRSPV notion. In this section we work in a set-up of Set-up \ref{setup:1}.
\begin{setup}[Set-up of preRSPV]\label{setup:4rr}
	The set-up is the same as Set-up \ref{setup:1} except that we are considering a pair of protocols $(\pi_{\ttest},\pi_{\tcomp})$ in this set-up instead of solely $\pi$.
\end{setup}
Below we formalize the completeness and soundness.
\begin{defn}[Completeness of preRSPV]\label{defn:3.9}
	We say $(\pi_{\ttest},\pi_{\tcomp})$ under Set-up \ref{setup:1} is complete if:
	\begin{itemize}
		\item In $\pi_{\ttest}$, when the server is honest, the passing probability (the trace of the output state projected onto $\Pi_{\fpass}$) is negligibly close to $1$.
		\item In $\pi_{\tcomp}$, when the server is honest, the output state of the protocol is negligibly close to 
		$$\underbrace{\ket{\fpass}\bra{\fpass}}_{\bflag}\otimes\underbrace{\rho_{tar}}_{\bD,\bQ}$$
		as described in equation \eqref{eq:pretar}\eqref{eq:tar}.
	\end{itemize}
\end{defn}
\begin{defn}[Soundness of preRSPV]\label{defn:3.10}
	We say $(\pi_{\ttest},\pi_{\tcomp})$ under Set-up \ref{setup:1} is $(\delta,\epsilon)$-sound if:\par
	For any efficient quantum adversary $\fAdv$, there exists an efficient quantum operation $\fSim$ such that for any state $\rho_{0}\in \tD(\cH_{\bS}\otimes \cH_{\bbE})$:\par
	If
	$$\tr(\Pi_{\fpass}(\pi_{\ttest}^{\fAdv}(\rho_{0})))\geq 1-\delta$$
	then
	$$\Pi_{\fpass}(\pi^{\fAdv}_{\tcomp}(\rho_{0}))\approx^{ind}_{\epsilon} \underbrace{\Pi_{\fpass}}_{\text{on }\bflag}(\underbrace{\fSim}_{\text{on }\bS,\bQ,\bflag}(\underbrace{\rho_{tar}}_{\bD,\bQ}\otimes \underbrace{\rho_{0}}_{\bS,\bbE}))$$
\end{defn}
	\paragraph{Amplification to RSPV} As discussed before, a preRSPV defined above could be amplified to an RSPV. 
	\begin{mdframed}[backgroundcolor=black!10]
		Below we assume $(\pi_{\ttest},\pi_{\tcomp})$ is a preRSPV under Set-up \ref{setup:1} that is complete and $(\delta,\epsilon_0)$-sound.
	\begin{prtl}[Amplification from preRSPV to RSPV]\label{prtl:1}An RSPV protocol under Set-up \ref{setup:1} is constructed as follows.\par
		Parameters: approximation error parameter $1^{1/\epsilon}$, security parameter $1^\kappa$. It is required that $\epsilon>\epsilon_0$.\par
		Output registers: client-side classical registers $\bflag^{(\tout)},\bD^{(\tout)}$, server-side quantum register $\bQ^{(\tout)}$.
		\begin{enumerate}
			\item Define $L=\frac{512}{\delta(\epsilon-\epsilon_0)^3}$, $p=\frac{\epsilon-\epsilon_0}{8}$.\par
			For each $i\in [L]$:\begin{enumerate}
				\item The client randomly chooses $\text{mode}^{(i)}=\ttest$ with probability $p$ and $\text{mode}^{(i)}=\tcomp$ with probability $1-p$.
			\item The client executes $\pi_{\text{mode}^{(i)}}$ with the server. Store the outputs in registers $\bflag^{(i)},\bD^{(i)},\bQ^{(i)}$.
			\end{enumerate}
			\item The client sets $\bflag^{(\tout)}$ to be $\ffail$ if any one of $\bflag^{(i)}$ is $\ffail$. Otherwise it sets $\bflag^{(\tout)}$ to be $\fpass$.\par
			The client randomly chooses $i\in [L]$ such that $\tmode^{(i)}=\tcomp$ and sends $i$ to the server. Both parties use the states in $\bD^{(i)},\bQ^{(i)}$ as the outputs.
		\end{enumerate}
	\end{prtl}
\end{mdframed}
	Then we show the protocol above is an RSPV. The completeness and efficiency are from the protocol; below we prove the soundness.
\begin{thm}\label{thm:3.3}
Protocol \ref{prtl:1} is $\epsilon$-sound.
\end{thm}
The proof has some similarities with the proof of Theorem \ref{thm:3.1.1thm}.
\begin{proof}
Consider an adversary $\fAdv$. Denote the initial state as $\rho_0$, and denote the output state by the end of the $i$-th round of step 1 as $\rho_i$. Then by the soundness of preRSPV we get, for any $i\in [L]$, there exists an efficient simulator $\fSim_i$ such that if 
\begin{equation}\label{eq:17r}
	\tr(\Pi_{\fpass}^{\bflag^{(i)}}(\pi_{\ttest}^{\fAdv_i}(\Pi_{\fpass}^{\bflag^{(\leq i-1)}}(\rho_{i-1}))))>(1-\delta)\tr(\Pi_{\fpass}^{\bflag^{(\leq i-1)}}(\rho_{i-1}))
\end{equation}
then
\begin{equation}\label{eq:18r}
	\Pi_{\fpass}^{\bflag^{(i)}}(\pi_{\tcomp}^{\fAdv_i}(\Pi_{\fpass}^{\bflag^{(\leq i-1)}}(\rho_{i-1})))\approx_{\epsilon_0}^{ind}\Pi_{\fpass}^{\bflag^{(i)}}(\fSim_i(\underbrace{\rho_{tar}}_{\bD^{(i)},\bQ^{(i)}}\otimes \Pi_{\fpass}^{\bflag^{(\leq i-1)}}(\rho_{i-1})))
\end{equation}
	Define $S_{\text{low pass}}$ as the set of $i$ such that \eqref{eq:17r} does not hold (thus \eqref{eq:18r} holds for the complement of $S_{\text{low pass}}$). Then we note that when the size of $S_{\text{low pass}}$ is $\geq \frac{8}{p\delta(\epsilon-\epsilon_0)}$ the overall passing probability is $\leq \frac{\epsilon-\epsilon_0}{8}$.\par
	The simulator $\fSim$ applied on $(\rho_{tar}\otimes \rho_0)$ is defined as follows:
	\begin{enumerate}
		\item Sample a random coin $i\leftarrow [L]$. This is for simulating the client's random choice in the second step of Protocol \ref{prtl:1}.
		\item Run $\tilde\pi_{<1.i}$ on $\rho_0$ and get $\tilde\rho_{i-1}$. Here $\tilde\pi_{<1.i}$ denotes the simulated protocol execution of Protocol \ref{prtl:1} until the beginning of the $i$-th round of the first step. Here ``simulated protocol execution'' means that, instead of interacting with the client, the simulator simulates a client on its own. So the difference of $\rho_{i-1}$ and $\tilde\rho_{i-1}$ is only the locations of these ``client-side registers''.
		\item Run $\fSim_i$ on $\rho_{tar}\otimes \tilde\rho_{i-1}$ where $\rho_{tar}\in \tD(\cH_{\bD^{(\tout)}}\otimes\cH_{\bQ^{(\tout)}})$.
		\item Run $\tilde\pi_{>1.i}$ on $\fSim_i(\rho_{tar}\otimes \tilde\rho_{i-1})$. Here $\tilde\pi_{>1.i}$ is defined as above.
		\item Set $\bflag^{(\tout)}$ to be $\ffail$ if any of the flag or simulated flag is $\ffail$; otherwise set $\bflag^{(\tout)}$ to be $\fpass$. Disgard all the auxiliary registers for simulating the client side.
	\end{enumerate}
	\par
	We prove this simulator achieves what we want. We could compare the simulated output states with the output states from the real execution and see where the distinguishing advantage could come from (see also the proof of Theorem \ref{thm:3.1.1thm}):
	\begin{itemize}
		\item In the original protocol $i$ is not sampled randomly but on all the $comp$ rounds, while in the simulation $i$ is sampled randomly from $[L]$. But in the original protocol with high probability the number of $test$ rounds is at most $2pL$ so this part contributes an error of at most $\frac{\epsilon-\epsilon_0}{8}$ on both size of equation \eqref{eq:sim}.
		\item When $i\not\in S_{\text{low pass}}$, \eqref{eq:18r} holds, which gives the indistinguishabilty for the simulated execution and real execution. Equation \eqref{eq:18r} itself contributes an error of $\epsilon_0$.
		\item The $i$-th round is not necessarily a $comp$ round, but when we simulate the $i$-th round we simply assume it's a $comp$ round for simulating it using \eqref{eq:18r}. This contributes an error of $p$ on both sides of equation \eqref{eq:sim}.
		\item When $i\in S_{\text{low pass}}$ \eqref{eq:18r} does not hold. Although we do not know the size of $S_{\text{low pass}}$, we could divide this set into $S_{\text{head}}$ and $S_{\text{tail}}$ where $S_{\text{head}}$ is the set of the first $\frac{8}{p\delta(\epsilon-\epsilon_0)}$ elements and $S_{\text{tail}}$ is the set of the remaining elements. Below we bound the effect of these two sets on the distinguishing advantage.
		\begin{itemize}
			\item The size of $S_{\text{head}}$ is at most $\frac{8}{p\delta(\epsilon-\epsilon_0)}$ thus it contributes an error of $\frac{8}{Lp\delta(\epsilon-\epsilon_0)}$ on both size of equation \eqref{eq:sim}.
			\item For $i\in S_{\text{tail}}$ the passing probability is no more than $\frac{\epsilon-\epsilon_0}{8}$ thus this part contributes at most an error of $\frac{\epsilon-\epsilon_0}{8}$ on both size of equation \eqref{eq:sim}.
		\end{itemize}
	\end{itemize}
	Summing them up we get the total approximation error to be bounded by $\epsilon_0+2(\frac{\epsilon-\epsilon_0}{8}+p+\frac{8}{Lp\delta(\epsilon-\epsilon_0)}+\frac{\epsilon-\epsilon_0}{8})\leq \epsilon$.
\end{proof}
\subsubsection{PreRSPV with the score}\label{sec:3.5.2}
As discussed before, we need to consider a variant of preRSPV where the client counts scores to certify the server's operation. There will be a $\bscore$ register and a optimal winning probability $\OPT$ in the set-up of this notion. Below we formalize the set-up, completeness and soundness.
\begin{setup}[Set-up for a preRSPV with the score]\label{setup:4}
	Compare to Set-up \ref{setup:4rr}, we consider a real number $\OPT\in [0,1]$.\par
	Then we consider an additional client-side classical register $\bscore$ with value in $\{\fwin,\flose,\perp\}$. $\pi_{\ttest}$ outputs a value in $\{\fwin,\flose\}$ into the $\bscore$ register; $\pi_{\tcomp}$ does not touch the $\bscore$ register and leave it in the default value $\perp$.\par The other parts of the set-up are the same as Set-up \ref{setup:4rr}.
\end{setup}
\begin{defn}\label{defn:3.11}
	We say $(\pi_{\ttest},\pi_{\tcomp})$ under Set-up \ref{setup:4} is complete if:
	\begin{itemize}
		\item In $\pi_{\ttest}$, when the server is honest, the passing probability (the trace of the output state projected onto $\Pi_{\fpass}$) is negligibly close to $1$ and the winning probability (the trace of the output state projected onto $\Pi_{\fpass}\Pi_{\fwin}^{\bscore}$) is negligibly close to $\OPT$.
		\item In $\pi_{\tcomp}$, when the server is honest, the output state of the protocol is negligibly close to 
		$$\underbrace{\ket{\fpass}\bra{\fpass}}_{\bflag}\otimes\underbrace{\ket{\perp}\bra{\perp}}_{\bscore}\otimes\underbrace{\rho_{tar}}_{\bD,\bQ}$$
		as described in equation \eqref{eq:pretar}\eqref{eq:tar}.
	\end{itemize}
\end{defn}
\begin{defn}\label{defn:3.12r}
For $(\pi_{\ttest},\pi_{\tcomp})$ under Set-up \ref{setup:4}, we say  winning probability $\OPT$ is $(\delta,\lambda)$-optimal if:\par
For any efficient quantum adversary $\fAdv$, for any state $\rho_{0}\in \tD(\cH_{\bS}\otimes \cH_{\bbE})$, at least one of the following two is true:
\begin{itemize}
	\item (Low passing)
	\begin{equation}\label{eq:ifopt}\tr(\Pi_{\fpass}(\pi_{\ttest}^{\fAdv}(\rho_{0})))\leq 1-\delta+\fneg(\kappa)\end{equation}
	\item (Optimal winning)
	\begin{equation}\label{eq:thenopt}\tr(\Pi_{\fpass}\Pi_{\fwin}^{\bscore}(\pi_{\ttest}^{\fAdv}(\rho_{0})))\leq \OPT+\lambda+\fneg(\kappa)\end{equation}
\end{itemize}
\end{defn}
\begin{defn}\label{defn:3.12}
	We say $(\pi_{\ttest},\pi_{\tcomp})$ under Set-up \ref{setup:4} is $(\delta,\epsilon)$-sound if:\par
	For any efficient quantum adversary $\fAdv$, there exists an efficient quantum operation $\fSim$ such that for any state $\rho_{0}\in \tD(\cH_{\bS}\otimes \cH_{\bbE})$, at least one of the following three is true:
	\begin{itemize}
		\item (Low passing) $$\tr(\Pi_{\fpass}(\pi_{\ttest}^{\fAdv}(\rho_{0})))\leq 1-\delta+\fneg(\kappa)$$
	\item (Low winning) $$\tr(\Pi_{\fpass}\Pi_{\fwin}^{\bscore}(\pi_{\ttest}^{\fAdv}(\rho_{0})))\leq \OPT-\delta$$
	\item (Simulation)
	$$\Pi_{\fpass}(\pi^{\fAdv}_{\tcomp}(\rho_{0}))\approx^{ind}_{\epsilon} \underbrace{\Pi_{\fpass}}_{\text{on }\bflag}(\underbrace{\fSim}_{\text{on }\bS,\bQ,\bflag}(\underbrace{\rho_{tar}}_{\bD,\bQ}\otimes \underbrace{\rho_{0}}_{\bS,\bbE}))$$\end{itemize}
\end{defn}
\paragraph{Variants of Definition \ref{defn:3.12r}, \ref{defn:3.12}}We note that there are also other ways of expressing the optimality of OPT and the soundness; for example, we could only remove the \eqref{eq:ifopt} part (and only keep the \eqref{eq:thenopt} part) and remove the parameter $\delta$ from the definition. The choices of these definitions are based on the consideration of keeping a balance between the difficulty of the preRSPV-with-the-score-to-RSPV amplification process and the difficulty of the construction of the preRSPV-with-the-score protocol. (Note that when the soundness/optimality gets weaker the construction gets easier but the amplification gets harder.)
\paragraph{Amplification to RSPV} Similar to Section \ref{sec:preRSPVdef}, the preRSPV under Set-up \ref{setup:4} could also be amplified to an RSPV. Below we do this amplification in two steps: we first amplify a preRSPV under Set-up \ref{setup:4} to a preRSPV under Set-up \ref{setup:4rr} (that is, removing the need for the score in the modeling). Then we use the results in Section \ref{sec:preRSPVdef} to get an RSPV protocol.
\begin{mdframed}[backgroundcolor=black!10]
	Below we assume $(\pi_{\ttest},\pi_{\tcomp})$ is a preRSPV under Set-up \ref{setup:4} that is complete, has $(\delta_0,\lambda)$-optimal winning probability $\OPT$ and is $(\delta_0,\epsilon_0)$-sound.
\begin{prtl}[Amplification from preRSPV under Set-up \ref{setup:4} to preRSPV under Set-up \ref{setup:4rr}]\label{prtl:3r}A preRSPV protocol $(\pi_{\ttest}^\prime,\pi_{\tcomp}^\prime)$ under Set-up \ref{setup:4rr} is constructed as follows.\par
	Parameters: approximation error parameter $1^{1/\epsilon}$, security parameter $1^\kappa$. It is required that $\epsilon>\epsilon_0$ and $\lambda<\frac{1}{6}\delta_0(\epsilon-\epsilon_0)$.\par
	Output registers: client-side classical registers $\bflag^{(\tout)},\bD^{(\tout)}$, server-side quantum register $\bQ^{(\tout)}$.\par
	Define $L=4\kappa/(((\frac{1}{6}\delta_0(\epsilon-\epsilon_0)-\lambda))^2(\epsilon-\epsilon_0))$, $\text{threshold}=(\OPT-\frac{1}{2}(\frac{1}{6}\delta_0(\epsilon-\epsilon_0)-\lambda))L$.\par
	$\pi_{\ttest}^\prime$ is defined as follows:
	\begin{enumerate}
		\item For each $i\in [L]$:\begin{enumerate}\item The client executes $\pi_{\ttest}$ with the server. Store the outputs in registers $\bflag^{(i)},\bscore^{(i)},\bD^{(i)},\bQ^{(i)}$.\end{enumerate}
		\item The client sets $\bflag^{(\tout)}$ to be $\ffail$ if any one of $\bflag^{(i)}$ is $\ffail$. Then the client counts the number of $\fwin$ in all the $\bscore$ registers and sets $\bflag^{(\tout)}$ to be $\ffail$ if the total number of $\fwin$ is less than $\text{threshold}$. Otherwise it sets $\bflag^{(\tout)}$ to be $\fpass$.
	\end{enumerate}\par
	$\pi_{\tcomp}^\prime$ is defined as follows:
	\begin{enumerate}
		\item The client randomly chooses $i_{stop}\in [L]$.\par
		For each $i\in [i_{stop}-1]$:\begin{enumerate}
		 \item The client executes $\pi_{\ttest}$ with the server. Store the outputs in registers $\bflag^{(i)},\bscore^{(i)},\bD^{(i)},\bQ^{(i)}$.\end{enumerate}
		\item The client sets $\bflag^{(\tout)}$ to be $\ffail$ if any $\bflag^{(i)}$ is $\ffail$. Otherwise the client executes $\pi_{\tcomp}$ with the server; store the outputs in registers $\bflag^{(\tout)},\bD^{(\tout)},\bQ^{(\tout)}$.\par
	\end{enumerate}
\end{prtl}
\end{mdframed}
\begin{thm}
The protocol is complete.
\end{thm}
\begin{proof}
	Note that the expected value of the total score in $\pi_{\ttest}^\prime$ is $\OPT\cdot L$. By the Chernoff bounds the probability of ``the total score is less than $(\OPT-\frac{1}{2}(\frac{1}{6}\delta_0(\epsilon-\epsilon_0)-\lambda))L$'' is upper bounded by a negligible function of $\kappa$.\par
 The other sources of not passing are all negligible.
\end{proof}
The efficiency is from the protocol; below we prove the soundness.
\begin{thm}
	Protocol \ref{prtl:3r} is $(\delta,\epsilon)$-sound where $\delta=\frac{1}{2}(\frac{1}{6}\delta_0(\epsilon-\epsilon_0)-\lambda)$.
\end{thm}
\paragraph{Intuition for the proof}The overall structure of the proof has some similarity with the proof of Theorem \ref{thm:3.3}. One main technical obstacle is how to do the probability calculation for the random process of generating the scores: we would like to show that if the server could win in sufficiently big number of tests, in a large portion of the rounds, the probability that it wins should be big --- which implies that the corresponding comp mode should prepare the target state. Formalizing these probability arguments requires some techniques from the probability theory; here we use the Markov inequality\footnote{In general the Markov inequality is loose but here we do not care about the tightness of the calculation. For a tighter analysis we could consider, for example, the Azuma's inequality (which also has more restrictions).} and make careful analysis.
\begin{proof}
	Consider an adversary $\fAdv$. Denote the initial state as $\rho_0$, and denote the output state by the end of the $i$-th round of step 1 as $\rho_i$. 
	Then by the soundness of $(\pi_{\ttest},\pi_{\tcomp})$ we get, for any $i\in [L]$, there exists an efficient simulator $\fSim_i$ such that at least one of the following three is true:
	\begin{itemize}\item (Low passing) \begin{equation}\label{eq:17rrrr}\tr(\Pi_{\fpass}^{\bflag^{(i)}}(\pi_{\ttest}^{\fAdv_i}(\Pi_{\fpass}^{\bflag^{(\leq i-1)}}(\rho_{i-1}))))\leq(1-\delta_0)\tr(\Pi_{\fpass}^{\bflag^{(\leq i-1)}}(\rho_{i-1}))\end{equation} \item (Low winning) \begin{equation}\label{eq:17rr}
	\tr(\Pi_{\fpass}^{\bflag^{(i)}}\Pi_{\fwin}^{\bscore^{(i)}}(\pi_{\ttest}^{\fAdv_i}(\Pi_{\fpass}^{\bflag^{(\leq i-1)}}(\rho_{i-1}))))\leq(\OPT-\delta_0)\tr(\Pi_{\fpass}^{\bflag^{(\leq i-1)}}(\rho_{i-1}))
	\end{equation}
	\item (Simulation)
	\begin{equation}\label{eq:18rr}
		\Pi_{\fpass}^{\bflag^{(i)}}(\pi_{\tcomp}^{\fAdv_i}(\Pi_{\fpass}^{\bflag^{(\leq i-1)}}(\rho_{i-1})))\approx_{\epsilon_0}^{ind}\Pi_{\fpass}^{\bflag^{(i)}}(\fSim_i(\underbrace{\rho_{tar}}_{\bD^{(i)},\bQ^{(i)}}\otimes \Pi_{\fpass}^{\bflag^{(\leq i-1)}}(\rho_{i-1})))
	\end{equation}\end{itemize}
		Define $S_{\text{low pass}}$ as the set of $i$ such that \eqref{eq:17rrrr} holds. By the condition that the overall protocol passes with probability $\geq 1-\delta$, we have that $|S_{\text{low pass}}|\leq 
		\frac{1}{36}(\frac{1}{6}\delta_0(\epsilon-\epsilon_0)-\lambda)L$.\par
		Define $S_{\text{low win}}$ as the set of $i$ such that \eqref{eq:17rr} holds (thus \eqref{eq:18rr} holds for the complement of $S_{\text{low pass}}\cup S_{\text{low win}}$). Below we need to bound the total score on both the $S_{\text{low win}}$ part and its complement. For upper bounding the scores on the complement of $S_{\text{low win}}$ we use the optimality of $\OPT$, which translates to:\par
		For each $i\in [L]-S_{\text{low pass}}$, there is
		\begin{equation}\label{eq:20rr}
			\tr(\Pi_{\fpass}^{\bflag^{(i)}}\Pi_{\fwin}^{\bscore^{(i)}}(\pi_{\ttest}^{\fAdv_i}(\Pi_{\fpass}^{\bflag^{(\leq i-1)}}(\rho_{i-1}))))\leq(\OPT+\lambda)\cdot\tr(\Pi_{\fpass}^{\bflag^{(\leq i-1)}}(\rho_{i-1}))
		\end{equation}
		Below we show the passing probability could not be high under the condition that $|S_{\text{low win}}|\geq \frac{1}{6}(\epsilon-\epsilon_0)L$. The total expectation for the scores in each step is at most $$\underbrace{\frac{1}{6}(\epsilon-\epsilon_0)L\cdot(\OPT-\delta_0)}_{\text{rounds in $S_{\text{low win}}$}}+\underbrace{(L-\frac{1}{6}(\epsilon-\epsilon_0)L)\cdot (\OPT+\lambda)}_{\text{rounds outside $S_{\text{low win}}$}}\leq (\OPT-(\frac{1}{6}\delta_0(\epsilon-\epsilon_0)-\lambda))L$$ which is smaller than the threshold by a significant gap. By the Markov inequality we have the passing probability is no more than $1-\frac{1}{2}(\frac{1}{6}\delta_0(\epsilon-\epsilon_0)-\lambda)/\OPT<1-\delta$, which violates condition of the overall passing probability. In summary, we have proved:
		$$|S_{\text{low win}}|\leq \frac{1}{6}(\epsilon-\epsilon_0)L$$
Let's construct the simulator for the $\pi_{\tcomp}^\prime$ mode and show that it achieves what we need. The simulator $\fSim$ applied on $(\rho_{tar}\otimes \rho_0)$ is defined as follows:
	\begin{enumerate}
		\item Sample a random coin $i_{\text{stop}}\leftarrow [L]$.
		\item Simulate the first $(i_{\text{stop}}-1)$ rounds of $\pi_{\tcomp}^\prime$ to get $\tilde\rho_{i_{\text{stop}}-1}$. 
		Here the simulation is similar to what we did in the proof of Theorem \ref{thm:3.3}; so the difference of $\rho_{i_{\text{stop}}-1}$ and $\tilde\rho_{i_{\text{stop}}-1}$ is only the locations of the ``client-side registers''.
		\item Run $\fSim_i$ on $\rho_{tar}\otimes \tilde\rho_{i_{\text{stop}}-1}$ where $\rho_{tar}\in \tD(\cH_{\bD^{(\tout)}}\otimes\cH_{\bQ^{(\tout)}})$.
		\item Set $\bflag^{(\tout)}$ to be $\ffail$ if any of the flag or simulated flag is $\ffail$; otherwise set $\bflag^{(\tout)}$ to be $\fpass$. Disgard all the auxiliary registers for simulating the client side.
	\end{enumerate}
	We prove this simulator simulates the output of $\pi_{\tcomp}^\prime$. We could compare the simulated output states with the output states from the real execution and see where the distinguishing advantage (or called approximation error) could come from:
	\begin{itemize}
		\item Equation \eqref{eq:18rr} itself contributes an error of $\epsilon_0$.
		\item Equation \eqref{eq:18rr} only holds for indices in the complement of $S_{\text{low pass}}\cup S_{\text{low win}}$, but $i_{stop}$ may still fall within $S_{\text{low pass}}$. This contributes an error of no more than $\frac{1}{6}(\epsilon-\epsilon_0)\times 2$.
		\item Equation \eqref{eq:18rr} only holds for indices in the complement of $S_{\text{low pass}}\cup S_{\text{low win}}$, but $i_{stop}$ may still fall within $S_{\text{low win}}$. This contributes an error of $\frac{1}{6}(\epsilon-\epsilon_0)\times 2$.
	\end{itemize}
	Summing them up completes the proof.
\end{proof}
After getting a preRSPV under Set-up \ref{setup:4rr}, we could further amplify it using Protocol \ref{prtl:1} to get an RSPV.
\subsubsection{A temporary variant of PreRSPV with initial resource states}\label{sec:3.5.3}
In the later sections, as an intermediate notion, we need to consider a variant of preRSPV that starts with some resource states (recall the discussion on resource states in Section \ref{sec:3.3.1}). These states are assumed to be in the honest form in both the honest setting and malicious setting. Below we formalize this variant, as a preparation for later constructions.
\begin{setup}[Variant of Set-up \ref{setup:1} with initial resource states]\label{setup:5r}
	The registers are as follows: the client holds a classical register $\bD^{(\tin)}$ and the server holds a quantum register $\bQ^{(\tin)}$. The client also holds a classical register $\bD^{(\tout)}$ and the server also holds a quantum register $\bQ^{(\tout)}$, which are both renamed from $\bD,\bQ$ in Set-up \ref{setup:1}. We also need to consider a state $\rho_{rs}\in \tD(\cH_{\bD^{(\tin)}}\otimes\cH_{\bQ^{(\tin)}})$ which is used as part of the initial state in in the honest setting. The other parts of the set-up are the same as Set-up \ref{setup:1}.
\end{setup}
The completeness is defined in the same way as Definition \ref{defn:3.9} with one difference: when we consider the execution of $\pi_{\ttest}$ and $\pi_{\tcomp}$ the initial state should be $\rho_{rs}$. The soundness is defined in the same way as Definition \ref{defn:3.10} with the following differences differences: when we consider the execution of $\pi_{\ttest}$ and $\pi_{\tcomp}$ the initial state should be  $\rho_{rs}\otimes\rho_0$; and the registers that the simulator could work on now consist of $\bS,\bQ^{(\tin)},\bQ^{(\tout)},\bflag$.
\begin{setup}[Variant of Set-up \ref{setup:4} with initial resource states]\label{setup:6r}
	Similar to what we did in Set-up \ref{setup:5r}, we consider registers $\bD^{(\tin)}$, $\bQ^{(\tin)}$ in the set-up and rename $\bD,\bQ$ as $\bD^{(\tout)}$, $\bQ^{(\tout)}$.
\end{setup}
The completeness and soundness are adapted similarly.
\paragraph{Compilation to the normal form of preRSPV} We note that once we get a protocol under Set-up \ref{setup:5r}, \ref{setup:6r}, we could compile the protocol following the approach in Section \ref{sec:3.4.2}: we could use an RSPV with $\rho_{rs}$ being the target state to prepare the resource states, then the overall protocol will be a preRSPV under the definitions in Section \ref{sec:preRSPVdef}, \ref{sec:3.5.2}. Then we could apply the results in Section \ref{sec:preRSPVdef}, \ref{sec:3.5.2} to amplify them to an RSPV.
\section{RSPV and Test of a Qubit}\label{sec:r4}
One basic primitive in RSPV is the RSPV for BB84 states. Basically this primitive has been studied in several existing works \cite{GVRSP,qfactory,BGKPV23} but these existing works could use different frameworks and notions from our work so some translation is needed. Existing ways to express the relations between operators or states include:
\begin{itemize}
	\item Characterize the underlying states: if the test passes with high probability, then the states to be tested are as desired. This could further include rigidity-based soundness and simulation-based soundness.
	\item Characterize the measurement operators: if the test passes with high probability, then the operators used by the server in the test are as desired. This further includes:
	\begin{itemize}\item the operators satisfy, for example, anticommutation relations.
	\item the operators are, for example, the Pauli X and Z operators up to an isometry. This is similar to the rigidity-based soundness discussed before.
	\end{itemize}
\end{itemize}
In this section we study the relation between RSPV and a notion called ``test of a qubit'', and study the protocols for RSPV for BB84 states.\par
By \cite{BGKPV23} we know that we could construct a test for a qubit assuming NTCF (which could be adapted to weak NTCF). Although test for a qubit is closely related to RSPV, they are still not the same\footnote{We thank Kaniuar Bacho, James Bartusek, Yasuaki Okinaka and anonymous reviewers for discussions}, so some technical works are needed for translation of results. A test for a qubit is different from an RSPV in the following ways:
\begin{itemize}
	\item The soundness of a test for a qubit is described in terms of anticommutation of operators \cite{BGKPV23}; the soundness of RSPV is described by simulation, defined on states.
	\item In a test for a qubit the qubit is finally tested and destroyed; in RSPV we want to preserve the qubit on the server side.
\end{itemize}
We would like to translate the results in the language of test-of-a-qubit to our RSPV framework. We will not start from scratch, but try to make use of existing results as much as possible. Below we describe our approach in more detail.\par
\begin{itemize}
	\item The test of a qubit protocol could be regarded as a preRSPV-with-score (see Section \ref{sec:3.5.2}), which could be amplified to an RSPV. This solves the second issue above.
	\item We translate the soundness property as follows. Starting from the anticommutation relation, we know the server's operators are actually Pauli X and Pauli Z operators up to an efficiently computable isometry. Then this together with the fact that the server could win the test with close-to-optimal probability implies that the underlying states are as desired. Finally we use the (computational) basis-blindness to prove the indistinguishability of remaining parts of the states which implies the simulation-based soundness.
\end{itemize}
In Section \ref{sec:b1} we first review the results in \cite{BGKPV23}; in Section \ref{sec:b2} we give our construction for RSPV-for-BB84 and prove its soundness.
\subsection{Test of a Qubit}\label{sec:b1}
The protocol in \cite{BGKPV23}, instantiated with weak NTCF, could be described as follows.
\begin{mdframed}[backgroundcolor=black!10]
	\begin{prtl}\label{prtl:14}
		Below we describe the test of a qubit protocol as given in \cite{BGKPV23}.\par
		Parameters: completeness error tolerance parameter $1^{1/\mu}$, soundness error $1^{1/\delta}$, security parameter $1^\kappa$.\par
		\begin{enumerate}
			\item (Phase A) The client and the server perform some protocol.\par
			In the end the joint state between the client and the server in the honest setting is $\mu$-close to the following state: the client-side $\bflag$ register holds value $\fpass$, the client holds $\theta_1,\theta_2\leftarrow_r \{0,1\}^2$, the server holds $\ket{+_{4\theta_1+2\theta_2+1}}$.\par
			Break out from the protocol if the $\bflag$ is $\ffail$.
			\item (Phase B) The client randomly samples $c\leftarrow_r\{0,2\}$ and sends it to the server.\par
			The server measures the state on basis $\ket{+_{c}},\ket{+_{c+4}}$ and send back the measurement result to the client. The client checks the server's response is $u_c(4\theta_1+2\theta_2+1)$ (see Definition \ref{defn:a1}) and record $\bscore$ to be $\fwin$ if the check passes.
		\end{enumerate}
	\end{prtl}
\end{mdframed}
Note that we re-describe the protocol to be consistent with the notations in later sections.\par
By \cite{BGKPV23}, the following holds for Protocol \ref{prtl:14}:
\begin{thm}[Computational basis-blindness]\label{thm:b1}
	Suppose the server-side state in the end of phase A corresponding to client-side value $\theta=4\theta_1+2\theta_2+1$ is $\rho_{\theta}$. Then $\rho_{1}+\rho_5\approx^{\tind}\rho_3+\rho_{7}$.
\end{thm}
\begin{thm}
	Suppose the server passes the protocol with probability $\geq 1-\delta$. Then the winning probability is no more than $\cos^2(\pi/8)+\fpoly(\delta)$. 
\end{thm}
\begin{thm}\label{thm:b2}
	Suppose the server passes the protocol with probability $\geq 1-\delta$ and wins with probability $\geq \cos^2(\pi/8)-\delta$. Suppose $\rho=\sum_{\theta\in \{1,3,5,7\}}\rho_{\theta}$ where $\rho_{\theta}$ is defined in Theorem \ref{thm:b1}. Suppose the server's operations corresponding to $c=0$ and $c=2$ are described by observables $X_0,X_2$. Then
	\begin{equation}\label{eq:b2e}\tr(\{X_0,X_2\}^2\rho)\leq \fpoly(\delta)\end{equation}
\end{thm}
\subsection{RSPV for BB84 States}\label{sec:b2}
Below in Section \ref{sec:a.2.2} and \ref{sec:a.2.3} we first review the lemmas needed for translating test of a qubit to RSPV for BB84 states; then in Section \ref{sec:a.2.4} we construct and prove the soundness of RSPV for BB84 states.
\subsubsection{Extracting the qubit from the anticommutation relation}\label{sec:a.2.2}
From \cite{GVRSP,vidicknotes} we know that the anticommutation implies that the operators could be seen as the Pauli X and Pauli Z operators up to an efficiently computable isometry. Below we review the statement (with a basis rotation for consistency to later sections; note that $\ket{+_0},\ket{+_2},\ket{+_4},\ket{+_6}$ are isometric to BB84 states.)
\begin{thm}[By \cite{vidicknotes}]\label{thm:antiiso}
	Suppose $X_0,X_2,\rho$ satisfies 
	$$\tr(\{X_0,X_2\}^2\rho)\approx_{O(\delta)} 0$$
	Then there exists a quantum isometry $U^{X_0,X_2}$ efficiently computable from $X_0,X_2$ such that
	$$X_0(\rho)\approx_{\fpoly(\delta)} ((U^{X_0,X_2})^\dagger(\ket{+_0}\bra{+_0}-\ket{+_4}\bra{+_4})U^{X_0,X_2})(\rho)$$
	$$X_2(\rho)\approx_{\fpoly(\delta)} ((U^{X_0,X_2})^\dagger(\ket{+_2}\bra{+_2}-\ket{+_6}\bra{+_6})U^{X_0,X_2})(\rho)$$
\end{thm}
Theorem \ref{thm:antiiso} allows us characterize the operators. In the next section we state a lemma that allows us to further characterize the underlying states to be measured (together with the condition that the server wins the test with close-to-optimal probability).

\subsubsection{Lemmas for characterizing the states to be tested}\label{sec:a.2.3}
\begin{lem}\label{lem:a.7}
	For any density operator $\rho$ there is
	$\frac{1}{2}(\tr(\ket{+_0}\bra{+_0}\rho)+\tr(\ket{+_2}\bra{+_2}\rho))\leq \cos^2(\pi/8)$. And if $\frac{1}{2}(\tr(\ket{+_0}\bra{+_0}\rho)+\tr(\ket{+_2}\bra{+_2}\rho))\approx_{\delta} \cos^2(\pi/8)$, there is $\rho\approx_{\fpoly(\delta)}\ket{+_1}\bra{+_1}\otimes\psi$ for some density operator $\psi$.
\end{lem}
The proof is a simple linear algebra calculation (see also the appendix of \cite{cvqcinlt}).\par
The following lemma is for proving the indistinguishability on the remaining part of the states (other than the qubits to be measured). We first state the statistical indistinguishability version of it.
\begin{lem}\label{lem:a.8}
	Suppose $\rho_{\theta^\prime}+\rho_{\theta^\prime+4}\approx_{\delta}\rho_{\theta^{\prime\prime}}+\rho_{\theta^{\prime\prime}+4}$ where $\theta^\prime\neq \theta^{\prime\prime}$. Further suppose $\rho_{\theta}=\ket{+_\theta}\bra{+_\theta}\otimes \psi_{\theta}$ for each $\theta\in \{\theta^\prime,\theta^{\prime\prime},\theta^\prime+4,\theta^{\prime\prime}+4\}$. Then $(\psi_{\theta})_{\theta\in \{\theta^\prime,\theta^{\prime\prime},\theta^\prime+4,\theta^{\prime\prime}+4\}}$ are $\fpoly(\delta)$-close to each other.
\end{lem}
Then we state the computational analog of Lemma \ref{lem:a.8}.
\begin{lem}\label{lem:a.9}
	Suppose $\rho_{\theta^\prime}+\rho_{\theta^\prime+4}\approx_{\delta}^{\tind}\rho_{\theta^{\prime\prime}}+\rho_{\theta^{\prime\prime}+4}$ where $\theta^\prime\neq \theta^{\prime\prime}$. Further suppose $\rho_{\theta}=\ket{+_\theta}\bra{+_\theta}\otimes \psi_{\theta}$ for each $\theta\in \{\theta^\prime,\theta^{\prime\prime},\theta^\prime+4,\theta^{\prime\prime}+4\}$. Then $(\psi_{\theta})_{\theta\in \{\theta^\prime,\theta^{\prime\prime},\theta^\prime+4,\theta^{\prime\prime}+4\}}$ are $\fpoly(\delta)$-indistinguishable to each other.
\end{lem}
Note that, similar to Lemma \ref{lem:a.8} we do not have the condition that these $\rho$ have approximately the same trace value (but this could be proved from the condition).
\subsubsection{Constructions and proofs for RSPV for BB84}\label{sec:a.2.4}
Below we first prove Protocol \ref{prtl:14} could be seen as a preRSPV-with-score.
\begin{lem}\label{lem:b3}
	Consider $(\pi_{\ttest},\pi_{\tcomp})$ where $\pi_{\ttest}$ is the whole protocol of Protocol \ref{prtl:14}, and the computation-mode is the phase A part. Then this protocol is a preRSPV-with-score with optimal winning probability $\cos^2(\pi/8)$ and is $(\delta,\fpoly(\delta))$-sound.
\end{lem}
The optimality of winning probability $\cos^2(\pi/8)$ is proved in \cite{BGKPV23} and does not need translation.
\begin{proof}[Proof of Lemma \ref{lem:b3}]
	First apply Theorem \ref{thm:b2} we know if the server passes the protocol with probability $\geq 1-\delta$ and wins with probability $\geq \cos^2(\pi/8)-\delta$, \eqref{eq:b2e} holds. Apply Theorem \ref{thm:antiiso} we know $X_0,X_2$ are efficiently isometric to $\ket{+_0}\bra{+_0}-\ket{+_4}\bra{+_4}$ and $\ket{+_2}\bra{+_2}-\ket{+_6}\bra{+_6}$. Use the condition that the server wins with probability $\geq \cos^2(\pi/8)-\delta$ again, by Lemma \ref{lem:a.7} we know $\rho_{\theta}\approx_{\fpoly(\delta)}\ket{+_\theta}\bra{+_\theta}\otimes\psi_\theta$ for each $\theta\in \{1,3,5,7\}$. Finally by Lemma \ref{lem:a.9} $\psi_{\theta}$ are approximately indistinguishable to each other, which proves the rigidity-based soundness and implies the simulation-based soundness.
\end{proof}
Then this preRSPV-with-score could be amplified to an RSPV for BB84:
\begin{thm}
Under the same assumptions as Protocol \ref{prtl:14}, there exists an RSPV for BB84 states.
\end{thm}
\begin{proof}
	Apply the amplification described in Section \ref{sec:3.5.2} to the preRSPV-with-score described above completes the proof.
\end{proof}
\subsection{Detailed Description of the Information-theoretic Core for Test of a Qubit}\label{sec:a.3}
In this section we give the missing details when we describe the test of a qubit in Protocol \ref{prtl:14}. Here we describe the information-theoretic core, which will be used in later sections. Note that although \cite{BGKPV23} proves stronger results than this information-theoretic core (in more detail, they only need computational basis-blindness instead of statistical basis-blindness), an explicit discussion of the information-theoretic core makes our work and the analysis later more self-contained and accessible.
\begin{defn}[Repeat of Definition 3.1 in \cite{GVRSP}]\label{defn:a1}
Consider four positive semi-definite operators $(\phi_\theta)_{\theta\in \{1,3,5,7\}}$ and two single-qubit observables $X_0$ and $X_2$. Recall that for single-qubit observables the eigenvalues are $\pm 1$; for each $i\in \{0,2\}$, use $X^{0}_i$ to denote the projection onto the eigenvector of $X_i$ with eigenvalue $+1$, and use $X^{1}_i$ to denote the projection onto the space of eigenvector of $X_i$ with eigenvalue $-1$. Thus $X_i=X^0_i-X^1_i$.\par
 For $\theta\in \{1,3,5,7\}$ let $u_0(\theta),u_2(\theta)\in\{0,1\}$ be functions as follows: \begin{itemize}\item $u_0(\theta) = 0$ if and only if $\theta\in\{1,7\}$;\item $u_2(\theta) = 0$ if and only if $\theta\in\{1,3\}$.\end{itemize} The winning probability is
\begin{equation}\label{eq:67} \frac{1}{4} \sum_{\theta\in\{1,3,5,7\}} \frac{1}{2} \sum_{i\in\{0,2\}}\tr\big(X_{i}^{u_i(\theta)} \phi_\theta\big) \;.\end{equation}
\end{defn}
We give some explanation of the definition. The special choice of indexing is for matching the main protocol in \cite{GVRSP}. $\theta\in \{1,3,5,7\}$ encodes two classical bits, which are $u_0(\theta)$ and $u_2(\theta)$. The correspondence is as follows:
\begin{itemize}
    \item $\theta=1$: $u_0(\theta)=0,u_2(\theta)=0$;
    \item $\theta=3$: $u_0(\theta)=1,u_2(\theta)=0$;
    \item $\theta=5$: $u_0(\theta)=1, u_2(\theta)=1$;
    \item $\theta=7$: $u_0(\theta)=0, u_2(\theta)=1$.    
\end{itemize}
The following states and observables achieve the optimal winning probability $\cos^2(\pi/8)$. From this example we could also get an intuitive interpretation of the test and the indices.
\begin{fact}\label{fact:100}
    Recall
    $$|+_\theta\rangle=\frac{1}{\sqrt{2}}(|0\rangle+e^{\mi\theta}|1\rangle),\theta\in\{0,1,2\cdots 7\}$$
    Then under the following $(\phi_\theta)_{\theta\in \{1,3,5,7\}}$ and $X_0,X_2$, the winning probability is $\cos^2(\pi/8)$:
    $$\phi_\theta=|+_\theta\rangle\langle +_\theta|,\quad \forall \theta\in \{1,3,5,7\}$$
    $$X_0^0:=\text{ projection onto }|+_0\rangle$$
    $$X_0^1:=\text{ projection onto }|+_4\rangle$$
    $$X_2^0:=\text{ projection onto }|+_2\rangle$$
    $$X_2^1:=\text{ projection onto }|+_6\rangle$$
    Recall for all $i\in \{0,2\}$, $X_i=X^0_i-X^1_i$.
\end{fact}
\paragraph{Comparison to \cite{GVRSP,qfactory}} First we change several changes on notations compared to \cite{GVRSP}; for example, compared to Definition 3.1 in \cite{GVRSP}, we make it explicit that $u_0$, $u_2$ are functions from $\{1,3,5,7\}$ to $\{0,1\}$. Then we note that in \cite{GVRSP} the self-testing property is based on the trace-1 condition on the initial states; but we feel that it's more convenient to formalize the self-testing property based on the basis-blindness property, which is similar to the approaches in \cite{qfactory} (which is called \emph{blind self-testing} in \cite{qfactory}).\par
\begin{lem}\label{lem:a.5}
	Let $(\phi_\theta)_{\theta\in\{1,3,5,7\}}$ be a tuple of positive semidefinite operators. For simplicity define $N=\tr(\frac{1}{4}\sum_{\theta\in \{1,3,5,7\}}\phi_{\theta})$. Suppose $(\phi_\theta/\tr(\frac{1}{4}\sum_{\theta\in \{1,3,5,7\}}\phi_{\theta}))_{\theta\in\{1,3,5,7\}}$ are information-theoretic basis-blind (that is, $\phi_1+\phi_5\approx_{\fneg(\kappa)N}\phi_3+\phi_7$.)  Consider Definition \ref{defn:a1} defined on this tuple of states. Then \eqref{eq:67} is no more than ${N}(\cos^2(\pi/8)+\fneg(\kappa))$. What's more, if \eqref{eq:67} is $\delta N$-close to ${N}(\cos^2(\pi/8))$, then 
	\[ \sum_{\theta\in\{1,3,5,7\}} \tr\big( \{X_0,X_{2}\}^2 \phi_\theta \big) \,=\,(O({\delta})+\fneg(\kappa))\cdot N\;.\]
\end{lem}
\paragraph{On the proof of Lemma \ref{lem:a.5}}One way to prove Lemma \ref{lem:a.5} is to reduce it to the self-testing property of the CHSH game in the non-local game setting; below we call the two prover Alice and Bob and the verifier is simply named as verifier. The self-testing property of the CHSH game could be found in many existing works \cite{NZ23}.\par
To do the reduction, imagine Bob is holding $\phi_{\theta}\in \{\phi_1,\phi_5\}$ and Alice is holding $u_2(\theta)$. (Then the Bob-side's state, after tracing out the Alice's system, is $\phi_1+\phi_5$). By the statistical basis-blindness, we know there exists a local operation on Alice's side such that after this operation Bob is holding $\phi_{\theta}\in \{\phi_3,\phi_7\}$ and Alice is holding $u_2(\theta)$. Then the verifier's question to Alice is whether or not to apply this local operation before the measurement, which will collapse Bob's system to either $\{\phi_1,\phi_5\}$ or $\{\phi_3,\phi_7\}$; the question to Bob is whether to measure the state using observable $X_0$ or $X_2$. Then the winning probability of the CHSH game is exactly described by \eqref{eq:67} and the self-testing property of the CHSH translates to Lemma \ref{lem:a.5}.

\section{Remote Operator Application with Verifiability}\label{sec:4}
\subsection{Overview}
In this section we introduce a new notion named \emph{remote operator application with verifiability} (ROAV), for certifying server's operations.
\paragraph{Definitions of ROAV} In Section \ref{sec:4.1} we formalize the notion of ROAV for a POVM $(\cE_1,\cE_2\cdots \cE_D)$. Recall that in Section \ref{sec:1.3.2} we have informally introduced the notion of ROAV, which is defined to be a tuple $(\rho_{\ttest},\pi_{\ttest},\pi_{\tcomp})$ that certifies server-side operations.\par
We will also formalize a variant of ROAV where the input register of $(\cE_1,\cE_2\cdots \cE_D)$ is bigger than the server-side of $\rho_{\ttest}$; so the remaining part will be left to the server to decide. In Section \ref{sec:4.3.2} this part will be used to hold the witness state in a Hamiltonian ground state testing protocol.
\paragraph{Potential Applications of ROAV} In Section \ref{sec:4.3r} we discuss two potential applications of ROAV. In more detail:
	\begin{itemize}\item In Section \ref{sec:4.3} we show an approach for constructing RSPV from ROAV.\par
	Recall that we currently do not have a ``universal RSPV'', which is an RSPV for general state families. Our approach provides a potential way for constructing more complicated RSPV from simpler RSPV and ROAV that are possibly easier to construct, which potentially makes progress to the problem of constructing universal RSPV.\par
	In more detail, suppose we would like to construct an RSPV for the target state $\cE(\rho)$; but $\cE(\rho)$ might have a complicated form and it's not easy to construct RSPV for it directly. In this section we show that this problem could be decomposed into two pieces: the constructions of an ROAV for $\cE$ and a specific RSPV that is not related to $\cE$.
	\item In Section \ref{sec:4.3.2} we show an approach for constructing Hamiltonian ground state testing protocol using ROAV. As a review, existing Hamiltonian ground energy testing protocols like \cite{FHM,Grilo17} has a structure as follows:\par
Input: a local XZ-Hamiltonian $H=\sum_i\gamma_iH_i$. We would like to design a protocol do decide whether its ground state energy is $\leq a$ or $\geq b$.\par
The honest server gets a witness state $w$ that achieves the ground state energy.
\begin{enumerate}
	\item Repeat (sequentially or in parallel) the following for many rounds: the client samples a random $H_i$ and uses some protocols to get the measurement results of operator $H_i$ on the server-side state.
	\item The client calculates the weighted average of the measurement results and compares it to $a,b$ to decide.
\end{enumerate}
Typically the difficulty in the template above is how to make the server measure the $H_i$ honestly. 
 In this subsection we give a new approach for this problem. We show a protocol that reduces this problem to the following two protocols:
\begin{enumerate}
	\item An ROAV for tensor products of Bell basis measurements.
	\item A specific RSPV for state families that depend on the input Hamiltonian and the test state of the ROAV.
\end{enumerate}
The protocol takes ideas from the Grilo's Hamiltonian testing protocol \cite{Grilo17} in the non-local game setting. Roughly speaking, in Grilo's protocol the verifier randomly executes one of two modes of protocols, the $\foprtortest$ mode and $\fenergytest$ mode, where:
\begin{itemize}
	\item In $\foprtortest$ mode the vrifier certifies that the provers do the measurements required. Especially, Prover 2 should do Bell basis measurements to teleport some states to Prover 1.
	\item In $\fenergytest$ mode the witness state is teleported to Prover 1 and Prover 1 should do the energy testing on the state.
\end{itemize}
By assuming the existence of ROAV for Bell basis measurements and the RSPV for the states needed (the $\rho_{test}$ in ROAV and a state $\rho_{\tcomp}$ for energy testing), we could take this protocol to the single-server setting. Note that we need a variant of ROAV where the server could prepare part of the input states of the Bell measurements for holding the witness state.\par
 In more detail, in the $\foprtortest$ mode the client executes $\pi_{\ttest}$ to certify that the server has perform Bell basis measurements on registers $\bQ^{(\tin)}$ and $\bw$, where $\bQ^{(\tin)}$ is used to hold the server side of both $\rho_{\ttest}$ and $\ket{\Phi}$; the state of $\bw$ is decided by the server. If the server passes in the $\foprtortest$ mode, by the soundness of ROAV, in the $\fenergytest$ mode a Bell basis measurement is applied on the server side of $\ket{\Phi}$ and $\bw$, which teleports the state in $\bw$ to the register $\bP$ .\par
 However, we want a protocol where the client is completely classical so it could not do quantum testing directly on the state in register $\bP$. However, we could imagine that $\bP$ has already been measured beforehand and consider the post-measurement state on $\bQ^{(\tin)}$; we denote this state as $\rho_{\tcomp}$. Assuming an RSPV for $\rho_{\tcomp}$, the client could test the energy in the $\fenergytest$ mode.
\end{itemize}
Finally we note that we still need to construct ROAV concretely to make these reductions work; this is left to future works.
\subsection{Definitions of ROAV}\label{sec:4.1}
Let's first formalize the set-up of an ROAV protocol.
\begin{setup}[Set-up for an ROAV]\label{setup:5}An ROAV for a POVM $(\cE_1,\cE_2\cdots \cE_D)$ is in the form of $(\rho_{\ttest},\pi_{\ttest},\pi_{\tcomp})$ under the following set-up:\par
	The protocols $\pi_{\ttest},\pi_{\tcomp}$ take the following parameters: approximation error parameter $1^{1/\epsilon}$, security parameter $1^\kappa$.\par
	Input registers: the client-side classical register $\bD^{(\tin)}$, the server-side quantum register $\bQ^{(\tin)}$.\par
	Output registers: the client-side classical register $\bD^{(\tout)}$ with value in $[D]$, the server-side quantum register $\bQ^{(\tout)}$, and the client-side classical register $\bflag$ with value in $\{\fpass,\ffail\}$.\par
	The environment holds a quantum register $\bP$.\par
	$\rho_{\ttest}\in \tD(\cH_{\bD^{(\tin)}}\otimes\cH_{\bQ^{(\tin)}})$. $(\cE_1,\cE_2\cdots \cE_D)$ maps states in $\bQ^{(\tin)}$ to $\bQ^{(\tout)}$. $\dim(\bP)=\dim(\bQ^{(\tin)})$.\par
	Define the state \begin{equation}\label{eq:16}\ket{\Phi}=\frac{1}{\sqrt{D}}\sum_{i\in [D]}\underbrace{\ket{i}}_{\bP}\otimes \underbrace{\ket{i}}_{\bQ^{(\tin)}}.\end{equation}
For modeling the initial states in the malicious setting, assume the server-side registers (excluding $\bQ^{(\tin)},\bQ^{(\tout)}$) are denoted by register $\bS$, and assume the environment (excluding $\bP$) is denoted by register $\bbE$.
\end{setup}
 The completeness and soundness are defined as follows.
\begin{defn}[Completeness of ROAV]
	We say $(\rho_{\ttest},\pi_{\ttest},\pi_{\tcomp})$ under Set-up \ref{setup:5} is complete if when the server is honest:
	\begin{itemize}\item In $\pi_{\ttest}$ when the initial state is $\rho_{\ttest}$ the passing probability is negligibly close to $1$.
		\item In $\pi_{\tcomp}$ when the initial state is $\ket{\Phi}$ the output state of the protocol is negligibly close to
		$$\underbrace{\ket{\fpass}\bra{\fpass}}_{\bflag}\otimes(\underbrace{\bbI}_{\text{on }\bP}\otimes \cE_{tar})(\Phi)$$
	where
	$$\cE_{tar}(\underbrace{\cdot}_{\bQ^{(\tin)}})=\sum_{i\in [D]}\underbrace{\ket{i}\bra{i}}_{\bD^{(\tout)}}\otimes \underbrace{\cE_i(\cdot)}_{\bQ^{(\tout)}}$$
	\end{itemize}
\end{defn}
\begin{defn}[Soundness of ROAV]\label{defn:4.4}
	We say $(\rho_{\ttest},\pi_{\ttest},\pi_{\tcomp})$ under Set-up \ref{setup:5} is $(\delta,\epsilon)$-sound if:\par  
	For any efficient quantum adversary $\fAdv$, there exists an efficient quantum operation $\fSim$ such that for any state $\rho_{0}\in \tD(\cH_{\bS}\otimes \cH_{\bbE})$:\par
	If
	 $$\tr(\Pi_{\fpass}(\pi_{\ttest}^\fAdv(\rho_{\ttest}\otimes\rho_{0})))\leq 1-\delta$$
	\item 
	 then \begin{equation}\label{eq:13}\Pi_{\fpass}(\pi_{\tcomp}^{\fAdv}(\Phi\otimes\rho_{0}))\approx^{ind}_{\epsilon}\underbrace{\Pi_{\fpass}}_{\text{on }\bflag}(\underbrace{\fSim}_{\text{on }\bS,\bQ^{(\tin)},\bQ^{(\tout)},\bflag}(\underbrace{\cE_{tar}(\Phi)}_{\bD^{(\tout)},\bQ^{(\tout)},\bP}\otimes \underbrace{\rho_{0}}_{\bS,\bbE}))\end{equation}
\end{defn}
\subsubsection{Variant: ROAV with extra server-side states (besides the EPR parts)}\label{sec:4.1.2}
Note that in the definitions above we assume the input register of $(\cE_1,\cE_2\cdots \cE_D)$ is the same as the server-side of $\rho_{\ttest}$ (which are both $\bQ^{(\tin)}$). Below we formalize a variant where the input register of $(\cE_1,\cE_2\cdots \cE_D)$ is bigger than the server-side of $\rho_{\ttest}$.
\begin{setup}[Set-up for an ROAV with extra server-side states]\label{setup:9}
	Compare to Set-up \ref{setup:5}, the server holds an additional quantum register $\bw$. 
	The other parts of the set-up are the same as Set-up \ref{setup:5}.
\end{setup}
\begin{defn}
We say $(\rho_{\ttest},\pi_{\ttest},\pi_{\tcomp})$ under Set-up \ref{setup:9} is complete if for any $\rho_{0}\in \tD(\cH_{\bw}\otimes\cH_{\bbE})$:
\begin{itemize}
\item In $\pi_{\ttest}$ when the initial state is $\rho_{\ttest}\otimes \rho_0$ the passing probability is negligibly close to $1$.
\item In $\pi_{\tcomp}$ when the initial state is $\Phi\otimes \rho_0$ the output state of the protocol is negligibly close to
 $$\underbrace{\ket{\fpass}\bra{\fpass}}_{\bflag}\otimes(\underbrace{\bbI}_{\bP}\otimes\cE_{tar}\otimes \underbrace{\bbI}_{\bbE})(\Phi\otimes\rho_{0})$$
 where
 \begin{equation}\label{eq:23r}\cE_{tar}(\underbrace{\cdot}_{\bQ^{(\tin)},\bw})=\sum_{i\in [D]}\underbrace{\ket{i}\bra{i}}_{\bD^{(\tout)}}\otimes \underbrace{\cE_i(\cdot)}_{\bQ^{(\tout)}}\end{equation}
\end{itemize}
\end{defn}
The soundness definition contains an additional simulator for simulating the states on $\bw$ compared to Definition \ref{defn:4.4}.
\begin{defn}
	We say $(\rho_{\ttest},\pi_{\ttest},\pi_{\tcomp})$ under Set-up \ref{setup:9} is $(\delta,\epsilon)$-sound if:\par
	For any efficient quantum adversary $\fAdv$, there exist efficient quantum operations $\fSim=(\fSim_0,\fSim_1)$ such that for any state $\rho_{0}\in \tD(\cH_{\bS}\otimes \cH_{\bbE})$:\par
	If
	 $$\tr(\Pi_{\fpass}(\pi_{\ttest}^\fAdv(\rho_{\ttest}\otimes\rho_{0})))\leq 1-\delta$$
	\item 
	 then \begin{equation}\Pi_{\fpass}(\pi_{\tcomp}^{\fAdv}(\Phi\otimes\rho_{0}))\approx^{ind}_{\epsilon}\underbrace{\Pi_{\fpass}}_{\text{on }\bflag}(\underbrace{\fSim_1}_{\text{on }\bS,\bQ^{(\tin)},\bQ^{(\tout)},\bflag}((\underbrace{\bbI}_{\bP}\otimes\underbrace{\cE_{tar}}_{\text{as }\eqref{eq:23r}}\otimes \underbrace{\bbI}_{\bbE})(\Phi\otimes \underbrace{\fSim_0}_{\text{on }\bw,\bS}(\underbrace{\rho_{0}}_{\bS,\bbE}))))\end{equation}
\end{defn}
\subsection{Potential Applications of ROAV}\label{sec:4.3r}\subsubsection{Building RSPV from ROAV}\label{sec:4.3}
In this subsection we give a protocol for building RSPV protocols from ROAV and RSPV that are possibly more basic. The intuition is discussed in the beginning of Section \ref{sec:4} where $\rho$ there corresponds to $\rho_{\tcomp}$ below.
\begin{mdframed}[backgroundcolor=black!10]
	Suppose $(\rho_{\ttest},\pi_{\ttest},\pi_{\tcomp})$ is an ROAV under Set-up \ref{setup:5} for target operator $\cE$. Suppose the client also holds a classical register $\bD^{(\tcomp)}$ and consider a state $\rho_{\tcomp}\in \tD(\cH_{\bD^{(\tcomp)}}\otimes \cH_{\bQ^{(\tin)}})$. Suppose $\pi_0(\tmode),\tmode\in \{\ttest,\tcomp\}$ is an RSPV (where the client could choose the states, as described in Set-up \ref{setup:3}, Section \ref{sec:3.3.1}) for the following honest behavior:\begin{itemize}\item If $\tmode=\ttest$, prepare the state $\rho_{\ttest}$ on registers $\bD^{(\tin)},\bQ^{(\tin)}$. \item If $\tmode=\tcomp$, prepare the state $\rho_{\tcomp}$ on registers $\bD^{(\tcomp)},\bQ^{(\tin)}$.\end{itemize}
	Suppose $\pi_0$ is $\epsilon_0$-sound; suppose $(\rho_{\ttest},\pi_{\ttest},\pi_{\tcomp})$ is $(\delta,\epsilon_1)$-sound.
\begin{prtl}\label{prtl:2}$(\pi_{\ttest}^\prime,\pi_{\tcomp}^\prime)$ below achieves a preRSPV for target state $\cE(\rho_{\tcomp})$.\par
	Output registers: client-side classical registers $\bD^{(\tcomp)}$, $\bD^{(\tout)}$, $\bflag$; server-side quantum register $\bQ^{(\tout)}$.\par 
	$\pi_{\ttest}^\prime$ is defined as follows:
		\begin{enumerate}
			\item The client executes $\pi_0(\ttest)$ with the server. Store the outputs in $\bflag^{(1)},\bD^{(\tin)},\bQ^{(\tin)}$.
		\item The client executes $\pi_{\ttest}$ with the server. Store the outputs in $\bflag^{(2)}$.
		\item The client sets $\bflag$ to be $\ffail$ if any one of $\bflag^{(1)}, \bflag^{(2)}$ is $\ffail$; otherwise it sets $\bflag$ to be $\fpass$.
	\end{enumerate}\par
	$\pi_{\tcomp}^\prime$ is defined as follows:
	\begin{enumerate}
		\item The client executes $\pi_0(\tcomp)$ with the server. Store the outputs in $\bflag^{(1)},\bD^{(\tcomp)},\bQ^{(\tin)}$.
		\item The client executes $\pi_{\tcomp}$ with the server. Store the outputs in $\bflag^{(2)}, \bD^{(\tout)},\bQ^{(\tout)}$.
		\item The client sets $\bflag$ to be $\ffail$ if any one of $\bflag^{(1)}, \bflag^{(2)}$ is $\ffail$; otherwise it sets $\bflag$ to be $\fpass$.
	\end{enumerate}
\end{prtl}
\end{mdframed}
The completeness is from the protocol description: it passes in the test mode and prepares $\cE(\rho_{\tcomp})$ in the comp mode. The efficiency is also trivial. Below we state the soundness.
\begin{thm}\label{thm:4.3}
	Protocol \ref{prtl:2} is a preRSPV for target state $\cE(\rho_{\tcomp})$ that is $(\epsilon_0+\delta,\epsilon_0+\epsilon_1)$-sound.
\end{thm}
Compare the ROAV soundness with what we want, a missing step is how to relate the $\Phi$ in the ROAV soundness with $\rho_{\tcomp}$. This step is by the following fact.
\begin{fact}\label{fact:2}
	For $\Phi$ defined in Set-up \ref{setup:5}, there exists a POVM such that measuring $\bP$ and storing results in $\bD^{(\tcomp)}$ transforms $\Phi$ to $\rho_{\tcomp}$.
\end{fact}
\begin{proof}[Proof for Theorem \ref{thm:4.3}]
	Suppose the adversary is $\fAdv$, by the soundness of $\pi_0$ (Definition \ref{defn:choseninputs}) there exists an efficient quantum operation $\fSim_1$ such that for any $\rho_0\in\tD(\cH_{\bS}\otimes\cH_{\bbE})$,
	\begin{equation}\label{eq:30p}\Pi_{\fpass}^{\bflag^{(1)}}(\pi_0^{\fAdv_1}(\ttest)(\rho_0))\approx_{\epsilon_0}^{ind}\Pi_{\fpass}^{\bflag^{(1)}}(\underbrace{\fSim_1}_{\text{on }\bQ^{(\tin)},\bS,\bflag^{(1)}}(\rho_{\ttest}\otimes\rho_0))\end{equation}
	\begin{equation}\label{eq:31p}\Pi_{\fpass}^{\bflag^{(1)}}(\pi_0^{\fAdv_1}(\tcomp)(\rho_0))\approx_{\epsilon_0}^{ind}\Pi_{\fpass}^{\bflag^{(1)}}(\fSim_1(\rho_{\tcomp}\otimes\rho_0))\end{equation}
	To prove the soundness of Protocol \ref{prtl:2}, let's assume $\fAdv$ could make $\pi_{\ftest}^\prime$ passes with probability $\geq 1-(\epsilon_0+\delta)$. This combined with \eqref{eq:30p} implies that if the protocol $\pi_{test}$ is executed on initial states $\rho_{\ttest}\otimes\rho_0$ against adversary $\fAdv_2\circ\fSim_1$, the passing probability is $\geq 1-\delta$. Then by the soundness of ROAV we have that, there exists an efficient quantum operation $\fSim_2$ such that
	\begin{equation}\label{eq:32}\Pi_{\fpass}(\pi_{\fcomp}^{\fAdv_2}(\fSim_1(\Phi\otimes \rho_0)))\approx^{ind}_{\epsilon_1}\Pi_{\fpass}(\underbrace{\fSim_2}_{\text{on }\bS,\bQ^{(\tin)},\bQ^{(\tout)},\bflag^{(2)}}(\underbrace{\cE}_{\bQ^{(\tin)}\rightarrow \bQ^{(\tout)}}(\underbrace{\Phi}_{\bP,\bQ^{(\tin)}})\otimes\rho_0))\end{equation}
	Note that the register $\bP$ is in the environment, which is accessible by the distinguisher but not by any operators explicitly appeared in \eqref{eq:32}. Thus we could first apply Fact \ref{fact:2} and make the distinguisher measure the $\bP$ register of $\Phi$ to collapse $\Phi$ to $\rho_{\tcomp}$, and then assume the collapsing happen in the beginning. Thus we get 
	\begin{equation}\label{eq:30}\Pi_{\fpass}(\pi_{\fcomp}^{\fAdv_2}(\fSim_1(\rho_{\tcomp}\otimes \rho_0)))\approx^{ind}_{\epsilon_1}\Pi_{\fpass}(\fSim_2(\cE(\rho_{\tcomp})\otimes\rho_0))\end{equation}
	Combining \eqref{eq:30} and \eqref{eq:31p} completes the proof.
\end{proof}
\subsubsection{Testing ground state energy by ROAV}\label{sec:4.3.2}
In this subsection we give a Hamiltonian ground energy testing protocol based on specific ROAV and RSPV.
We first give the set-ups for the ROAV in our protocol. The set-up is based on the template set-up for ROAV with the server-side states (Set-up \ref{setup:9}). 
\begin{setup}\label{setup:10}
	Parameters: \begin{itemize}\item problem size parameter $1^n$ which describes the size of the witness of the Hamiltonian;\item $1^K$ which describes the number of repetition;\item approximation error parameter $1^{1/\epsilon}$;\item security parameter $1^\kappa$.\end{itemize}\par
	Registers: 
	\begin{itemize}\item client-side classical registers $\bD^{(\tin)},\bD^{(\tcomp)},\bD^{(\tout)}$; $\bD^{(\tout)}$ holds $2nK$ classical bits. $\bD^{(\tin)},\bD^{(\tcomp)}$ are defined to be compatible with the states below. \item server-side quantum register $\bQ^{(\tin)}$, $\bw$; each of both holds $nK$ qubits.\item a register $\bP$ in the environment; $\dim(\bP)=\dim(\bQ^{(\tin)})$.\end{itemize} For $k\in [K]$, denote the $k$-th block of $\bD^{(\tout)}$ as $\bD^{(\tout)}_k$ which holds $2n$ bits and $\bD^{(\tcomp)}_k,\bQ^{(\tin)}_k$, $\bP_k$ are defined similarly.\par
	Define the following state and operations:\par
	$\rho_{\ttest}\in \tD(\cH_{\bD^{(\tin)}}\otimes\cH_{\bQ^{(\tin)}})$. $\ket{\Phi}$ is the EPR entanglement between $\bP$ and $\bQ^{(\tin)}$.\par
	Define $\cE_{Bells}$ to be the operation that measures $\bQ^{(\tin)},\bw$ in the Bell basis and stores the results in $\bD^{(\tout)}$; thus when the state in $\bP,\bQ^{(\tin)}$ is $\ket{\Phi}$, the application of $\cE_{Bell}$ teleports the state in $\bw$ to $\bP$.\par
	Corresponding to an XZ local Hamimltonian $H=\sum_{j\in [m]}\gamma_jH_j$, $1^K$, define $\rho_{\tcomp}\in \tD(\cH_{\bD^{(\tcomp)}}\otimes\cH_{\bQ^{(\tin)}})$ as the outcome of the following operations applied on $\ket{\Phi}$: 
	\begin{enumerate}\item For each $k\in [K]$:\par
		 The client randomly chooses $j^{(k)}\in [m]$ and $\bP_k$ is measured as follows:\par
\begin{itemize}\item If the observable on the $t$-th qubit of $H_{j^{(k)}}$ is $\sigma_Z$, measure the $t$-th qubit of $\bP_k$ on the computational basis. The measurement outcome is represented by value in $\{0,1\}$.
	\item If the observable on the $t$-th qubit of $H_{j^{(k)}}$ is $\sigma_X$, measure the $t$-th qubit of $\bP_k$ on the Hadamard basis. The measurement outcome is represented by value in $\{0,1\}$.
\end{itemize}
Store $j^{(k)}$ together with the measurement results on $\bD^{(\tcomp)}_k$. Suppose the index $j^{(k)}$ is stored in register $\bD^{(\tcomp)(\tindex)}_k$ and the measurement results are stored in register $\bD^{(\tcomp)(\tmr)}_k$
	\end{enumerate}
\end{setup}
Thus the application of $\cE_{Bells}$ on $\rho_{\tcomp}$ and $w\in\tD(\cH_{\bw})$ could be understood as follows: $w\in \tD(\cH_{\bw})$ is teleported to $\bP$, and then the client samples $j^{(k)}$ for each $k\in [K]$ and measures $H_{j^{(k)}}$ and stores the outcome in $\bD^{(\tcomp)(\tmr)}_k$. $\bD^{(\tcomp)(\tmr)}_k$ records the outcome in each round of Hamiltonian ground state testing, under the one-time-pad-encryption under corresponding keys in $\bD^{(\tout)}_k$. \par
Then we introduce notations for calculating the estimated energy in such a procedure:
\begin{nota}
	Then define $\val^H(\bD^{(\tcomp)},\bD^{(\tout)})$ as follows:
$$\val^H(\bD^{(\tcomp)},\bD^{(\tout)})=\frac{1}{K}\sum_{k\in [K]}\val^H(\bD^{(\tcomp)}_k,\bD^{(\tout)}_k)$$
$$\val^H(\bD^{(\tcomp)}_k,\bD^{(\tout)}_k)=\gamma_{j^{(k)}}\cdot (-1)^{\tParity(\bD^{(\tcomp)(\tmr)}_k\oplus\tdecodekey^{H_{j^{(k)}}}(\bD^{(\tout)}_k))}$$
where $j^{(k)}$ is the value of $\bD^{(\tcomp)(\tindex)}_k$, and $\tdecodekey^{H_{j^{(k)}}}(\bD^{(\tout)}_k)$ is defined as follows: 
\begin{itemize}\item  If the observable on the $t$-th qubit of $H_{j^{(k)}}$ is $\sigma_Z$, the decode key for the corresponding bit in $\bD^{(\tcomp)(\tmr)}_k$ is the $(2t-1)$-th bit of $\bD^{(\tout)}_k$.
	\item If the observable on the $t$-th qubit of $H_{j^{(k)}}$ is $\sigma_X$, the decode key for the corresponding bit in $\bD^{(\tcomp)(\tmr)}_k$ is the $2t$-th bit of $\bD^{(\tout)}_k$.
\end{itemize}
\end{nota}
Thus the $\val^H(\bD^{(\tcomp)},\bD^{(\tout)})$ is the estimation of the Hamiltonian ground state energy by taking the weighted average of $K$ sampling.
\begin{mdframed}[backgroundcolor=black!10]
	Under Set-up \ref{setup:10}, suppose $(\rho_{\ftest},\pi_{\ftest},\pi_{\fcomp})$ is an ROAV for $\cE_{Bells}$ that is $(\delta,\epsilon)$-sound. Suppose $\pi_{0}(\tmode,H,1^K)$ as an RSPV (where the client could choose the states, as described in Set-up \ref{setup:3}, Section \ref{sec:3.3.1}) for the following honest behavior:
	\begin{itemize}
		\item If $\tmode=\ttest$, prepare $\rho_{\ttest}$.
		\item If $\tmode=\tcomp$, prepare $\rho_{\tcomp}$ as defined in Set-up \ref{setup:10}.
	\end{itemize}
	and $\pi_0$ is $\epsilon_0$-sound.
\begin{prtl}\label{prtl:3}
	Input: an XZ 5-local Hamiltonian $H=\sum_{j\in [m]}\gamma_jH_j$, $a,b$, $b-a\geq 1/\fpoly(n)$, as Definition \ref{defn:2.2}. Witness size parameter $1^n$. Security parameter $1^\kappa$.\par
	Take $K=100\kappa^2\frac{1}{(b-a)^2}$. 
	\begin{enumerate}
			\item The client samples $\tmode\in\{\foprtortest,\fenergytest\}$ with probability $(\frac{1}{2},\frac{1}{2})$ randomly. Depending on the value of $\roundtype$: 
			\begin{itemize}
				\item If $\tmode=\foprtortest$:\begin{enumerate}
				\item The client executes $\pi_0(\ttest,1^n,1^K)$. \item Then the client executes $\pi_{\ttest}$ with the server. 
				\item Reject if any step fails and accept otherwise.
				\end{enumerate}
				\item If $\tmode=\fenergytest$:\begin{enumerate}
				\item The client executes $\pi_0(\tcomp,H,1^K)$. \item Then the client executes $\pi_{\tcomp}$ with the server. 
				\item Reject if any step fails or $\val^H(\bD^{(\tcomp)},\bD^{(\tout)})\geq \frac{a+b}{2}$ and accept otherwise.
				\end{enumerate}
			\end{itemize}
	\end{enumerate}
\end{prtl}
\end{mdframed}
The completeness is by an application of the Chernoff bound and the efficiency is trivial. Below we state and prove the soundness.
\begin{thm}
 When $H$ has ground state energy $\geq b$, Protocol \ref{prtl:3} accepts with probability at most $\fneg(\kappa)+\max\{1-\frac{1}{2}(\delta-\epsilon_0),\frac{1}{2}+\frac{1}{2}(\epsilon+\epsilon_0)\}$.
\end{thm}
\begin{proof}
	Suppose $H$ has ground state energy $\geq b$, the adversary is $\fAdv$, \footnote{Note that we omit the initial state which is previously described by $\rho_0\in \tD(\cH_{\bS}\otimes\cH_{\bbE})$. This is because since we only need to bound the passing probability of the overall protocol instead of proving a simulation-based soundness, we could simply assume the initial state is prepared by $\fAdv$.} and the protocol accepts with probability\footnote{Below we omit the $\fneg(\kappa)$ part in the soundness error during the proof.} $\max\{1-\delta+\epsilon_0,\epsilon+\epsilon_0+\frac{1}{2}\}$. This implies that in the $\foprtortest$ mode the protocol passes with probability $\geq 1-\delta+\epsilon_0$. By the soundness of $\pi_0$ there exists a simulator $\fSim_1$ that simulates the state after the first step with approximation error $\epsilon_0$, which implies that $\pi_{\ttest}$ running from state $\rho_{\ttest}$ passes with probability $\geq \delta$. By the soundness of ROAV we have that, there exists an efficent simulator $\fSim_2$, a quantum state $w\in \tD(\cH_{\bw})$ such that
\begin{equation}\label{eq:39}
	\Pi_{\fpass}(\pi_{\fcomp}^{\fAdv_2\circ\fSim_1}(\Phi))\approx^{ind}_{\epsilon}\Pi_{\fpass}(\fSim((\underbrace{\bbI}_{\bP}\otimes\cE_{Bells})(\Phi\otimes w)))
\end{equation}
Now we consider a distinguisher that measures the register $\bP$ as described in Set-up \ref{setup:10} to collapse it to $\rho_{\tcomp}$. 
This implies
\begin{equation}\label{eq:41}
	\Pi_{\fpass}(\pi_{\fcomp}^{\fAdv_2\circ\fSim_1}(\rho_{\tcomp}))\approx^{ind}_{\epsilon}\Pi_{\fpass}(\fSim((\underbrace{\bbI}_{\bD^{(\tcomp)}}\otimes\cE_{Bells})(\rho_{\tcomp}\otimes w)))
\end{equation}
The passing probability of the right hand side could be bounded directly by Chernoff bound, which is negligible when the ground state energy is $\geq b$. Then \eqref{eq:41}, the approximation error of $\pi_0$ together with the fact that the $\fenergytest$ mode happens with probability $\frac{1}{2}$ imply that the passing probability is at most $\frac{1}{2}+\frac{1}{2}(\epsilon+\epsilon_0)$.
\end{proof}

\section{Concrete Constructions of RSPV Protocols}\label{sec:5}
\subsection{Overview}\label{sec:5.1}
Let's recall the reduction diagram (Figure \ref{fig:2}). In the following sections we will complete these reductions step by step; let's first give an overview that elaborates the intuition of each step.
\subsubsection{From BB84 to $\fOneBlock$ and $\fOneBlockTensor$} In Protocol \ref{prtl:4} we build an RSPV for state family $$\{\frac{1}{\sqrt{2}}(\ket{x_0}+\ket{x_1}:x_0,x_1\in \{0,1\}^m,\tHW(x_0\oplus x_1)=1,\tParity(x_0)=0)\}$$ from the RSPV for BB84 states; this protocol is called $\fOneBlock$. The intuition is as described in Section \ref{sec:1.3.3}: notice that states in this family could be written as sequence of BB84 states where there is only one $\ket{+}$ state and no $\ket{-}$ state; thus the client only needs to do a repeat-and-pick on RSPV-for-BB84 protocol.\par
	 In Protocol \ref{prtl:5} we build the RSPV for 
	 \begin{equation}\label{eq:46}\{\frac{1}{\sqrt{2}}(\ket{x_0}+\ket{x_1}):x_0,x_1\in \{0,1\}^m,\text{HW}(x_0\oplus x_1)=1,\tParity(x_0)=0\}^{\otimes n}\end{equation}
	 which is the tensor products of the state family in $\fOneBlock$; this protocol is called $\fOneBlockTensor$. This is by taking the sequential repetition (as discussed in Section \ref{sec:3.4.1}) of the protocol $\fOneBlock$.
	 \subsubsection{Construction of $\fMultiBlock$} In Section \ref{sec:5.3} we go from $\fOneBlockTensor$ to construct an RSPV for the following state family:
\begin{equation*}\{\frac{1}{\sqrt{2}}(\ket{x_0^{(1)}||x_0^{(2)}||\cdots ||x_0^{(n)}}+\ket{x_1^{(1)}||x_1^{(2)}||\cdots ||x_1^{(n)}}):\end{equation*}
\begin{equation}\label{eq:47}\forall i\in [n],x_0^{(i)},x_1^{(i)}\in \{0,1\}^m,\text{HW}(x_0^{(i)}\oplus x_1^{(i)})=1;\tParity(x_0^{(1)})=0\}\end{equation}
Roughly speaking, the client forces the server to measure the state in \eqref{eq:46} to get \eqref{eq:47}. To see this clearly, we could first consider two blocks as an example. Suppose the server initially holds the state
$$(\ket{x_0^{(1)}}+\ket{x_1^{(1)}})\otimes (\ket{x_0^{(2)}}+\ket{x_1^{(2)}}),\quad\forall b,i,\tParity(x_b^{(i)})=b$$
The client asks the server to measure the total parity of the strings it holds. Then if the server performs the measurement honestly, the state will collapse to:
\begin{equation}\label{eq:outcome0}\text{outcome}=0:\ket{x_0^{(1)}||x_0^{(2)}}+\ket{x_1^{(1)}||x_1^{(2)}}\end{equation}
\begin{equation}\label{eq:outcome1}\text{outcome}=1:\ket{x_0^{(1)}||x_1^{(2)}}+\ket{x_1^{(1)}||x_0^{(2)}}\end{equation}
The client could update the keys (that is, the state description) using the reported outcome and the original keys. If the client does the same for each $i=2,3\cdots n$, in the honest setting the state in \eqref{eq:46} will collapse to \eqref{eq:47}.\par
So what if the server cheats? One possible attack is that the server may not do the measurements and report the total parity honestly. To detect this attack, we will first construct a preRSPV $(\fMultiBlockTest,\fMultiBlockComp)$ and use $\fMultiBlockTest$ to test the server's behavior. In $\fMultiBlockTest$ after getting the total parities the client will ask the server to measure all the states on the computational basis and report the measurement results; the client could check the results with its keys and see whether the results is consistent with the original keys and the total parities. As a concrete example, in \eqref{eq:outcome0} if the client asks the server to measure all the states, the measurement results should be either $x_0^{(1)}||x_0^{(2)}$ or $x_1^{(1)}||x_1^{(2)}$, otherwise the server is caught cheating.\par
After we get a preRSPV, we could make use of the amplification in Section \ref{sec:preRSPVdef} to get an RSPV for \eqref{eq:47}.\par
Let's briefly discuss how the security proof goes through. One desirable property of the security proof of this step is that, we only need to analyze a ``information-theoretic core'': if we assume the initial state is \eqref{eq:46}, the analysis of the preRSPV will be purely information-theoretic, which means, the soundness holds against unbounded provers and we do not need to work on computational notions in the security analysis. After we prove the soundness of this information-theoretic part, we could prove the soundness of the overall protocol by calling the abstract properties in Section \ref{sec:3.4.2}, \ref{sec:preRSPVdef}.
	 \subsubsection{Construction of $\fKP$} In Section \ref{sec:5.4} we go from $\fMultiBlock$ to construct an RSPV for state family 
	 \begin{equation}\label{eq:kp}\{\frac{1}{\sqrt{2}}(\ket{0}\ket{x_0}+\ket{1}\ket{x_1}):x_0,x_1\in \{0,1\}^n\}\end{equation}
This is by first calling $\fMultiBlock$ to prepare a sufficiently big state in the form of \eqref{eq:47}, and then letting the client reveals suitable information to allow the server to transform \eqref{eq:47} to \eqref{eq:kp}. Let's explain the intuition.\par
We first note that by calculating the parity of the first block, \eqref{eq:47} could be transformed to
\begin{equation}\label{eq:46pl}\frac{1}{\sqrt{2}}(\ket{0}\ket{x_0^{(1)}||x_0^{(2)}||\cdots }+\ket{1}\ket{x_1^{(1)}||x_1^{(2)}||\cdots}):
\forall i,x_0^{(i)},x_1^{(i)}\in \{0,1\}^m,\text{HW}(x_0^{(i)}\oplus x_1^{(i)})=1\end{equation}
The difference to \eqref{eq:kp} is that these keys are not sampled uniformly randomly. Note that in \eqref{eq:46pl} we omit some conditions on the keys and focus on the most significant one.\par
 Then the client will reveal lots of information about these keys, which allows the server to transform each two blocks in \eqref{eq:46pl} to a pair of uniform random bits. Let's use the first two blocks as an example. The client will reveal $x_0^{(1)}$ and $x_1^{(2)}$. Then by doing xor between them and the corresponding blocks, \eqref{eq:46pl} becomes:
 \begin{equation}\label{eq:46pl2}\frac{1}{\sqrt{2}}(\ket{0}\ket{0^m||000\cdots 1000\cdots||\cdots }+\ket{1}\ket{000\cdots 1000\cdots||0^m||\cdots})\end{equation}
  Now the $000\cdots 1000\cdots$ could be seen as a unary encoding of a random number in $[m]$. By choosing $m$ to be power of 2, converting unary encoding to binary encoding and trimming out extra zeros, this state becomes
  \begin{equation}\label{eq:46pl3}\frac{1}{\sqrt{2}}(\ket{0}\ket{\gamma_0^{(1)}||\cdots }+\ket{1}\ket{\gamma_1^{(1)}||\cdots})\end{equation}
  where $\gamma_0,\gamma_1\in \{0,1\}^{[\log m]}$ and are uniformly random. Doing this for each two blocks in \eqref{eq:46pl} gives \eqref{eq:kp}.
	 \subsubsection{Construction of $\fQFac$ (RSPV for $\ket{+_\theta}$)} Now we have an RSPV for states $\frac{1}{\sqrt{2}}(\ket{0}\ket{x_0}+\ket{1}\ket{x_1})$ with uniformly random $x_0,x_1$; we are going to construct an RSPV for the state\footnote{The protocol name $\fQFac$ is from \cite{qfactory}.}
	 $$\ket{+_\theta}=\frac{1}{\sqrt{2}}(\ket{0}+e^{\mi\pi\theta/4}\ket{1}),\theta\in \{0,1,2\cdots 7\}.$$
	 We first note that existing work \cite{GVRSP} also takes a similar approach: the client first instruct the server to prepare the state $\ket{0}\ket{x_0}+\ket{1}\ket{x_1}$, and then transforms it to $\ket{+_\theta}$.\par
Let's explain the constructions. The overall ideas for the construction are basically from \cite{GVRSP} (adapted to the languages of our framework). We will first construct a preRSPV with the score, as follows: in both the test mode and comp mode the client will instruct the server to prepare $\ket{+_\theta}$ state, then in the test mode the client will instruct the server to measure $\ket{+_\theta}$, and stores a score based on the reported result; in the comp mode $\ket{+_\theta}$ will be kept. Then once we show this protocol is indeed a preRSPV, we could amplify it to an RSPV as in Section \ref{sec:3.5.2}.\par
The first step is to allow the honest server to transform $\frac{1}{\sqrt{2}}(\ket{0}\ket{x_0}+\ket{1}\ket{x_1})$ to $\ket{+_\theta}$. There are multiple ways to do it, for example:
\begin{enumerate}\item The client will first instruct the server to introduce a phase of $e^{\mi\pi\theta_{2,3}/4},\theta_{2,3}\in \{0,1,2,3\}$ where $\theta_{2,3}$ are hidden in the server's view. This could be done by selecting the xor of the first two bits of $x_0,x_1$ as $\theta_2,\theta_3$. Then on the one hand $\theta_{2,3}=2\theta_2+\theta_3$ will be completely hidden; on the other hand using the phase-table-like technique in \cite{revgt,cvqcinlt} the honest server could introduce a phase of $e^{\mi\pi\theta_{2,3}/4}$ to the $x_1$ branch up to a global phase.
	\item The server does a Hadamard measurement on each bit of the $x$-part and get a measurement result $d$; this introduces a phase of $e^{\mi\pi(d\cdot (x_0+x_1))}$ to the qubits. Then the server sends back $d$ to the client and the client could calculate $\theta_1=(d\cdot (x_0+x_1))\mod 2$. The server holds a single qubit in the state $\ket{+_\theta},\theta=4\theta_1+2\theta_2+\theta_3$.
\end{enumerate}
Now in the comp mode we are done. In the test mode the client will continue to ask the server to make measurement on a random basis $\ket{+_\varphi},\ket{+_{\varphi+4}},\varphi\leftarrow_r\{0,1,2\cdots 7\}$. First we could see that when $\theta=\varphi$ the measurement will collapse to $\ket{+_\varphi}$ with probability $1$, and when $\theta=\varphi+4$ the measurement will collapse to $\ket{+_{\varphi+4}}$ with probability $1$. Thus the client could record a ``win'' score if he has seen such an outcome. But solely doing this does not give us a full control on the server's behavior and states; an important idea is that, when $\varphi$ is close to $\theta$ (or $\theta+4$), the measurement should also collapse to $\ket{+_\varphi}$ (or $\ket{+_{\varphi+4}}$, correspondingly). This gives some probability of losing even for an honest server; however by analyzing the game it's possible to say ``if the server wins with probability close to the optimal winning probability, the state before the testing measurement has to be close to the target state up to an isometry'', which is still sufficient for amplification.
\paragraph{Security proofs, existing works and their limitations, and our approach}
So how could the quoted claim just now be proved? Existing works like \cite{GVRSP,qfactory} has already done a lot of works on this part. In \cite{qfactory} the authors introduce a notion called \emph{blind self-testing}. In more detail, let's denote the server-side state by the time that the comp mode is done corresponding to the client-side phase $\theta$ as $\rho_{\theta}$. (In other words, the overall state is $\sum_{\theta}\ket{\theta}\bra{\theta}\otimes\rho_{\theta}$.) Then the blind self-testing requires that the state $\rho_{\theta_{2,3}}+\rho_{\theta_{2,3}+4}$ is the same for any $\theta_{2,3}$, which is called (information-theoretic) \emph{basis-blindness} (here $\theta_{2,3}$ is the ``basis'' and the notion means that the basis is completely hidden after randomization of $\theta_1$). \cite{qfactory} shows that if the initial state satisfies the basis blindness property, the claim ``high winning probability $\Rightarrow$ close to the target state up to an isometry'' holds.\par
However, the proof of this claim given in \cite{qfactory} does not generalize to the computational analog of basis blindness. \cite{qfactory} does not solve the problem and leave the security of the whole protocol as a conjecture. \cite{GVRSP} makes use of computational indistinguishability between states in the form of $\rho_{\theta_{2,3}}+\rho_{\theta_{2,3}+4}$ together with some quantum information theoretic arguments to prove the claim; in their proofs computational indistinguishability and quantum information theoretic arguments are mixed together, which could be complicated to work on and lead to sophisticated details \cite{discussionwithVidick}.\par
	 Our new hammer for getting rid of this problem is the $\fKP$ protocol, which is an RSPV for $\ket{0}\ket{x_0}+\ket{1}\ket{x_1}$. (Note that in previous works the preparation of $\ket{0}\ket{x_0}+\ket{1}\ket{x_1}$ is not known to satisfy the RSPV soundness in the malicious setting.) By starting from $\fKP$, we are able to prove the security in a much nicer way:
	 \begin{enumerate}\item The ``information-theoretic core'': in the security proof we could first simply assume the server holds exactly the state $\ket{0}\ket{x_0}+\ket{1}\ket{x_1}$; then we could prove the $\rho_{\theta}$ generated from it has the (information-theoretic) basis blindness property; then the analysis of the testing on $\ket{+_\theta}$ is basically from existing results \cite{GVRSP,qfactory}. \item Once we complete the analysis of this information-theoretic core, we could compile it using results in Section \ref{sec:3.4.2} to get the desired soundness for the overall protocol, and then compile the preRSPV to an RSPV by amplification procedure in Section \ref{sec:3.5.2}.\end{enumerate}
\subsubsection{A summary}
In summary we have achieve each step of the reductions as shown in Figure \ref{fig:2}. The construction of $\fOneBlock$ is by a repeat-and-pick procedure, $\fOneBlockTensor$ is by the sequential composition, $\fKP$ is by revealing some information to allow the server to transform the state to some other forms. For constructions of $\fMultiBlock$ and $\fQFac$, we first analyze an ``information-theoretic core'' where we only need to work on statistical closeness, and compile the IT-core (IT$=$information-theoretic) to a full protocol by calling existing soundness properties.
\subsection{From BB84 to $\fOneBlock$ and $\fOneBlockTensor$}
We define protocol $\fBB$ as follows.
\begin{defn}\label{defn:bbdf}
$\fBB(1^{1/\epsilon},1^\kappa)$ is defined to be an RSPV protocol for $\{\ket{0},\ket{1},\ket{+},\ket{-}\}$ that is complete, efficient and $\epsilon$-sound.
\end{defn}
The protocols below will build on the $\fBB$ protocol.
\begin{mdframed}[backgroundcolor=black!10]
	\begin{prtl}[$\fOneBlock$]\label{prtl:4}This is the RSPV for state family $\{\frac{1}{\sqrt{2}}(\ket{x_0}+\ket{x_1}:x_0,x_1\in \{0,1\}^m,\tHW(x_0\oplus x_1)=1,\tParity(x_0)=0)\}$.\par
		Parameters: problem size $1^m$, 
            approximation error parameter $1^{1/\epsilon}$, security parameter $1^\kappa$.\par
			Output registers: client-side classical registers $\bK^{(\tout)}=(\bx^{(\tout)}_0,\bx^{(\tout)}_1)$, where each of both holds $m$ bits; client-side classical register $\bflag$ with value in $\{\fpass,\ffail\}$; server-side quantum register $\bQ^{(\tout)}$ which holds $m$ qubits.\par
			Take $L=4(m+\kappa)$.
			\begin{enumerate}
				\item For $i\in [L]$:\begin{enumerate}
				 \item Execute $\fBB(1^{L/\epsilon},1^\kappa)$. The client stores the outcome in $\{0,1,+,-\}$ in register $\bD^{(\ttemp)(i)}$ and the server stores the outcome in $\bQ^{(\ttemp)(i)}$.
				\end{enumerate}
				The client sets $\bflag$ to be $\ffail$ if any round fails.
				\item The client randomly samples indices $i^{(1)},i^{(2)}\cdots i^{(m)}\in [L]^m$ such that there is no repetition, there is exactly one $i$ among them such that $\bD^{(\ttemp)(i)}$ is ``+'', and there is no $i$ among them such that $\bD^{(\ttemp)(i)}$ is ``-''. The client tells the server its choices and server could store the states in $\bQ^{(\ttemp)(i^{(1)})},\bQ^{(\ttemp)(i^{(2)})},\cdots \bQ^{(\ttemp)(i^{(m)})}$ in $\bQ^{(\tout)}$. Then the state in $\bQ^{(\tout)}$ could be equivalently written as
				\begin{equation}\label{eq:44}\frac{1}{\sqrt{2}}(\ket{x_0}+\ket{x_1}:x_0,x_1\in \{0,1\}^m,\tHW(x_0\oplus x_1)=1,\tParity(x_0)=0)\end{equation}
				and the client calculates $x_0,x_1$ and stores them in $\bx^{(\tout)}_0,\bx^{(\tout)}_1$.
			\end{enumerate}
	\end{prtl}
\end{mdframed}
The completeness and efficiency are from the protocol description.
\begin{prop}
	Protocol \ref{prtl:4} is $\epsilon$-sound.
\end{prop}
\begin{proof}
	By sequential composability of RSPV (Section \ref{sec:3.4.1}) there exists a simulator that simulates the outputs of the first step from the state $\rho_{tar,BB84}^{\otimes L}$ where $\rho_{tar,BB84}$ to denote the target state of $\fBB$.\par
	Then the client samples indices, calculates $\bD^{(\tout)}_0,\bD^{(\tout)}_1$ and disgards the other registers. We note that, if the client does these operations from  $\rho_{tar,BB84}^{\otimes L}$, the remaining state could be perfectly simulated from \eqref{eq:44} (up to different locations of client-side registers): the simulator prepares the state being disgarded and reverses the honest execution.\par
	Combining the two simulators above completes the proof.
\end{proof}
    \begin{mdframed}[backgroundcolor=black!10]
        \begin{prtl}[$\fOneBlockTensor$]\label{prtl:5} The state family is $$\{\frac{1}{\sqrt{2}}(\ket{x_0}+\ket{x_1}):x_0,x_1\in \{0,1\}^m,\tHW(x_0\oplus x_1)=1,\tParity(x_0)=0\}^{\otimes n}$$.\par
            Parameters: problem size $1^m,1^n$, approximation error parameter $1^{1/\epsilon}$, security parameter $1^\kappa$.\par
            Output registers: client-side classical registers $\bK=(\bK^{(i)})_{i\in [n]}$, $\bK^{(i)}=(\bx^{(i)}_0,\bx^{(i)}_1)$, where each of both holds $m$ classical bits; client-side classical register $\bflag$ with value in $\{\fpass,\ffail\}$; server-side quantum register $\bQ=(\bQ^{(i)})_{i\in [n]}$, where each of them holds $m$ qubits.
			\begin{enumerate}
				\item For each $i\in [n]$:\begin{enumerate}
				 \item Execute $\fOneBlock(1^m,1^{n/\epsilon},1^\kappa)$. The client stores the outcome in $\bx^{(i)}_0,\bx^{(i)}_1$, and the server stores the outcome in $\bQ^{(i)}$.\end{enumerate}
				 The client sets $\bflag$ to be $\ffail$ if any round fails.
			\end{enumerate}
        \end{prtl}
    \end{mdframed}
	The completeness and efficiency are from the protocol description.
	\begin{prop}
		Protocol \ref{prtl:5} is $\epsilon$-sound.
	\end{prop}
	The proof is by the sequential composition of RSPV (Section \ref{sec:3.4.1}).
    \subsection{Construction of $\fMultiBlock$}\label{sec:5.3}
	In this section we construct RSPV for states
	\begin{equation*}\{\frac{1}{\sqrt{2}}(\ket{x_0^{(1)}||x_0^{(2)}||\cdots ||x_0^{(n)}}+\ket{x_1^{(1)}||x_1^{(2)}||\cdots ||x_1^{(n)}}):\end{equation*}
\begin{equation}\label{eq:47rr}\forall i\in [n],x_0^{(i)},x_1^{(i)}\in \{0,1\}^m,\text{HW}(x_0^{(i)}\oplus x_1^{(i)})=1;\tParity(x_0^{(1)})=0\}\end{equation}
We will give the preRSPV protocol, analyze its information-theoretic core, and prove the soundness of the preRSPV and amplify it to an RSPV protocol.
	\subsubsection{PreRSPV protocol for \eqref{eq:47rr}}
    \begin{mdframed}[backgroundcolor=black!10]
        \begin{prtl}\label{prtl:8}
            The state family is \eqref{eq:47rr}. Below we construct a preRSPV for it.\par
            Parameters: problem size $1^m,1^n$,
            approximation error parameter $1^{1/\epsilon}$, security parameter $1^\kappa$. It is required that $\epsilon>11n/\sqrt{m}$.\par
			Output registers: client-side classical registers $\bK^{(\tout)}=(\bK^{(\tout)(i)})_{i\in [n]}$, $\bK^{(\tout)(i)}=(\bx^{(\tout)(i)}_0,\bx^{(\tout)(i)}_1)$, where each of both holds $m$ classical bits; client-side classical register $\bflag$ with value in $\{\fpass,\ffail\}$; server-side quantum register $\bQ=(\bQ^{(i)})_{i\in [n]}$, where each of them holds $m$ qubits.\par
			Take $\epsilon_0=\epsilon-10n/\sqrt{m}$.\par
			$\fMultiBlockTest$ is defined as:
            \begin{enumerate}
                \item Both parties run $\fOneBlockTensor(1^m,1^n,1^{1/\epsilon_0},1^\kappa)$; store the client-size output in register $\bK^{(\ttemp)}$ and the server-side output states in registers $\bQ$. The client sets $\bflag$ to be $\ffail$ if it fails.
                \item For each $i\in [2,\cdots n]$, the server evaluates $\fxorparity(1,i)$ defined as the xor of parities of $\bQ^{(1)}$ and $\bQ^{(i)}$. The server sends the value back to the client and the client stores it in $\bxor^{(i)}$.
				\item The client asks the server to measure all the $\bQ$ register on the computational basis and reveal the results. The client checks: \begin{itemize}\item For each $i\in [n]$ the server's measurement results for $\bQ^{(i)}$ is in $\bK^{(\ttemp)(i)}$.\item For each $i\in [2,\cdots n]$ the xor of parities of measurement results for $\bQ^{(1)}$ and $\bQ^{(i)}$ is the same as $\bxor^{(i)}$.\end{itemize}
				The client sets $\bflag$ to be $\ffail$ of check fails and $\fpass$ otherwise.
            \end{enumerate}\par
			$\fMultiBlockComp$ is defined as:
			\begin{enumerate}
				\item The same as the first step of $\fMultiBlockTest$.
				\item First do the same thing as the second step of $\fMultiBlockTest$.\par
				The client calculates $\bK^{(\tout)(i)}$ as follows: $\bK^{(\tout)(1)}=\bK^{(\ttemp)(1)}$. For $i\geq 2$ $(\bx^{(\tout)(i)}_0,\bx^{(\tout)(i)}_1)=(\bx^{(\ttemp)(i)}_{\bxor^{(i)}},\bx^{(\ttemp)(i)}_{1\oplus \bxor^{(i)}})$. 
			\end{enumerate}
        \end{prtl}
    \end{mdframed}
	The completeness and efficiency are from the protocol description. To prove its soundness, we first analyze its second and third steps, which is its ``information-theoretic core''.
	\subsubsection{Analysis of the information-theoretic core}\label{sec:5.3.2}
	Let's formalize a set-up that describes the information-theoretic core of Protocol \ref{prtl:8}.
	\begin{setup}\label{setup:11}
		Parameters: $1^m,1^n$.\par
		Consider the following registers as used in Protocol \ref{prtl:8}: client-side classical registers $\bK^{(\ttemp)}$ for holding the classical description of the states prepared by $\fOneBlockTensor$. Server-side quantum register $\bQ$ for holding the quantum states from $\fOneBlockTensor$. Client-side classical register $\bxor=(\bxor^{(2)},\bxor^{(3)}\cdots \bxor^{(n)})$ where each of them is a single bit.\par
		Part of the initial state is $\rho_{OBT}\in \tD(\cH_{\bK^{(\ttemp)}}\otimes\cH_{\bQ})$ where $\rho_{OBT}$ is the target state of $\fOneBlockTensor$. Use $\fsendxorparity$ to denote the operation that measures $\fxorparity(1,i)$ for each $i\in [2,n]$ and sends the results to register $\bxor$. Then the target state of Protocol \ref{prtl:8} is $\fsendxorparity(\rho_{OBT})$ (up to a change of client side descriptions). For modeling the initial states in the malicious setting, consider server-side quantum register $\bS$ and environment register $\bbE$. The initial states in the malicious setting that we consider could be described as $\rho_{OBT}\otimes\rho_0$ where $\rho_0\in \tD(\cH_{\bS}\otimes\cH_{\bbE})$.\par
		We use $(\fMultiBlockTest_{\geq 2},\fMultiBlockComp_{\geq 2})$ to denote the protocols that starts from the second step of Protocol \ref{prtl:8}. Below we introduce notations for describing each step of protocol executions.\begin{itemize}
		\item Use $\Pi_{\in \bK^{(\ttemp)}}^{\bQ}$ to denote the projection onto the space that the value of $\bQ$ is within $\bK^{(\ttemp)}$. Use $\Pi_{\fxorparity(\bQ)=\bxor}$ to denote the projection onto the space that for each $i\in [2,\cdots n]$ the xor of parities of $\bQ^{(1)}$ and $\bQ^{(i)}$ is the same as the value in $\bxor^{(i)}$. Thus the passing condition in the third step of $\fMultiBlockTest$ corresponds to\footnote{In other words, if the adversary could pass the checking in the third step of $\fMultiBlockTest$ with high probability, then consider the state before the adversary sends back the measurement results, this state is approximately in the space of $\Pi_{\in \bK^{(\ttemp)}}^{\bQ}\Pi_{\fxorparity(\bQ)=\bxor}$.} $\Pi_{\in \bK^{(\ttemp)}}^{\bQ}\Pi_{\fxorparity(\bQ)=\bxor}$.
		\item Suppose $\fAdv_2,\fAdv_3$ are two superoperators operated on $\bQ,\bS$, which describe the adversary's operations in the second and third steps; we could purify them to be unitaries on $\bQ,\bS$ and some reference registers. Use $\fsend_{\text{to }\bxor}$ to denote the message sending operation that sends the outcome of $\fAdv_2$ to $\bxor$.
	\end{itemize}\par
		For purifying the density operators, consider reference registers $\bR_{\bK^{(\ttemp)}},\bR_{\bxor}$. Use $\ket{\varphi_{OBT}}$ to denote the purification of $\rho_{OBT}$ where $\bK^{(\ttemp)}$ is purified by $\bR_{\bK^{(\ttemp)}}$.
	\end{setup}
We will show that if the adversary could pass in the test mode, the output state in the comp mode could be statistically simulated from the target state (we could simply work on $\fsendxorparity(\rho_{OBT})$, which is the same as the target state up to a client-side change of representation of descriptions). Below we define the simulator and the corresponding statement is Corollary \ref{cor:5.5}. Corollary \ref{cor:5.5} is a corollary of Proposition \ref{prop:5.3}, which relates the real execution with simulation using low-level descriptions. To prove Proposition \ref{prop:5.3}, we first prove Lemma \ref{lem:5.4}. These statements are stated and prove below.   
\begin{defn}\label{defn:simcons}
	Under Set-up \ref{setup:11}, for any $\fAdv_2$, define $\fSim$ as follows:
		\begin{enumerate}
			\item Apply $\fAdv_2$.
			\item Instead of doing $\fsend_{\text{to }\bxor}$, the simulator simply copies\footnote{We mean bit-wise CNOT.} the response to a temporary register $\bsimxor$. Denote this step as $\fsend_{\text{to }\bsimxor}$.\par
			Then $\bsimxor$ is disgarded. Denote this step as $\fDisgard_{\bsimxor}$.
		\end{enumerate}
\end{defn}
	\begin{prop}\label{prop:5.3}Under Set-up \ref{setup:11}, for each $\rho_0\in \tD(\cH_{\bS}\otimes\cH_{\bbE})$, $\fAdv_2,\fAdv_3$, there is
		 \begin{align}&\Pi_{\in \bK^{(\ttemp)}}^{\bQ}\Pi_{\fxorparity(\bQ)=\bxor}\circ\fAdv_3\circ\fsend_{\text{to }\bxor}\circ\fAdv_2(\rho_{OBT}\otimes\rho_0)\label{eq:462}\\\approx_{4n/\sqrt{m}}&\Pi_{\in \bK^{(\ttemp)}}^{\bQ}\circ\fDisgard_{\bsimxor}\circ\Pi_{\bxor=\bsimxor}\circ\fAdv_3\circ\fsend_{\text{to }\bsimxor}\circ\fAdv_2(\fsendxorparity(\rho_{OBT})\otimes\rho_0)\label{eq:472}\end{align}
		 Here $\Pi_{\bxor=\bsimxor}$ denotes the projection onto the space that the value of $\bxor$ is the same as the value of $\bsimxor$.
	\end{prop}
	The form of Proposition \ref{prop:5.3} strikes a balance between what is easy to prove and what we want. What we want is something like ``if the test mode passes, the output state of the comp mode is the same as the output of $\fSim$''. Such a statement is given in Corollary \ref{cor:5.5} which will be proved as a corollary of Proposition \ref{prop:5.3}. So what does Proposition \ref{prop:5.3} mean? We first note in \eqref{eq:472} there are indeed the operators used in the construction of $\fSim$. Then based on the operations of $\fSim$, \eqref{eq:472} continues to apply $\fAdv_3$ and $\Pi_{\in \bK^{(\ttemp)}}^{\bQ}$, and insert a projection $\Pi_{\bxor=\bsimxor}$, which makes it a simulation of a subspace of the test mode. Proposition \ref{prop:5.3} does not rely on the condition that the test mode passes with high probability; in Corollary \ref{cor:5.5} we will combine this condition with Proposition \ref{prop:5.3} and remove these extra operators to get a simulation of the comp mode.\par 
	We will purify and decompose the state to prove Proposition \ref{prop:5.3}. To do this let's first prove a lemma.
	\begin{lem}\label{lem:5.4}
		For any normalized pure state $\ket{\varphi_0}\in \cH_{\bS}\otimes\cH_{\bbE}$, any operator $O$ that operates on $\bQ,\bS$ and could be written as unitaries and projections, there is
		\begin{equation}\label{eq:48ne}\Pi_{\in \bK^{(\ttemp)}}^{\bQ}O(\ket{\varphi_{OBT}}\otimes\ket{\varphi_0})\approx_{2n/\sqrt{m}}\sum_{\vec{b}\in \{0,1\}^n}\Pi_{\vec{\bx}_{\vec{b}}^{(\ttemp)}}^{\bQ}O\Pi_{\vec{\bx}_{\vec{b}}^{(\ttemp)}}^{\bQ}(\ket{\varphi_{OBT}}\otimes\ket{\varphi_0})\end{equation}
		where $\Pi_{\vec{\bx}_{\vec{b}}^{(\ttemp)}}^{\bQ}$ denotes the projection onto the space that the value of $\bQ$ is the same as $\bx^{(\ttemp)(1)}_{b^{(1)}}||\bx^{(\ttemp)(2)}_{b^{(2)}}\cdots \bx^{(\ttemp)(n)}_{b^{(n)}}$, where we use $b^{(1)}b^{(2)}\cdots b^{(n)}$ to denote the coordinates of $\vec{b}$.
	\end{lem}
	A variant of this lemma where the states are computationally indistinguishable is used in \cite{cvqcinlt}. Below we give the proof.
	\begin{proof}[Proof of Lemma \ref{lem:5.4}]
		We consider a sequence of states in the following form, for each $i\in [0,n]$:
		\begin{equation}\label{eq:49ne}\Pi_{\in \bK^{(\ttemp)(>t)}}^{\bQ^{(>t)}}\sum_{\vec{b}\in \{0,1\}^t}\Pi_{\vec{\bx}_{\vec{b}}^{(\ttemp)(\leq t)}}^{\bQ^{(\leq t)}}O\Pi_{\vec{\bx}_{\vec{b}}^{(\ttemp)(\leq t)}}^{\bQ^{(\leq t)}}(\ket{\varphi_{OBT}}\otimes\ket{\varphi_0})\end{equation}
		where $\Pi_{\vec{\bx}_{\vec{b}}^{(\ttemp)(\leq t)}}^{\bQ^{(\leq t)}}$ denotes the projection onto the space that the value of $\bQ^{(\leq t)}$ is the same as $\bx^{(\ttemp)(1)}_{b^{(1)}}||\bx^{(\ttemp)(2)}_{b^{(2)}}\cdots \bx^{(\ttemp)(t)}_{b^{(t)}}$, where we use $b^{(1)}b^{(2)}\cdots b^{(t)}$ to denote the coordinates of $\vec{b}$, and $\Pi_{\in \bK^{(\ttemp)(>t)}}^{\bQ^{(>t)}}$ denotes the projection onto the space that the value of $\bQ^{(> t)}$ is within $\bK^{(\ttemp)(> t)}$. Thus for $t=0$ we get the left hand side of \eqref{eq:48ne} and for $t=n$ we get the right hand side of \eqref{eq:48ne}.\par
		We only need to prove the difference of \eqref{eq:49ne} has difference at most $\frac{2}{\sqrt{m}}$ for between $t-1,t$ for each $t\in [n]$. The subtraction of these two states gives:
		\begin{equation*}\Pi_{\in \bK^{(\ttemp)(>t)}}^{\bQ^{(>t)}}\sum_{\vec{b}\in \{0,1\}^{t-1}}\Pi_{=\bx_{1}^{(\ttemp)(t)}}^{\bQ^{(t)}}\Pi_{\vec{\bx}_{\vec{b}}^{(\ttemp)(< t)}}^{\bQ^{(< t)}}O\Pi_{\vec{\bx}_{\vec{b}}^{(\ttemp)(< t)}}^{\bQ^{(< t)}}\Pi_{=\bx_{0}^{(\ttemp)(t)}}^{\bQ^{(t)}}(\ket{\varphi_{OBT}}\otimes\ket{\varphi_0})\end{equation*}
		\begin{equation}\label{eq:50ne}+\Pi_{\in \bK^{(\ttemp)(>t)}}^{\bQ^{(>t)}}\sum_{\vec{b}\in \{0,1\}^{t-1}}\Pi_{=\bx_{0}^{(\ttemp)(t)}}^{\bQ^{(t)}}\Pi_{\vec{\bx}_{\vec{b}}^{(\ttemp)(< t)}}^{\bQ^{(< t)}}O\Pi_{\vec{\bx}_{\vec{b}}^{(\ttemp)(< t)}}^{\bQ^{(< t)}}\Pi_{=\bx_{1}^{(\ttemp)(t)}}^{\bQ^{(t)}}(\ket{\varphi_{OBT}}\otimes\ket{\varphi_0})\end{equation}
		each of them could be understood as guessing one of $\bx^{(\ttemp)(t)}$ given the other. The relation between $\bx^{(\ttemp)(t)}_0,\bx^{(\ttemp)(t)}_1$ is that they differ in one bit, thus there are $m$ choices that are equally possible. Thus the norm of \eqref{eq:50ne} could be upper bounded by $2/\sqrt{m}$, which completes the proof.
	\end{proof}
	\begin{proof}[Proof of Proposition \ref{prop:5.3}]
		Consider the purification of Proposition \ref{prop:5.3}. Then $\rho_{0}$ is replace by a pure state $\ket{\varphi_0}$ in $\cH_{\bS}\otimes\cH_{\bbE}$, $\rho_{OBT}$ is replaced by $\ket{\varphi_{OBT}}$, $\fAdv$s are considered to be unitaries, $\fsend$ operators will also copy the information to the corrsponding reference register (corresponding to $\bxor$, it's $\bR_{\bxor}$). The partial trace $\fDisgard_{\bsimxor}$ is handled as follows: the purification simply removes it from \eqref{eq:472} and requires that the purification of both sides of \eqref{eq:462}\eqref{eq:472} could be transformed to each other by a local operaion on $\bR_{\bxor}$ and ${\bsimxor}$ (in other words, we do not require the purified states to be close to each other; what we need is only the closeness when $\bR_{\bxor}$ and ${\bsimxor}$ are all traced out).\par
		Thus \eqref{eq:462} becomes
		\begin{equation}\label{eq:512}\Pi_{\in \bK^{(\ttemp)}}^{\bQ}\Pi_{\fxorparity(\bQ)=\bxor}\fAdv_3\fsend_{\text{to }\bxor}\fAdv_2(\ket{\varphi_{OBT}}\otimes\ket{\varphi_0})\end{equation}
		and \eqref{eq:472} becomes
		\begin{equation}\label{eq:522}\Pi_{\in \bK^{(\ttemp)}}^{\bQ}\Pi_{\bxor=\bsimxor}\fAdv_3\fsend_{\text{to }\bsimxor}\fAdv_2(\fsendxorparity\ket{\varphi_{OBT}}\otimes\ket{\varphi_0})\end{equation}
		Applying Lemma \ref{lem:5.4} to \eqref{eq:512}\eqref{eq:522}, we get that \eqref{eq:512} is $(2n/\sqrt{m})$-close to
		\begin{equation}\label{eq:532}
			\sum_{\vec{b}\in \{0,1\}^n}\Pi_{\vec{\bx}_{\vec{b}}^{(\ttemp)}}^{\bQ}\Pi_{\fxorparity(\bQ)=\bxor}\fAdv_3\fsend_{\text{to }\bxor}\fAdv_2\Pi_{\vec{\bx}_{\vec{b}}^{(\ttemp)}}^{\bQ}(\ket{\varphi_{OBT}}\otimes\ket{\varphi_0})
		\end{equation}
		and \eqref{eq:522} is $(2n/\sqrt{m})$-close to
		\begin{equation}\label{eq:542}
			\sum_{\vec{b}\in \{0,1\}^n}\Pi_{\vec{\bx}_{\vec{b}}^{(\ttemp)}}^{\bQ}\Pi_{\bxor=\bsimxor}\fAdv_3\fsend_{\text{to }\bsimxor}\fAdv_2(\fsendxorparity\Pi_{\vec{\bx}_{\vec{b}}^{(\ttemp)}}^{\bQ}\ket{\varphi_{OBT}}\otimes\ket{\varphi_0})
		\end{equation}
		We could compare \eqref{eq:532} and \eqref{eq:542} directly. Both of them are a summation of $2^n$ different branches so we could compare each branch of them. On each branch given by $\Pi_{\vec{\bx}_{\vec{b}}^{(\ttemp)}}^{\bQ}$, $\fxorparity(\bQ)$ has a fixed value; for simplicity let's denote it as $\alpha$. In \eqref{eq:532} $\fsend$ sends a value to $\bxor$ (which is not necessarily $\alpha$ but contains a sub-branch with value $\alpha$) and finally a projection onto $\alpha$ is performed (notice that we could assume \eqref{eq:532} first do $\Pi_{\vec{\bx}_{\vec{b}}^{(\ttemp)}}^{\bQ}$ and then do $\Pi_{\fxorparity(\bQ)=\bxor}$). In \eqref{eq:542} $\fsendxorparity$ writes $\alpha$ to $\bxor$, and $\fsend$ sends a value to $\bsimxor$ (which is not necessarily $\alpha$ but contains a sub-branch with value $\alpha$) and finally a projection $\Pi_{\bxor=\bsimxor}$ is performed. We could see the only difference is the usage of $\bsimxor$; thus \eqref{eq:532}\eqref{eq:542} are the same up to a local operation on $\bsimxor,\bR_{\bxor}$.\par
		This completes the proof.
	\end{proof}
	\begin{cor}\label{cor:5.5}
	Under Set-up \ref{setup:11}, $(\fMultiBlockTest_{\geq 2},\fMultiBlockComp_{\geq 2})$ is $(n/\sqrt{m},10n/\sqrt{m})$-sound.
	\end{cor}
	\begin{proof}[Proof of Corollary \ref{cor:5.5}]
		For any adversary $\fAdv$ and $\rho_0\in \tD(\cH_{\bS}\otimes\cH_{\bbE})$, $\fMultiBlockTest_{\geq 2}$ passes with probability $\geq 1-n/\sqrt{m}$ translates to:
		\begin{equation}\label{eq:552}\tr(\Pi_{\in \bK^{(\ttemp)}}^{\bQ}\Pi_{\fxorparity(\bQ)=\bxor}\circ\fAdv_3\circ\fsend_{\text{to }\bxor}\circ\fAdv_2(\rho_{OBT}\otimes\rho_0))\geq 1-n/\sqrt{m}\end{equation}
		This is as given in \eqref{eq:462}, which by Proposition \ref{prop:5.3} is $(4n/\sqrt{m})$-close to \eqref{eq:472}. Both \eqref{eq:462}\eqref{eq:472} has a form of projecting a trace-1 operator to a subspace; this fact, \eqref{eq:462}\eqref{eq:472} and \eqref{eq:552} imply that the states before the projection are also close to each other:\footnote{The details are as follows. We could first use \eqref{eq:552} and \eqref{eq:462}\eqref{eq:472} to show the trace of \eqref{eq:472} is $(5n/\sqrt{m})$-close to 1, which implies that \eqref{eq:472} is $(5n/\sqrt{m})$-close to the state before the projection. Combining it with \eqref{eq:552} and \eqref{eq:462}\eqref{eq:472} completes the proof.}
		\begin{equation}
			\fAdv_3\circ\fsend_{\text{to }\bxor}\circ\fAdv_2(\rho_{OBT}\otimes\rho_0)\approx_{10n/\sqrt{m}}\fAdv_3\circ\fSim(\fsendxorparity(\rho_{OBT})\otimes\rho_0)
		\end{equation}
		where we encapsulate \eqref{eq:472} using the construction of $\fSim$ in Definition \ref{defn:simcons}.\par
		Inverting $\fAdv_3$ completes the proof.
	\end{proof}
		\subsubsection{Compilation and amplification to preRSPV and RSPV}\label{sec:5.3.3}
		We could first prove the soundness of Protocol \ref{prtl:8} using Corollary \ref{cor:5.5} and the soundness of $\fOneBlockTensor$. (see \ref{sec:3.4.2}).
		\begin{thm}
			Protocol \ref{prtl:8} is $(\epsilon-11n/\sqrt{m},\epsilon)$-sound.
		\end{thm}
		\begin{proof}
			By the soundness of $\fOneBlockTensor$ the output state of the first step of Protocol \ref{prtl:8} could be simulated from $\rho_{OBT}$ with approximation error $\epsilon_0=\epsilon-10n/\sqrt{m}$. This together with the condition that $\fMultiBlockTest$ passes with probability $\geq 1-(\epsilon-11n/\sqrt{m})$ implies that $\fMultiBlockTest_{\geq 2}$ passes with probability $\geq 1-n/\sqrt{m}$ (against the corresponding adversary). By Corollary \ref{cor:5.5} $\fMultiBlockComp_{\geq 2}$ is simulated by $\fSim$ with approximation error $10n/\sqrt{m}$. Combining $\fSim$ with the simulator for the first step of Protocol \ref{prtl:8} gives a simulator for $\fMultiBlockComp$ (up to a change of client-side representation of descriptions). This completes the proof.
		\end{proof}
		Then we could amplify Protocol \ref{prtl:8} (preRSPV for \eqref{eq:47rr}) to an RSPV protocol:
	\begin{mdframed}[backgroundcolor=black!10]
		\begin{prtl}[$\fMultiBlock$]
			The set-up is the same as Protocol \ref{prtl:8}.\par
			The protocol is the amplification procedure (Protocol \ref{prtl:1}) running on Protocol \ref{prtl:8}.
		\end{prtl}
	\end{mdframed}
	The completeness, efficiency and soundness are from the properties of Protocol \ref{prtl:8} and the preRSPV-to-RSPV amplification procedure (Section \ref{sec:preRSPVdef}).
    \subsection{Construction of $\fKP$}\label{sec:5.4}
	\begin{mdframed}[backgroundcolor=black!10]
		To formalize the protocol, define function $\tutob:\{0,1\}^m\rightarrow\{0,1\}^{\log_2(m)}$ as follows: given a string in the form of $000\cdots 1000\cdots$, it outputs the binary representation for the locations of number $1$.
		\begin{prtl}[$\fKP$]\label{prtl:10} This is the RSPV for state family $\{\frac{1}{\sqrt{2}}(\ket{0}\ket{x_0}+\ket{1}\ket{x_1}):x_0,x_1\in \{0,1\}^n\}$.\par
		Parameters: problem size $1^n$, approximation error parameter $1^{1/\epsilon}$, security parameter $1^\kappa$.\par
			Output registers: client-side classical registers $\bK^{(\tout)}=(\bx^{(\tout)}_0,\bx^{(\tout)}_1)$, where each of both holds $n$ bits; client-side classical register $\bflag$ with value in $\{\fpass,\ffail\}$; server-side quantum register $\bQ^{(\tout)}=(\bQ^{(\tout)(\tsubs)},\bQ^{(\tout)(\tkey)})$ which hold $1$ qubit and $n$ qubits correspondingly.\par
			Take $n_0=2n$, $m_0$ to be the smallest power of 2 such that $m>(12n_0/\epsilon)^2$.\par
			\begin{enumerate}
				\item Execute $\fMultiBlock(1^{m_0},1^{n_0},1^{1/\epsilon},1^\kappa)$. The client stores the outcome in $\bK^{(\ttemp)}=(\bK^{(\ttemp)(i)})_{i\in [n_0]},\bK^{(\ttemp)(i)}=(\bx_0^{(\ttemp)(i)},\bx_1^{(\ttemp)(i)})$ and the server stores the outcome in $\bQ^{(\ttemp)}$. In the honest setting the server holds\begin{equation*}\{\frac{1}{\sqrt{2}}(\ket{x_0^{(\ttemp)(1)}||x_0^{(\ttemp)(2)}||\cdots ||x_0^{(\ttemp)(2n)}}+\ket{x_1^{(\ttemp)(1)}||x_1^{(\ttemp)(2)}||\cdots ||x_1^{(\ttemp)(2n)}}):\end{equation*}
				\begin{equation}\label{eq:47af}\forall i\in [2n],x_0^{(\ttemp)(i)},x_1^{(\ttemp)(i)}\in \{0,1\}^m,\text{HW}(x_0^{(\ttemp)(i)}\oplus x_1^{(\ttemp)(i)})=1;\tParity(x_0^{(\ttemp)(1)})=0\}\end{equation}
				and the client holds all these keys.\par
				The client sets $\bflag$ to be $\ffail$ if this step fails.
				\item For each $i\in [n_0/2]$, the client sends the following information to the server:
				\begin{itemize}
					\item $x^{(\ttemp)(2i-1)}_0,x^{(\ttemp)(2i)}_1$;
					\item The bits of $\tutob(x^{(\ttemp)(2i-1)}_0\oplus x^{(\ttemp)(2i-1)}_1)$ excluding the first bits; the bits of $\tutob(x^{(\ttemp)(2i)}_0\oplus x^{(\ttemp)(2i)}_1)$ excluding the first bits. (The length of this part is $2(\log_2(m_0)-1)$.)
				\end{itemize}
				With these information, the server could do the following transformation on the state:\par
				It first calculates and stores the parity of $\bQ^{(\ttemp)(1)}$ in $\bQ^{(\tout)(\tsubs)}$.\par
				Then for each $i$:
				\begin{enumerate}\item It xors $x^{(\ttemp)(2i-1)}_0,x^{(\ttemp)(2i)}_1$ to each block to transform each block into the form of $000\cdots 1000\cdots$;\item Then it transforms unary representation to binary representation for the location of non-zero bits.\item Then it only keeps the first bits for each block and transform the remaining bits to $0$ using the second part of the client-side messages.\end{enumerate} Denote the first bit of $\tutob(x^{(\ttemp)(2i)}_0\oplus x^{(\ttemp)(2i)}_1)$ as $b^{(out)(i)}_0$ and denote the first bit of $\tutob(x^{(\ttemp)(2i-1)}_0\oplus x^{(\ttemp)(2i-1)}_1)$ as $b^{(out)(i)}_1$, the server-side state in the end is
				\begin{equation}\label{eq:57fn}\frac{1}{\sqrt{2}}(\underbrace{\ket{0}}_{\bQ^{(\tout)(\tsubs)}}\underbrace{\ket{b_0^{(out)(1)}||b_0^{(out)(2)}||\cdots b_0^{(out)(n)}}}_{\bQ^{(\tout)(\tkey)}}+\ket{1}\ket{b_1^{(out)(1)}||b_1^{(out)(2)}||\cdots b_1^{(out)(n)}})\end{equation}
				Note that for each $i$, $b_0^{(out)(i)},b_1^{(out)(i)}$ are uniformly random bits from $\{0,1\}^2$. The client then calculates and stores $(b_0^{(out)(1)}||b_0^{(out)(2)}||\cdots b_0^{(out)(n)},b_1^{(out)(1)}||b_1^{(out)(2)}||\cdots b_1^{(out)(n)})$ in register $(\bx_0^{(\tout)},\bx_1^{(\tout)})$.
			\end{enumerate}
		\end{prtl}
	\end{mdframed}
	The completeness and efficiency are from the protocol description. Below we state and prove the soundness.
	\begin{thm}
		Protocol \ref{prtl:10} is $\epsilon$-sound.
	\end{thm}
	\begin{proof}
		By the soundness of $\fMultiBlock$ there exists a simulator that simulates the output of the first step from \eqref{eq:47af}.\par
	 Then the client reveals lots of information to allow the server to transform the state to \eqref{eq:57fn}. Note that the information revealed by the client are all disgarded by the client; the client only calculates and keeps $\bx_0^{(\tout)},\bx_1^{(\tout)}$. Then starting from \eqref{eq:57fn}, the simulator could simulate the disgarded information on its own and reverse the transformation to to simulate the joint state corresponding to \eqref{eq:47af} (up to a cnange of locations of client-side registers).\par
		Combining the two simulators above completes the proof.
	\end{proof}
    \subsection{Construction of $\fQFac$ (RSPV for $\ket{+_\theta}$)}\label{sec:5.5}
	In this section we construct RSPV for $\ket{+_{\theta}}$ from $\frac{1}{\sqrt{2}}(\ket{0}\ket{x_0}+\ket{1}\ket{x_1})$. We will give the preRSPV-with-score protocol for it, analyze its information-theoretic core, and prove the soundness of the preRSPV and amplify it to an RSPV protocol.
	\subsubsection{PreRSPV-with-score for $\ket{+_\theta}$}
	\begin{mdframed}[backgroundcolor=black!10]
		\begin{prtl}\label{prtl:11}
			Below we construct a preRSPV with the score for $\ket{+_\theta},\theta\in \{0,1,2\cdots 7\}$.\par
			Parameters: a temporary approximation error parameter $1^{1/\epsilon_0}$, security parameter $1^\kappa$.\par
			Output registers:\begin{itemize}\item client-side classical register $\btheta$ with value in $\{0,1,2\cdots 7\}$. We also say $\btheta=(\btheta_1,\btheta_2,\btheta_3)$ where each of them is a classical bit and $\btheta=4\btheta_1+2\btheta_2+\btheta_3$.
				\item Client-side classical register $\bflag$ with value in $\{\fpass,\ffail\}$; client-side classical register $\bscore$ with value in $\{\fwin,\flose,\perp\}$.
				\item Server side quantum register $\bq$ which holds a single qubit.
			\end{itemize}\par
			Take $n=\kappa$.\par
			$\fQFacTest$ is defined as:
			\begin{enumerate}
				\item Both parties run $\fKP(1^n,1^{1/\epsilon_0},1^\kappa)$; store the client-size output in register $\bK=(\bx_0,\bx_1)$ and the server-side output states in registers $(\bq,\bQ)$:
				\begin{equation}\label{eq:59n}\frac{1}{\sqrt{2}}(\underbrace{\ket{0}}_{\bq}\underbrace{\ket{x_0}}_{\bQ}+\ket{1}\ket{x_1})\end{equation} The client sets $\bflag$ to be $\ffail$ if it fails.
				\item In \eqref{eq:59n}, denote the first two bits of $x_0$ as $b^{(1)}_0$ and $b^{(2)}_0$, denote the first two bits of $x_1$ as $b^{(1)}_1$ and $b^{(2)}_1$. The server could do control-phase operations to add the following phases to the $\bq$ register:
				\begin{align}
					&\frac{1}{\sqrt{2}}(\ket{0}\ket{x_0}+\ket{1}\ket{x_1})\\
					\rightarrow&\frac{1}{\sqrt{2}}(e^{\mi\pi(2b_0^{(1)}+b_0^{(2)})/4}\ket{0}\ket{x_0}+e^{\mi\pi(2b_1^{(1)}+b_1^{(2)})/4}\ket{1}\ket{x_1})\\
					=&(\frac{1}{\sqrt{2}}(\ket{0}\ket{x_0}+e^{\mi\pi(2(b_0^{(1)}\oplus b_1^{(1)})+(b_0^{(2)}\oplus b_1^{(2)}))/4}\ket{1}\ket{x_1}))\cdot\text{global phase}\label{eq:62n}
				\end{align}
				The client stores $b_0^{(1)}\oplus b_1^{(1)}$ in register $\btheta^{(2)}$ and stores $b_0^{(2)}\oplus b_1^{(2)}$ in register $\btheta_3$.\par
				Then the server does Hadamard measurements for each bit in $\bQ$; suppose the measurement result is $d$. This transforms \eqref{eq:62n} to
				$$(\frac{1}{\sqrt{2}}(\ket{0}+e^{\mi\pi(4d\cdot(x_0\oplus x_1)+2(b_0^{(1)}\oplus b_1^{(1)})+(b_0^{(2)}\oplus b_1^{(2)}))/4}\ket{1}))\cdot\text{global phase}$$
				The client stores $d\cdot(x_0\oplus x_1)\mod 2$ in register $\btheta_1$; the client sets $\bflag$ to be $\ffail$ if $d=0^n$ and $\fpass$ otherwise. The client keeps $\btheta,\bflag,\bscore$ and disgards all the other registers (including $\bK$ and $d$).
				\item The client randomly samples $\varphi\leftarrow \{0,1\cdots 7\}$ and sends $\varphi$ to the server.\par
				The server measures $\bq$ on basis $\ket{+_\varphi}$ and $\ket{+_{\varphi+4}}$ and gets a measurement result $r$: $r=0$ if the result is the former and $r=1$ if the result is the latter. The server sends back the measurement results $r$.\par
				The client sets the $\bflag$ register as follows:
				\begin{itemize}
					\item If $\theta-\varphi=0$, set $\bflag$ to be $\fpass$ if $r=0$ and $\ffail$ otherwise.
					\item If $\theta-\varphi+4=0$ (modulo $8$), set $\bflag$ to be $\fpass$ if $r=1$ and $\ffail$ otherwise.
					\item 
					In other cases, simply set $\bflag$ to be $\fpass$.
				\end{itemize}\par
				The client sets the $\bscore$ register as follows:
				\begin{itemize}
					\item If $|\theta-\varphi|\leq 1$ (where the distance is modulo $8$), set $\bscore$ to be $\fwin$ if $r=0$ and $\flose$ otherwise.
					\item If $|\theta-\varphi+4|\leq 1$ (where the distance is modulo $8$), set $\bscore$ to be $\fwin$ if $r=1$ and $\flose$ otherwise.
					\item If $|\theta-\varphi|=2$, simply set $\bscore$ to be $\fwin$.
				\end{itemize}
			\end{enumerate}\par
			$\fQFacComp$ is defined as:
			\begin{enumerate}
				\item Same as the first step of $\fQFacTest$.
				\item Same as the first step of $\fQFacTest$.
			\end{enumerate}
		\end{prtl}
	\end{mdframed}
	Note that the approximation error $\epsilon_0$ is only the approximation error for the first step and not the approximation error for the whole protocol. This makes it easier to tune the parameters in later proofs and constructions.\par 
	The completeness and efficiency are from the protocol description. The honest server wins in $\fQFacTest$ with probability $\frac{1}{2}+\frac{1}{2}\cos^2(\pi/8)$.
	\subsubsection{Analysis of the information-theoretic core}
	Let's formalize the set-up for the information-theoretic core of Protocol \ref{prtl:11}.
	\begin{setup}\label{setup:12}Parameter: $1^n$.\par
		$\OPT=\frac{1}{2}+\frac{1}{2}\cos^2(\pi/8)$.\par
		Consider the folowing registers, as used in the honest execution: client-side classical register $\bK=(\bx_0,\bx_1)$ where each of them holds $n$ bits, client-side classical register $\btheta$ with value in $\{0,1\cdots 7\}$, server-side quantum register $\bq,\bQ$ which hold 1 qubit and $n$ qubits each.\par
		Part of the initial state is $\rho_{KP}\in \tD(\cH_{\bK}\otimes\cH_{\bq,\bQ})$ where $\rho_{KP}$ is the target state of $\fKP$. For modeling the initial states in the malicious setting, consider server-side quantum register $\bS$ and environment register $\bbE$. The initial states in the malicious setting that we consider could be described as $\rho_{KP}\otimes\rho_{\text{ini}}$ where\footnote{Previously we use $\rho_0$ for this part but below we need to define $\rho_{\theta}$ so we change the notation here to avoid conflicts.} $\rho_{\text{ini}}\in \tD(\cH_{\bS}\otimes\cH_{\bbE})$.\par
		We use $(\fQFacTest_{\geq 2},\fQFacComp_{\geq 2})$ to denote the protocols that starts from the second step of Protocol \ref{prtl:11}, and use $\fQFacTest_{\text{step }3}$ to denote the third step of the $\fQFacTest$ protocol.\par
		 We use $\fAdv_{\text{step }2}$ to denote the adversary's operation on the second step of the protocols; then corresponding to $\rho_{KP}\otimes\rho_{\text{ini}}$ and $\fAdv_{\text{step }2}$, the joint state on the passing space after the step 2 of Protocol \ref{prtl:11} could be denoted as $$\underbrace{\ket{\fpass}\bra{\fpass}}_{\bflag}\sum_{\theta\in \{0,1\cdots 7\}}(\underbrace{\ket{\theta}\bra{\theta}}_{\btheta}\otimes\underbrace{\rho_{\theta}}_{\bq,\bQ,\bS,\bbE})$$ 
		 In other words, $\rho_{\theta}$ is the component of the output state of the step 2 where $\btheta$ is in value $\theta$ and $\bflag$ is in value $\fpass$.\par
		  Use $\fAdv_{\text{step }3}=(\fAdv_{\varphi})_{\varphi\in \{0,1\cdots 7\}}$ to denote the adversary's operation on the third step where $\fAdv_{\varphi}$ corresponds to the client's question $\varphi$.
	\end{setup}
	As discussed before, an important notion on the states after the second step is that the states $(\rho_{\theta})_{\theta\in \{0,1\cdots 7\}}$ satisfies a condition called \emph{(information-theoretic) basis blindness}, as follows.
	\begin{defn}[(Information-theoretic) basis blindness \cite{qfactory}]\label{defn:5.3}
		We say $(\rho_{\theta})_{\theta\in \{0,1\cdots 7\}}$ has (Information-theoretic) basis blindness if $\forall \theta_{2,3}\in \{0,1,2,3\},\frac{1}{2}(\rho_{\theta_{2,3}}+\rho_{\theta_{2,3}+4})\approx_{\fneg(n)}\frac{1}{8}\sum_{\theta\in \{0,1\cdots 7\}}\rho_{\theta}$
	\end{defn}
	\begin{lem}\label{lem:bst}
		$(\rho_{\theta})_{\theta\in \{0,1\cdots 7\}}$ generated in Set-up \ref{setup:12} has information-theoretic basis blindness.
	\end{lem}
	\begin{proof}
		This is equivalent to prove that any adversary working on $\rho_{KP}\otimes\rho_{\text{ini}}$ that does not have access to the client-side keys $\bK$ could only predict $(b_0^{(1)}\oplus b_1^{(1)},b_0^{(2)}\oplus b_1^{(2)})$ with probability $\frac{1}{4}+\fneg(n)$. To calculate the probability that an operation predicts this information from $\rho_{KP}\otimes\rho_{\text{ini}}$, we could write $\rho_{KP}$ as $(\Pi^{\bQ}_{=\bx_0}+\Pi^{\bQ}_{=\bx_1})\rho_{KP}(\Pi^{\bQ}_{=\bx_0}+\Pi^{\bQ}_{=\bx_1})$ to expand the expression for this probability, then this probability could be upper bounded by the sum of the following terms:
		\begin{itemize}
			\item The probability\footnote{We mean the trace of the post-projection state onto the space that the event happens} that the operation operating on $\Pi^{\bQ}_{=\bx_0}\rho_{KP}\Pi^{\bQ}_{=\bx_0}$ could predict $(b_0^{(1)}\oplus b_1^{(1)},b_0^{(2)}\oplus b_1^{(2)})$.
			\item The probability that the operation operating on $\Pi^{\bQ}_{=\bx_1}\rho_{KP}\Pi^{\bQ}_{=\bx_1}$ could predict $(b_0^{(1)}\oplus b_1^{(1)},b_0^{(2)}\oplus b_1^{(2)})$.
			\item The norm that the operation starting from $\ket{0}\ket{x_0}$ could predict $x_1$.
			\item The norm that the operation starting from $\ket{0}\ket{x_1}$ could predict $x_0$.
		\end{itemize}
		The last two terms are negligibly small and the first two terms sum to $\frac{1}{4}$, which completes the proof.
	\end{proof}
	Then by \cite{GVRSP,qfactory} we have the optimality of $\OPT$ and the  self-testing property for the test used in $\fQFacTest_{\text{step }3}$ given the condition that the input state satisfies information-theoretic basis blindness.
	Below we directly state it as the soundness of $(\fQFacTest_{\geq 2},\fQFacComp_{\geq 2})$ under Set-up \ref{setup:12}.
	\begin{thm}\label{thm:qfacbeforec}
		Under Set-up \ref{setup:12}, $(\fQFacTest_{\geq 2},\fQFacComp_{\geq 2})$ has $(\delta, \fpoly_1(\delta))$-optimal winning probability $\OPT$ and is $(\delta, \fpoly_2(\delta))$-sound for all $\delta$.
	\end{thm}
	See Appendix \ref{app:p1} for details.
	\subsubsection{Compilation and amplification to preRSPV-with-score and RSPV}
	The remaining steps are similar to what we did in Section \ref{sec:5.3.3}.
	\begin{thm}
		For Protocol \ref{prtl:11}, $\OPT$ is $(\delta-\epsilon_0,\fpoly_1(\delta)+\epsilon_0)$-optimal.
	\end{thm}
	\begin{thm}\label{thm:5.13r}
		Protocol \ref{prtl:11} is $(\delta-\epsilon_0,\fpoly_2(\delta)+\epsilon_0)$-sound for all $\delta$.
	\end{thm}
	\begin{proof}
		The proofs of the two theorems above are by combining the soundness of $\fKP$ (which leads to an approximation error $\epsilon_0$) and Theorem \ref{thm:qfacbeforec}.
	\end{proof}
	Below we give $\fQFac$, our RSPV for $\ket{+_\theta}$ states.
	\begin{mdframed}[backgroundcolor=black!10]
		\begin{prtl}[$\fQFac$]
			Parameters: approximation error parameter $1^{1/\epsilon}$, security parameter $1^\kappa$.\par
			Output registers: client-side classical register $\btheta$ with value in $\{0,1\cdots 7\}$, client-side classical register $\bflag$ with value in $\fpass,\ffail$, server-side quantum register $\bq$ holding a single qubit.\par 
			The protocol is the amplification procedure (Protocol \ref{prtl:3r}) running on Protocol \ref{prtl:11}. The condition $\epsilon_0<\epsilon$, $\lambda<\frac{1}{6}\delta_0(\epsilon-\epsilon_0)$ in Protocol \ref{prtl:3r} is satisfied by taking the $\epsilon_0$ in Protocol \ref{prtl:11} sufficiently small and tuning $\delta$ in Theorem \ref{thm:5.13r} to be suitably small.
		\end{prtl}
	\end{mdframed}
	The completeness, efficiency and soundness are from the properties of Protocol \ref{prtl:11} and the preRSPV-to-RSPV amplification procedure (Section \ref{sec:3.5.2}).
\section{Classical Verification of Quantum Computations from Cryptographic Group Actions}\label{sec:6}
In this section we combine our results with existing results to get new results on CVQC.
\subsection{RSPV and CVQC from weak NTCF}
Now we are ready to state the weak noisy trapdoor claw-free function (weak NTCF) assumption. We refer to Section \ref{sec:1.3.4} for an intuitive introduction. Below we give its formal definition. Note that there are also multiple styles for defining it; here we use the NTCF definition in \cite{cvqcinlt} and adapt it to weak NTCF:
\begin{defn}[Weak $\fNTCF$]\label{defn:ntcf}
	We define weak noisy trapdoor claw-free function (weak NTCF) as follows. It is parameterized by security parameter $\kappa$ and correctness error $\mu$ and is defined to be a class of polynomial time algorithms as below. $\fKg$ is a sampling algorithm. $\fDc$, $\fCHK$ are deterministic algorithms. $\fEv$ is allowed to be a sampling algorithm. $ \fpoly^\prime$ is a polynomial that determines the the range size. $$\fKg(1^{1/\mu}, 1^\kappa)\rightarrow (\sk,\pk),$$ $$\fEv_\pk: \{0,1\}\times \{0,1\}^{\kappa}\rightarrow \{0,1\}^{\fpoly^\prime(\kappa)},$$ $$\fDc_\sk: \{0,1\}\times \{0,1\}^{\fpoly^\prime(\kappa)}\rightarrow \{0,1\}^{\kappa}\cup \{\bot\},$$ $$\fCHK_{\pk}: \{0,1\}\times \{0,1\}^{\kappa}\times \{0,1\}^{\fpoly^\prime(\kappa)}\rightarrow \{\ftrue ,\ffalse\}$$ And they satisfy the following properties:\par
	\begin{itemize}
	\item (Correctness) 
	\begin{itemize}
	\item (Noisy 2-to-1) For all possible $(\sk,\pk)$ in the range of $\fKg(1^{1/\mu}, 1^\kappa)$ there exists a sub-normalized probability distribution $(p_y)_{y\in \{0,1\}^{\fpoly^\prime(\kappa)}}$ that satisfies: for any $y$ such that $p_y\neq 0$, $\forall b\in \{0,1\}$, there is $\fDc_\sk(b,y)\neq \bot$, and
	\begin{equation}\label{eq:64co}\fEv_\pk(\ket{+}^{\otimes \kappa})\approx_{\mu}\sum_{y:p_y\neq 0}\frac{1}{\sqrt{2}}(\ket{\fDc_\sk(0,y)}+\ket{\fDc_\sk(1,y)})\otimes \sqrt{p_y}\ket{y}\end{equation}
	\item (Correctness of $\fCHK$) For all possible $(\sk,\pk)$ in the range of $\fKg(1^\kappa)$, $\forall x\in \{0,1\}^{\kappa}, \forall b\in \{0,1\}$:
	$$\fCHK_\pk(b,x,y)=\ftrue\Leftrightarrow\fDc_{\sk}(b,y)=x$$
	\end{itemize}
	\item (Claw-free) For any BQP adversary $\fAdv$,
	\begin{equation}\Pr\left[\begin{aligned}&(\sk,\pk)\leftarrow \fKg(1^{1/\mu},1^\kappa),\\&\fAdv(\pk,1^{1/\mu},1^\kappa)\rightarrow (x_0,x_1,y):\quad x_0\neq \bot,x_1\neq \bot, x_0\neq x_1\\&\fDc_\sk(0,y)=x_0,\fDc_\sk(1,y)=x_1\end{aligned}\right]\leq \fneg(\kappa)\end{equation}
		\end{itemize}
	\end{defn}
	The ``noisy'' comes from the fact that $\fEv$ is allowed to be a sampling algorithm and the ``weak'' comes from the error term in \eqref{eq:64co}.\par
	We refer \cite{BCMVV,BKVV,cvqcinlt} for more information on NTCF techniques.\par
	 The following theorem could be proved based on the results of \cite{BGKPV23}:
	\begin{thm}\label{thm:6.1}
		Assuming the existence of weak NTCF, there exists an RSPV protocol for BB84 states $\fBB(1^{1/\mu},1^{1/\epsilon},1^\kappa)$ that is $\mu$-complete and $\epsilon$-sound for any $\mu,\epsilon=1/\fpoly(\kappa)$.
	\end{thm}
	Note that the correctness error in weak NTCF leads to a completeness error $\mu$ in the protocol.\par
	\cite{BGKPV23} constructed a test of a qubit from NTCF without assuming the adaptive hardcore bit property, which could be adapted easily to weak NTCF. In Section \ref{sec:b2} we already give the proof of Theorem \ref{thm:6.1} by tranlating test-of-a-qubit to RSPV for BB84.\par
	Combining it with our results in Section \ref{sec:5}, we could prove that all the protocols in Section \ref{sec:5} could be constructed from weak NTCF. Especially, we have:
	\begin{thm}\label{thm:1.2pn}
		Assuming the existence of weak NTCF, there exists an RSPV for state family $\{\frac{1}{\sqrt{2}}(\ket{0}\ket{x_0}+\ket{1}\ket{x_1}),x_0,x_1\in \{0,1\}^n\}$ that is $\mu$-complete and $\epsilon$-sound for any $\mu,\epsilon=1/\fpoly(\kappa)$..
	\end{thm} 
	\begin{thm}\label{thm:1.2n}
		Assuming the existence of weak NTCF, there exists an RSPV for state family $\{\ket{+_\theta}:=\frac{1}{\sqrt{2}}(\ket{0}+e^{\mi\pi\theta/4}\ket{1}),\theta\in \{0,1,2\cdots 7\}\}$ that is $\mu$-complete and $\epsilon$-sound for any $\mu,\epsilon=1/\fpoly(\kappa)$..
	\end{thm}
	Then recall that by \cite{FKD}, if the client is allowed to prepare and send lots of $\ket{+_{\theta}}$ states before the protocol, both parties could do quantum computation verification. As discussed before, this quantum communication could be compiled to classical communication using an RSPV for $\ket{+_\theta}$ states. Finally we note that although RSPV-based compilation introduces some non-negligible completeness error and soundness error, in CVQC problem these errors could be amplified to be exponentially small by sequential repetition as long as there is a significant completeness-soundness gap. Thus we have:
	\begin{thm}\label{thm:6.4}
		Assuming the existence of weak NTCF, there exists a CVQC protocol.
	\end{thm}
\subsection{CVQC from Assumptions on Group Actions}
Finally we review the results in \cite{AMR22}, which constructs weak TCF (which is stricter than weak NTCF) from cryptographic group actions (for example, isogeny).
\begin{assump}[Repeat of \cite{AMR22}]\label{assump:ga}
Assume the extended linear hidden shift assumption holds for some effective group action.
\end{assump}
\begin{thm}[Repeat of \cite{AMR22}]\label{thm:6.5}
	There exists a weak TCF assuming Assumption \ref{assump:ga}
\end{thm}
We refer to \cite{AMR22} for details.\par
Combining Theorem \ref{thm:6.4}, \ref{thm:6.5} we have:
\begin{cor}
Assuming \ref{assump:ga}, there exists a CVQC protocol.
\end{cor}
\appendix
\section{Proof of Theorem \ref{thm:qfacbeforec}}\label{app:p1}
We note that test in Definition \ref{defn:a1} corresponds exactly to the case of $\varphi\in \{0,2\}, \theta\in \{1,3,5,7\}$ in $\fQFacTest_{\text{step 3}}$. In $\fQFacTest_{\text{step 3}}$ the cases where the score could be nontrivially accumulated could be decomposed to this test rotated to different angles (that is, $\varphi\in \{0,2\}$, $\varphi\in \{1,3\}$, $\varphi\in \{4,6\}$, $\varphi\in \{5,7\}$). In trivial cases (that is, $|\theta-\varphi|\in \{0,2,4\}$) the value of $\bscore$ is simply $\fwin$ with probability close to $1$ (under the condition that the server passes with probability close to $1$). Thus we have $\OPT=\frac{1}{2}+\frac{1}{2}\cos^2(\pi/8)$ as in Set-up \ref{setup:12}. The soundness property could be proved in the way described in the beginning of Section \ref{sec:r4}, proof of Lemma \ref{lem:b3} and Section \ref{sec:a.3} (here we do not need to consider computational indistinguishability so lemmas based on statistical indistinguishability, like Lemma \ref{lem:a.5}, \ref{lem:a.8}, are sufficient).

\section{CVQC from Cryptographic Group Actions via \cite{morimae20}}\label{app:p2}
In this section we sketch (without formal proofs) another approach for achieving CVQC from weak NTCF, based on the results in \cite{morimae20}.\footnote{We thank anonymous reviewers for pointing out this approach.}\par
\cite{morimae20} gives a result on how to achieve verification of quantum computation in a model where trusted center sends BB84 states to the server and sends the classical description to the client. By replacing the state distribution step by callings to the RSPV protocols, we get a CVQC protocol. Note that it seems that RSPV for random BB84 states is not sufficient here; what we need is an RSPV for BB84 states where the client could choose which state to prepare, as discussed in Section \ref{sec:3.3.1}. Intuitively this primitive could be constructed by repeating the RSPV-for-BB84 for many times and letting the client to choose the desired state. Finally, as discussed in Section \ref{sec:r4} and Theorem \ref{thm:6.1}, RSPV for BB84 states could be constructed from cryptographic group actions, which completes the proof.
\bibliographystyle{plain}
\bibliography{bib}
	\end{document}

%% file: macros.tex
\usepackage[utf8]{inputenc}
\usepackage[T1]{fontenc}
\usepackage{microtype}
\usepackage{braket}
\usepackage{qcircuit}
\usepackage{amsthm}
\usepackage{amsmath,amssymb}
\usepackage{mathrsfs}
\usepackage{eufrak}
\usepackage{soul}
\usepackage{array}
\usepackage{makecell}
\usepackage{mdframed}
\usepackage{stmaryrd}
\usepackage{hyperref}
\usepackage[bottom]{footmisc}
\usepackage{scrextend}
\usepackage{tikz}
\usepackage{tikz-cd}
\usetikzlibrary{calc,arrows,decorations.pathreplacing,decorations.pathmorphing,intersections}
\usetikzlibrary{positioning}
\usetikzlibrary{fit}
\theoremstyle{plain}
\newtheorem{thm}{Theorem}[section]
\newtheorem{lem}[thm]{Lemma}

\newtheorem{prop}[thm]{Proposition}

\newtheorem{cor}[thm]{Corollary}
\newtheorem{assump}{Assumption}

\newtheorem{fact}{Fact}

\theoremstyle{definition}
\newtheorem{prtl}{Protocol}
\newtheorem{conv}{Convention}

\newtheorem{setup}{Set-up}
\newtheorem{defn}{Definition}[section]
\newtheorem{exmp}{Example}[section]
\newtheorem{nota}{Notation}[section]
\theoremstyle{remark}

\DeclareMathOperator{\tr}{tr}

\newcommand{\cE}{{\mathcal{E}}}
\newcommand{\cF}{{\mathcal{F}}}

\newcommand{\cH}{{\mathcal{H}}}

\newcommand{\mi}{{\mathrm{i}}}
\newcommand{\btheta}{{\boldsymbol{\theta}}}

\newcommand{\bS}{{\boldsymbol{S}}}
\newcommand{\bD}{{\boldsymbol{D}}}

\newcommand{\bx}{{\boldsymbol{x}}}
\newcommand{\bK}{{\boldsymbol{K}}}

\newcommand{\bR}{{\boldsymbol{R}}}

\newcommand{\bflag}{{\boldsymbol{flag}}}

\newcommand{\bscore}{{\boldsymbol{score}}}

\newcommand{\foprtortest}{{\mathsf{operatortest}}}
\newcommand{\fenergytest}{{\mathsf{energytest}}}

\newcommand{\bbC}{{\boldsymbol{C}}}
\DeclareMathOperator{\bE}{\mathbb{E}}

\DeclareMathOperator{\bN}{\mathbb{N}}
\DeclareMathOperator{\bZ}{\mathbb{Z}}
\newcommand{\bbI}{\mathbb{I}}

\newcommand{\fX}{{\sf X}}

\newcommand{\fZ}{{\sf Z}}

\newcommand{\fDisgard}{{\sf Disgard}}

\newcommand{\Domain}{{\text{Domain}}}

\newcommand{\fQFac}{{\mathsf{QFac}}}

\newcommand{\fDc}{{\mathsf{Dec}}}
\newcommand{\fEv}{{\mathsf{Eval}}}
\newcommand{\fKg}{{\mathsf{KeyGen}}}

\newcommand{\fneg}{{\mathsf{negl}}}
\newcommand{\fpoly}{{\mathsf{poly}}}

\newcommand{\fAdv}{{\mathsf{Adv}}}

\mdfdefinestyle{figstyle}{ 
  linecolor=black!7, 
  backgroundcolor=black!7
}

\mdfdefinestyle{figstyle}{ 
  linecolor=black!7, 
  backgroundcolor=black!7
}
\newcommand{\fSim}{{\mathsf{Sim}}}

\newcommand{\fNTCF}{{\mathsf{NTCF}}}
\newcommand{\fCHK}{{\mathsf{CHK}}}
\newcommand{\ftrue}{{\mathsf{true}}}
\newcommand{\ffalse}{{\mathsf{false}}}

\newcommand{\fpass}{{\mathsf{pass}}}
\newcommand{\ffail}{{\mathsf{fail}}}

\newcommand{\ftest}{{\mathsf{test}}}

\newcommand{\fcomp}{{\mathsf{comp}}}
\newcommand{\fwin}{{\mathsf{win}}}
\newcommand{\flose}{{\mathsf{lose}}}

\newcommand{\tind}{{\text{ind}}}

\newcommand{\sk}{{\text{sk}}}
\newcommand{\OPT}{{\text{OPT}}}

\newcommand{\pk}{{\text{pk}}}

\newcommand{\roundtype}{{\text{roundtype}}}

\newcommand{\val}{{\text{val}}}

\newcommand{\fOneBlock}{{\mathsf{OneBlock}}}
\newcommand{\fOneBlockTensor}{{\mathsf{OneBlockTensor}}}

\newcommand{\fMultiBlockTest}{{\mathsf{MultiBlockTest}}}

\newcommand{\fMultiBlockComp}{{\mathsf{MultiBlockComp}}}
\newcommand{\fMultiBlock}{{\mathsf{MultiBlock}}}

\newcommand{\fKP}{{\mathsf{KP}}}
\newcommand{\fQFacTest}{{\mathsf{QFacTest}}}
\newcommand{\fQFacComp}{{\mathsf{QFacComp}}}
\newcommand{\tdecodekey}{{\text{decodekey}}}

\newcommand{\fBB}{{\mathsf{BB84}}}
\newcommand{\fxorparity}{{\mathsf{xorparity}}}
\newcommand{\fsendxorparity}{{\mathsf{sendxorparity}}}
\newcommand{\fsend}{{\mathsf{send}}}

\newcommand{\bxor}{{\boldsymbol{xor}}}
\newcommand{\bq}{{\boldsymbol{q}}}
\newcommand{\bsimxor}{{\boldsymbol{simxor}}}

\newcommand{\bP}{{\boldsymbol{P}}}

\newcommand{\bbE}{{\boldsymbol{E}}}
\newcommand{\bQ}{{\boldsymbol{Q}}}

\newcommand{\bw}{{\boldsymbol{w}}}
\newcommand{\bregs}{{\boldsymbol{regs}}}
\newcommand{\breg}{{\boldsymbol{reg}}}
\newcommand{\tD}{{\text{D}}}
\newcommand{\tTD}{{\text{TD}}}

\newcommand{\tPos}{{\text{Pos}}}
\newcommand{\ttest}{{\text{test}}}
\newcommand{\tcomp}{{\text{comp}}}
\newcommand{\tmode}{{\text{mode}}}
\newcommand{\tvals}{{\text{vals}}}
\newcommand{\tReal}{{\text{Real}}}
\newcommand{\tIdeal}{{\text{Ideal}}}
\newcommand{\tHW}{{\text{HW}}}
\newcommand{\tParity}{{\text{Parity}}}
\newcommand{\tsize}{{\text{size}}}
\newcommand{\tindex}{{\text{index}}}
\newcommand{\tin}{{\text{in}}}
\newcommand{\tout}{{\text{out}}}
\newcommand{\ttemp}{{\text{temp}}}
\newcommand{\tmr}{{\text{mr}}}
\newcommand{\tutob}{{\text{u2b}}}
\newcommand{\tsubs}{{\text{subs}}}
\newcommand{\tkey}{{\text{key}}}